\documentclass[letterpaper,12pt]{article} 
\usepackage[top=1in, bottom=1in, left=1in, right=1in]{geometry}
\usepackage{setspace}
\usepackage{natbib}

\usepackage{enumerate}
\PassOptionsToPackage{hyphens}{url}
\usepackage[breaklinks]{hyperref}
\usepackage{url}
\usepackage{amsthm}
\usepackage{tikz}
\usetikzlibrary{arrows}
\usetikzlibrary{calc}
\usepackage[font=footnotesize]{caption} 

\usepackage{algorithm,algorithmic}

\usepackage{lindsten,abbreviations}
\usepackage{bm}

\graphicspath{{./}{./figs/}}

\newcommand{\C}{\ensuremath{\mathcal{C}}}

\renewcommand\mid{\,|\,}

\newcommand{\z}{\ensuremath{\mathbf{x}}}
\renewcommand{\setZ}{\setX}

\newcommand{\x}{\ensuremath{\mathbf{x}}} 
\newcommand{\ix}{\widetilde\x} 
\newcommand{\w}{\ensuremath{\mathbf{w}}} 
\newcommand{\rx}{\ensuremath{\check{\mathbf{x}}}} 
\newcommand{\isetX}{\widetilde\setX} 
\newcommand{\setC}{\ensuremath{\mathsf{C}}}
\renewcommand{\aa}{\ensuremath{\mathbf{a}}} 
\newcommand{\nw}{\ensuremath{\bm{W}}} 

\newcommand{\piN}{\widehat \pi^N} 
\newcommand\bwdkernel[1]{L_{#1}} 

\def\P{{\mathbb P}}

\def\pcv{\stackrel{\scriptscriptstyle \P}{\longrightarrow}}
\newcommand\rmd{\mathrm{d}}
\newcommand\ud{\rmd}
\def\1{{\mathbf 1}}

\newcommand\Nbound{\underline{N\!}\,}

\newcommand\dcsmc[1][\smc]{{D\&C-#1}\xspace} 
\newcommand\sir{SIR\@\xspace} 
\newcommand\RMC{\ensuremath{\mathsf{dc\_smc}}}
\newcommand\stdsir{\ensuremath{\mathsf{sir}}}

\newcommand\jj{j} 
\newcommand\nn{n} 

\newtheorem{prop}{Proposition}
\newtheorem{lemma}{Lemma}

\newtheoremstyle{smallremark}
  {0.5\topsep}   
  {0.5\topsep}   
  {\normalfont\small}  
  {0pt}       
  {\itshape} 
  {.}         
  {5pt plus 1pt minus 1pt} 
  {}          
\theoremstyle{remark}
\newtheorem*{example*}{Example}
\newtheorem{remark}{Remark}

\newcommand\sigmaalg{{\mathcal F}}
\newcommand\sigmagen{{\mathcal A}}
\newcommand\compl[1]{{#1}^{\textrm{c}}}

\newcommand{\mis}[1]{\langle D \rangle_2}
\newcommand{\miss}[1]{\mis{#1}^\prime}
\newcommand{\mi}[2]{(#1,#2)}

\newcommand\tleft{{t_1}}
\newcommand\tright{{t_2}}

\begin{document}
\title{\thetitle}
\date{30 June 2015}

\author{F. Lindsten{$\vphantom{u}^{1,3}$},
A. M. Johansen{$\vphantom{n}^2$},
C. A. Naesseth{$\vphantom{n}^3$},
B. Kirkpatrick{$\vphantom{n}^4$},\\
T. B. Sch{\"o}n{$\vphantom{n}^5$},
J. A. D. Aston{$\vphantom{n}^1$},
and A. Bouchard-C\^ot\'e{$\vphantom{e}^6$}
\let\thefootnote\relax\footnote{\noindent
    \scriptsize {1} University of Cambridge, {2} University of Warwick,
    {3} Link{\"o}ping University, {4} University of Miami, {5} Uppsala University,
    {6} University of British Columbia. Address for Correspondence: Fredrik
    Lindsten, Signal Processing Laboratory, CUED, Trumpington Street,
    Cambridge, CB2 3PU. Email: \url{fredrik.lindsten@eng.cam.ac.uk}.
}}

\maketitle


\begin{abstract}
  We propose a novel class of Sequential Monte Carlo (\smc) algorithms, appropriate for inference in
  probabilistic graphical models. This class of algorithms adopts a divide-and-conquer approach
  based upon an auxiliary tree-structured decomposition of the model of interest, turning the
  overall inferential task into a collection of recursively solved sub-problems.
  The proposed method is applicable to a broad class of probabilistic graphical models, \emph{including} models
  with loops. Unlike a standard
  \smc sampler, the proposed 
  \mbox{Divide-and-Conquer SMC} employs multiple independent
  populations of weighted particles, which are resampled, merged, and propagated as the method
  progresses. We illustrate empirically that this approach can outperform standard methods in terms
  of the accuracy of the posterior expectation and marginal likelihood
  approximations. Divide-and-Conquer SMC also opens up novel parallel implementation options and the
  possibility of concentrating the computational effort on the most challenging sub-problems.
  We demonstrate its performance on a Markov random field and on a hierarchical
  logistic regression problem.
  \vspace{1ex} \\
  \textbf{Keywords:} Bayesian methods, Graphical models, Hierarchical models, Particle filters
\end{abstract}

\section{Introduction}

Sequential Monte Carlo (\smc) methods are a popular class of algorithms for
approximating some sequence of probability distributions of interest,
$(\pi_t(\x_t) : t = \range{1}{n} )$. This is done by simulating, for each
distribution in the sequence, a collection of $\Np$ particles $\{\x_t^i\}_{i=1}^\Np$
with corresponding nonnegative importance weights $\{\w_t^i \}_{i=1}^\Np$, such
that the weighted empirical distribution $\piN_t(\rmd\x_t) \eqdef (\sum_j
\w_t^j)^{-1}\sum_i \w_t^i \delta_{\x_t^i}(\rmd\x_t)$ approximates $\pi_t$. The weighted
particles are generated sequentially, in the sense that the particles
generated at iteration $t$ depends on the particles generated up to iteration
$t-1$. 
 
The most well-known application of \smc is to the filtering problem in
general state-space hidden Markov models, see \eg, \citet{DoucetJ:2011} and
references therein. However, these methods are  much more
generally applicable and there has been much recent interest in using \smc for
sampling from probability distributions that do not arise from chain-shaped
probabilistic graphical models (\pgm{s}).  This typically involves using \smc
to target a sequence of auxiliary distributions which are constructed to admit
the original distribution as an appropriate marginal \citep{DelMoralDJ:2006}.
Examples include likelihood tempering \citep{DelMoralDJ:2006}, data tempering
\citep{Chopin:2002}, and sequential model decompositions
\citep{Bouchard-CoteSJ:2012,NaessethLS:2014}, to mention a few.  

For many statistical models of interest, however, a \emph{sequential} decomposition
might not be the most natural, nor computationally efficient, way of approaching the inference problem. 
In this contribution we propose an extension of the classical \smc framework, Divide-and-Conquer \smc (\dcsmc),
which we believe will further widen the scope of \smc samplers and
provide efficient computational tools for Bayesian inference within a broad
class of probabilistic models.

The idea underlying \dcsmc is that an approximation can be made to any
multivariate distribution by splitting the
collection of model variables into \emph{disjoint} sets and defining, for each
of these sets, a suitable auxiliary target distribution. Sampling from these
distributions is typically easier than sampling from the original distribution
and can be done 
in parallel, whereafter the results are merged to provide a solution to the
original problem of interest (correcting for the discrepancy between the
approximating and exact distributions by importance sampling techniques).   
Using the divide-and-conquer methodology, we recurse and repeat this
procedure for each of the components. This corresponds to breaking the overall inferential task into a collection of
simpler problems.
At any intermediate iteration of the \dcsmc algorithm we maintain multiple independent
sets of weighted particles, which are subsequently merged and propagated as the algorithm progresses,
 using rules similar to those employed in standard \smc.
The proposed method inherits some of the theoretical guarantees of standard
\smc methods. In particular, our simulation scheme can be used to provide
\emph{exact  approximations} of costly or intractable MCMC algorithms, via the
particle MCMC methodology \citep{AndrieuDH:2010}. 

Furthermore, we introduce a method for constructing the aforementioned decompositions
for a broad class of \pgm{s} of interest,
 which we call \emph{self-similar graphical models}. To construct auxiliary distributions, we remove edges and nodes in a \pgm of interest, creating smaller connected components as sub-graphical models. These sub-graphical models are then recursively decomposed as well. Note that this decomposition does not assume that the \pgm of interest is tree-shaped. 
Indeed, we demonstrate that the proposed methodology is effective not only when
the model has an obvious hierarchical structure
(for example, Figure~\ref{fig:method:hierarchical_model}), but also in cases where the
hierarchical decomposition is artificial (Figure~\ref{fig:ising}).
In either case, one iteratively exploits solutions to easier sub-problems
as a first step in the solution of a more complex problem.

We conclude this section with a summary of the structure of the remainder of
the paper. Section~\ref{sec:back} provides details on the background of the work
presented here: algorithms upon which it builds and those to which it is
related.
The basic \dcsmc and decomposition methodology is  presented in Section~\ref{sec:method},
including a discussion of its theoretical properties. 
A number of methodological extensions are presented in Section~\ref{sec:extensions} and
two realistic applications are presented in Section~\ref{sec:experiments}. The
paper concludes with a discussion.


\section{Background and problem formulation}\label{sec:back}

\subsection{Problem formulation}\label{sec:problem}

We let $\pi$ denote a probability distribution of interest, termed the
\emph{target distribution}. With a slight abuse of notation, we also denote
its density by $\pi(\x)$, $\x \in \setX$ (with respect to an anonymous reference
measure). The set $\setX$ is called the \emph{state space}, and could be
discrete, continuous or mixed (we assume throughout the paper that all spaces
are Polish and equipped with Borel $\sigma$-algebras). We assume that the density
$\pi$ can be written as $\pi(\x) = \gamma(\x) / Z$, where the unnormalized density
$\gamma(\x)$ can be computed point-wise, whereas evaluating the \emph{normalization
  constant} $Z = \int \gamma(\x) \ud \x$ may be computationally challenging.
The two problems with we are concerned are \emph{(1)} approximating the
normalization constant $Z$, and \emph{(2)}, computing integrals under $\pi$ of
some \emph{test function} $f:\setX \to \mathbb{R}$, $\int f(\x) \pi(\x) \ud \x$,
where $f(\x)$ can be computed point-wise. In a Bayesian context, \emph{(1)}
corresponds to approximating the marginal likelihood of the observed data, and
\emph{(2)}, computing the posterior expectation of some function, $f$, of the
parameters and latent variables, $\x$.

\subsection{Probabilistic graphical models}\label{sec:factor-graphs}

Problems \emph{(1)} and \emph{(2)} often arise in the context of \pgm{s}, a
formalism to encode dependencies between random variables in a probabilistic
model. Two sorts of graphical structures are commonly used by statisticians to
describe model dependencies: the Bayesian Network \citep{Pearl:1985}, which summarizes the
conditional independence structure of a Bayesian model using a directed
acyclic graph, and undirected graphs, which are often used to describe models
specified via the full conditional distribution of each node such as Markov
random fields (see below) and many spatial models such as conditional
autoregressions \citep{Besag:1974}. Here, we focus on the abstract factor graph
formalism, and remind the reader that the two formalisms mentioned above can
be easily converted to factor graphs; see, \eg, \citet{Bishop2006ML} for
details.

Two assumptions are required in order to write a model as a factor
graph. First, that the state space, $\setX$, takes the form of a product
space, $\setX = \setX_n = \isetX_1 \times \isetX_2 \times \dots \times \isetX_n$. It is convenient
to define the set of \emph{variables}, $V$, corresponding to the elements of
this factorization, $1, 2, \dots, n$. Second, that the unnormalized density $\gamma$
can be decomposed as, $\gamma(\x_n=(\ix_1,\ldots,\ix_n)) = \prod_{\phi \in F} \phi(S_\phi(\ix_1,\dots,\ix_n))$, where
$F$ is a set of \emph{factors} and the function $S_\phi$ returns a sub-vector of
$(\ix_1,\ldots\ix_n)$ containing those elements upon which factor $\phi$ depends.

Under these assumptions, a factor graph can be defined as a bipartite graph,
where the set of vertices is given by $F \cup V$, and where we place an edge
between a variable $v \in V$ and a factor $\phi \in F$ whenever the function $\phi$
depends on $\isetX_v$, i.e. when $\x_v$ is included in the vector returned by
$S_\phi(\x_1,\ldots,\x_n)$.
%
Throughout the paper, we use the convention that a variable with a tilde denotes a variable
taking values in a single dimension ($\ix_n \in \isetX_n$), while variables without tilde are
elements of a product space ($\x_n = (\ix_1,\ldots,\ix_n) \in \setX_n = \isetX_1 \times \ldots \isetX_n$).

\subsection{Sequential Monte Carlo}\label{sec:back-smc}
Sequential Monte Carlo (\smc) methods are a class of sampling algorithms able to address problems \emph{(1)} and \emph{(2)} defined in Section~\ref{sec:problem}.
More precisely, \smc can be used to simulate from a sequence of probability
distributions defined on a sequence of spaces of increasing dimension.
Let $\pi_t(\x_t)$, with $\x_t \eqdef \prange{\ix_1}{\ix_t}$, be a \pdf defined on the product space
\begin{align}
  \label{eq:bkg:smc:product-space}
  \setX_t = \isetX_1 \times \isetX_2 \times \dots \times \isetX_t.
\end{align}
Furthermore, as above, assume that $\pi_t(\x_t) = \gamma_t(\x_t)/Z_t$ where $\gamma_t$ can be evaluated
point-wise, but where the normalizing constant $Z_t$ is computationally intractable.
\smc provides a way to sequentially approximate the sequence of distributions
$\pi_1,\,\pi_2,\,\dots,\,\pi_n$. As a byproduct, it also provides \emph{unbiased} estimates
of the normalizing constants $Z_1,\,Z_2,\,\dots,\,Z_n$  \citep[Prop.~7.4.1]{DelMoral:2004}.

The SMC approximation of $\pi_t$ at iteration $t$ ($1 \leq t \le n$) takes the form of a \emph{particle population}. This population consists in a collection of $N$ pairs of \emph{particles} and \emph{weights}: $\{\x_t^i, \w_t^i\}_{i=1}^\Np$, where $\x_t^i \in \setX_t$ and $\w_t^i \ge 0$.
The particle population provides an approximation of $\pi_t$, in the (weak)
sense that expectations of a (sufficiently regular) \emph{test} function, $f$,
with respect to the discrete probability distribution obtained after normalizing the weights,
\begin{align}
  \label{eq:particle-approx}
  \piN_t(\cdot) \eqdef \frac{1}{\sum_{j=1}^N \w_t^j} \sum_{i=1}^N \w_t^i \delta_{\x_t^i}(\cdot),
\end{align}
approximate the expectation of that test function under $\pi_t$:
\[
\int \pi_t(\x_t) f(\x_t) \ud \x_t \approx (\sum_{j=1}^N \w_t^j)^{-1} \sum_{i=1}^N \w_t^i
f(\x_t^i).
\]
One can consider test functions of direct interest (as well as
considering the weak convergence of the approximating distributions which can
be established under various conditions) for
example, one would use $f(\x) = \x$ to approximate a mean, and $f(\x) =
\1_A(\x)$ to approximate the probability that $\x \in A$.

We review here the simplest type of SMC algorithm, Sequential Importance Resampling (\sir),
and refer the reader to \citet{DoucetJ:2011} for a more in-depth exposition.
Pseudo-code for the \sir method is given in Algorithm~\ref{alg:std-sir}.
We present the algorithm in a slightly non-standard recursive form because it
will be convenient to present the proposed \dcsmc algorithm recursively, and
presenting \sir in this way makes it easier to compare the two algorithms. Furthermore, since the focus of this paper
is ``static'' problems (\ie, we are not interested in online inference, such as filtering),
the sequential nature of the procedure need not be emphasized.
For ease of notation, we allow the procedure to be called for $t=0$, which returns an ``empty'' set of particles
and, by convention, $\gamma_0(\emptyset) = 1$.
(Hence, we do not need to treat the cases $t=1$ and $t>1$ separately in the algorithm.)
The main steps of the algorithm, \emph{resampling}, \emph{proposal sampling}, and \emph{weighting}, are detailed below.

\begin{algorithm}[h]\singlespace
  \caption{$\stdsir(t)$ }
  \label{alg:std-sir}
{\footnotesize
  \begin{enumerate}
 \item If $t = 0$, return $(\{\emptyset, 1\}_{i=1}^\Np, 1)$.
 \item $(\{ \x_{t-1}^i, \w_{t-1}^i \}_{i=1}^\Np, \widehat Z_{t-1}^\Np ) \gets \stdsir(t-1)$.
  \item \label{step:std-sir-resample} Resample $\{ \x_{t-1}^i, \w_{t-1}^i \}_{i=1}^\Np$
     to obtain the 
     unweighted particle population $\{ \rx_{t-1}^i, 1 \}_{i=1}^\Np$.
 \item \label{step:std-propagate} For particle $i = \range{1}{\Np}$:
   \begin{enumerate}
   \item Simulate $\ix_t^i \sim q_t(\cdot \mid \rx_{t-1}^i)$.
   \item Set $\x_t^i = (\rx_{t-1}^i, \ix_t^i)$.
   \item \label{step:std-sir-weighting} Compute
     $ {  \displaystyle    \w_t^i =
       \frac{ \gamma_t(\x_t^i) }{ \gamma_{t-1} (\rx_{t-1}^i)}
       \frac{1}{q_t(\ix_t^i \mid \rx_{t-1}^i)}   }$.
   \end{enumerate}
   \item \label{step:z-hat} Compute
$       \widehat Z_t^\Np = \left\{ \frac{1}{\Np} \sum_{i=1}^\Np \w_t^i  \right\} \widehat Z_{t-1}^\Np$.
 \item Return $( \{\x_t^i, \w_t^i\}_{i=1}^\Np, \widehat Z_t^\Np)$.
  \end{enumerate}
}
\end{algorithm}

Resampling (Line~\ref{step:std-sir-resample}), in its simplest form, consists
of sampling $N$ times from the previous population approximation $\piN_{t-1}$,
as defined in \eqref{eq:particle-approx}. This is equivalent to sampling the
number of copies to be made of each particle from a multinomial distribution with number of trials $N$ and probability vector $(\w_{t-1}^1, \dots, \w_{t-1}^N)/(\sum_{i=1}^N \w_{t-1}^i)$. Since resampling is done with replacement, a given particle can be resampled zero, one, or multiple times. Informally, the goal of the resampling step is to prune particles of low weights in order to focus computation on the promising parts of the state space. This is done in a way that preserves the asymptotic guarantees of importance sampling. After resampling, the weights are reset to $1/N$, since the weighting is instead encoded in the random multiplicities of the particles. Note that more sophisticated resampling methods are available, see, \eg, \citet{DoucCM:2005}.
Proposal sampling (Line~\ref{step:std-propagate}), is based on user-provided proposal densities $q_t(\ix_t \mid \x_{t-1})$.
For each particle $\rx^i_{t-1} \in \setX_{t-1}$ output from the resampling stage, we sample a successor state $\ix_t^i \sim q_t(\cdot \mid \rx_{t-1}^i)$.  The sampled successor is a single state $\ix_t^i \in \isetX_t$ which is appended to $\rx_{t-1}^i$, to form a sample for the $t$-th product space, $\x_t^i = (\rx_{t-1}^i, \ix_t^i) \in \setX_t$.
Finally, weighting (Line~\ref{step:std-sir-weighting}) is used to correct for
the discrepency between $\pi_{t-1}(\rx_{t-1}^i) q_t(\ix_t^i|\rx_{t-1}^i)$ and the
new target $\pi_t(\rx_{t-1}^i,\ix_{t}^i)$. Importantly, weighting can be performed
on the unnormalized target densities $\gamma_t$ and $\gamma_{t-1}$.
The algorithm returns a particle-based approximation $\piN_t$ of $\pi_t$, as in \eqref{eq:particle-approx}, as well as an unbiased estimate $\widehat Z_t^\Np$ of $Z_t$ (Line~\ref{step:z-hat}).
In practice, an important improvement to this basic algorithm is to perform resampling only when particle degeneracy is severe. This can be done by monitoring the effective sample size (ESS): $(\sum_{i=1}^N \w_t^i)^2/\sum_{i=1}^N (\w_t^i)^2$, and by resampling only when ESS is smaller than some threshold, say $N/2$ \citep{KongLW:1994}.
In Figure~\ref{fig:smc-flow} we illustrate the execution flow of the algorithm as arising from the recursive
function calls.

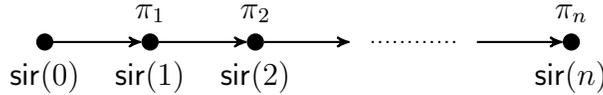
\begin{figure}[ptb]
  \centering
\tikzstyle{arrw} = [->,thick]
\tikzstyle{nd} = [draw,circle,thick,inner sep=0,minimum width=0.2cm,fill=black]
\def\voffset{0.8}
\def\voffsety{0.9}
\def\hoffset{1.4}
\begin{tikzpicture}[>=stealth',node distance=0.6cm]
  \node at (-\hoffset,0) (pi0) [nd,label=below:{$\stdsir(0)$}] {};
  \node at (0,0) (pi1) [nd,label={$\pi_1$},label=below:{$\stdsir(1)$}] {};
  \node at (\hoffset,0) (pi2) [nd,label={$\pi_2$},label=below:{$\stdsir(2)$}] {};
  \node at (2*\hoffset,0) (l) {};
  \node at (3*\hoffset,0) (r) {};
  \node at (4*\hoffset,0) (pi3) [nd,label={$\pi_n$},label=below:{$\stdsir(n)$}] {};
  \draw [arrw] (pi0)--(pi1);
  \draw [arrw] (pi1)--(pi2);
  \draw [arrw] (pi2)--(l);
  \draw [dotted, thick] (l)--(r);
  \draw [arrw] (r)--(pi3);
\end{tikzpicture}
  \caption{Computational flow of \sir (analogous for any \smc sampler). Each node corresponds
    to a call to $\stdsir$, the labels above show the corresponding target distribution,
    and the arrows illustrate the recursive dependencies of the algorithm. Note that this ``computational graph''
    of \smc is a chain, even if the sequence of target distributions does not correspond
    to a chain-structured \pgm.}
  \label{fig:smc-flow}
\end{figure}

The sequence of \emph{target distributions} $\{ \pi_t : t = \range{1}{n}\}$ can be constructed in
many different ways, which largely explains the generality and success of \smc.
The most basic construction, which is the classical application of \smc,
arises from chain-structured factor graphs (for example, state-space models or hidden Markov models).
For a chain-graph, the joint \pdf can be factorized as
$\pi(\x) = \frac{1}{Z} \prod_{t=1}^{n} \phi_t(\ix_{t-1}, \ix_{t})$, where $\x = \prange{\ix_1}{\ix_n}$;
see Section~\ref{sec:factor-graphs}.
(As above, to simplify the notation we have, without loss of generality, introduced a ``dummy variable'' $\ix_0 = \emptyset$.)
To simulate from the target distribution, the standard SIR algorithm employs a sequence of
\emph{intermediate distributions}: $\pi_t(\x_t) \propto \prod_{s=1}^{t} \phi_{s}(\ix_{s-1}, \ix_{s})$, where $\x_t = \prange{\ix_1}{\ix_t}$, $\ix_{s} \in \isetX_{s}$. Each $\pi_t$ can be written as $\gamma_t/Z_t$, where again $\gamma_t$ can be evaluated point-wise, but $Z_t$ is hard to compute. Importantly, we also have that $\pi_n = \pi$ by construction.
In fact, it is possible to make use of similar \emph{sequential decompositions} even when the original graph
is not a chain \citep{NaessethLS:2014}, as long as it is possible to find a sequence of auxiliary distributions
defined on increasing subsets of the model variables.

\subsection{\smc samplers and tempering}\label{sec:smc-samplers-review}

Another common approach is to make use of a sequence of auxiliary distributions for which we are interested
only in one of the marginals. Suppose that the densities of interest are defined over spaces which are not product spaces, $\widetilde\pi_t : \isetX_t \to [0, \infty)$.
For example, we may have $\widetilde\pi_t(\ix) \propto (\widetilde\pi(\ix))^{\alpha_t}$ as a tempered target distribution, with $\isetX_t = \isetX_{t-1} = \dots = \isetX_1$, and $q_t(\ix_t \mid \ix_{t-1})$ derived from a local MCMC move.
We can transform problems
of this type into a form suitable for \smc by using an auxiliary construction proposed by \citet{DelMoralDJ:2006},
which can be viewed as a (substantial) generalization of the annealed importance sampling method \citep{Neal:2001}.

The construction used by \citet{DelMoralDJ:2006} is to re-introduce a sequence of distributions defined
on product spaces $\setX_t = \isetX_1 \times \dots \times \isetX_t$ by defining,
\begin{align}
  \label{eq:auxiliary-pi}
  \pi_t(\x_t) = \widetilde\pi_t(\ix_t) \prod_{s=1}^{t-1} \bwdkernel{s}( \ix_s \mid \ix_{s+1}),
\end{align}
where $\x_t = \prange{\ix_1}{\ix_t} \in \setX_t$ as before.
In the above, $\bwdkernel{s}$ is a transition kernel from $\isetX_{s+1}$ to
$\isetX_{s}$---for instance an MCMC kernel---chosen by the user. For any choice of these
\emph{backward kernels}, $\pi_t$ admits $\widetilde\pi_t$ as a marginal by construction,
and it can thus be used as a proxy for the original target distribution $\widetilde\pi_t$.
Standard \smc algorithms can then be applied to the sequence of auxiliary distributions $\pi_t$, $t = \range{1}{n}$.
Using the structure of $\pi_t$ in \eqref{eq:auxiliary-pi}, the weight computation (Line~\ref{step:std-sir-weighting} of Algorithm~\ref{alg:std-sir}) is given by:
\begin{align}
  \label{eq:smc-samplers-weight}
  \w_t^i = \frac{\widetilde\gamma_t(\ix^i_t)}{\widetilde\gamma_{t-1}(\ix^i_{t-1})} \frac{\bwdkernel{t-1}(\ix^i_{t-1} \mid \ix^i_t)}{q_t(\ix^i_t \mid \ix^i_{t-1})},
\end{align}
where $\widetilde\gamma_t \propto \widetilde\pi_t$.
While the backward kernels $\bwdkernel{t}$ are formally arbitrary (subject to certain support restrictions),
they will critically influence the estimator variance.
If $q_t$ is a $\widetilde\pi_{t-1}$-reversible MCMC kernel, a typical choice is $L_{t-1} = q_t$ which results in a cancellation in
the weight expression \eqref{eq:smc-samplers-weight}: $\w_t^i = \widetilde\gamma_t(\ix_t) / \widetilde\gamma_{t-1}(\ix_{t})$.
See \citet{DelMoralDJ:2006} for further details and
guidance on the selection of the backward kernels.


\subsection{Related work}
Before presenting the new methodology in Section~\ref{sec:method}
we note that a number of related ideas have appeared in the
literature, although all have differed in key respects from the approach
described in the next section.

\citet{Koller1999ISBP} and \citet{BriersDS:2005,SudderthIIFW:2010}
address belief propagation using importance sampling and SMC, respectively, 
and these methods feature coalescence of particle systems, although they do not
provide samples targeting a distribution of interest in an \emph{exact}  
sense \citep{AndrieuDH:2010}. In contrast, the method proposed here yields
consistent estimates of the marginals and normalization constant, even when
approximating a graphical model with loops. Moreover, our method can handle
variables with constrained or discrete components, while much of the existing
literature relies on Gaussian approximations which may not be practical in
these cases.

Coalescence of particle systems in a
different sense is employed by \citet{JasraDSH:2008} who also use multiple
populations of particles; here the state space of the full parameter 
vector is partitioned, rather than the parameter vector itself. The
\emph{island particle model} of \citet{VergeDDM:2014} employs an ensemble of
particle systems which themselves interact according to the usual rules of SMC,
with every particle system targeting the same distribution over the full set
of variables. The \emph{local particle filtering} approach by \citet{RebeschiniH:2013}
attempts to address degeneracy (in a hidden Markov model context) via an
(inexact) localisation technique. Numerous authors have proposed custom SMC
algorithms for the purpose of inferring the structure of a latent tree, see \cite{TehDR:2008,Bouchard-CoteSJ:2012,LakshminarayananRT:2013}. These methods generally employ a single particle population. In contrast, 
our method assumes a known tree decomposition, and uses several particle populations.

\section{Methodology} \label{sec:method}
The proposed methodology is useful when the inference problem described in Section~\ref{sec:problem} can be
decomposed into a ``tree of auxiliary distributions'', as defined in Section~\ref{sec:tree-structured-auxiliary} below.
We present the basic \dcsmc method in Section~\ref{sec:dc-sir}, followed by fundamental convergence results
in Section~\ref{sec:theory}. Thereafter, we provide a concrete strategy for constructing the aforementioned tree-structured auxiliary distributions on which the \dcsmc algorithm operates. This strategy applies to many directed and undirected graphical modelling scenarios of practical interest (including models with cycles).  
It should be noted that, as with standard \smc algorithms,
a range of techniques are available to improve on the basic method presented in this section, and we discuss
several possible extensions in Section~\ref{sec:extensions}.

\subsection{Tree structured auxiliary distributions}\label{sec:tree-structured-auxiliary}

The proposed \dcsmc methodology generalizes the classical \smc framework from sequences (or chains) to trees.
As noted in Section~\ref{sec:back}, the \smc methodology is a general framework for simulating from essentially any
\emph{sequence} of distributions. Any such sequence can be organized on a chain, with subsequent distributions
being associated with neighbouring nodes on the chain; see Figure~\ref{fig:smc-flow}.
Note that the graph notion here
is used to describe the execution flow of the algorithm, and the sequence of distributions
organized on the chain does not necessarily correspond to a chain-structured \pgm.

In a similar way, \dcsmc operates on a \emph{tree of distributions}, which need not 
correspond to a tree-structured \pgm. 
Specifically, as in Section~\ref{sec:back-smc}, assume that we have a collection of (auxiliary) distributions, $\{\pi_t:t\in T\}$. However, instead of taking the index set $T$ to be nodes in a sequence, $T = \{1, 2, \dots, n\}$, we generalize $T$ to be nodes in a tree.
For all $t\in \T$, let $\C(t)\subset\T$ denote the
children of node $t$, with $\C(t) = \emptyset$
if $t$ is a leaf, and let $r\in\T$ denote the root of the tree.
We assume $\pi_t$ to have a density, also denoted by $\pi_t$, defined on a set $\setX_t$. We call such a collection a tree structured auxiliary distributions a \emph{tree decomposition} of
the target distribution $\pi$ (introduced in Section~\ref{sec:problem}) if it
has two properties. First, the root distribution is required to coincide with the target distribution, $\pi_r = \pi$.
The second is a consistency condition: we require that the spaces on which the node distributions are defined are
constructed recursively as
\begin{align}\label{eq:space-decomposition}
  \setX_t = \left( \otimes_{c \in \C(t)} \setX_c \right) \times \widetilde\setX_t
\end{align}
where the ``incremental'' set $\widetilde\setX_t$ can be chosen arbitrarily
(in particular, $\widetilde\setX_t = \emptyset$ for all $t$ in some proper subset of the nodes in $\T$ is a valid choice). Note that the second condition mirrors the
product space condition \eqref{eq:bkg:smc:product-space}. 
That is, the distributions $\{\pi_t : t\in T\}$ are defined on spaces of increasing
dimensions as we move towards the root from the leaves of the tree.

\begin{figure}[ptb]
  \centering
\newcommand\mydots{$\,\dots$}

\tikzstyle{arrw} = [->,thick]
\tikzstyle{nd} = [draw,circle,thick,inner sep=0,minimum width=0.2cm,fill=black]
\def\voffset{1}
\def\voffsety{0.9}
\def\hoffset{1.4}
\begin{tikzpicture}[>=stealth',node distance=0.6cm]
  \node at (0,0) (pi1) [nd,label=below:{$\pi_{c_1}$}] {};
  \node at (\hoffset,0) (pi2) [nd,label=below:{$\pi_{c_C}$}] {};
  \node at (2*\hoffset,0) (pi3) [nd] {};
  \node at (3*\hoffset,0) (pi4) [nd] {};
  \node at (\hoffset/2, \voffset/4) {\mydots} ;
  \node at (\hoffset*5/2, \voffset/4) {\mydots} ;
  \node at (\hoffset/2,\voffset) (pi5) [nd,label=above:{$\pi_t$}] {};
  \node at (\hoffset*5/2,\voffset) (pi6) [nd] {};
  \node at (\hoffset*3/2, \voffset) {\mydots} ;
  \node at (\hoffset*3/2,2*\voffset) (r) [nd,label=above:{$\pi_r$}] {};
  \draw [arrw] (pi1)--(pi5);
  \draw [arrw] (pi2)--(pi5);
  \draw [arrw] (pi3)--(pi6);
  \draw [arrw] (pi4)--(pi6);
  \draw[thick] (pi5)--($(pi5)!0.25!(r)$);
  \draw[thick,dotted] ($(pi5)!0.25!(r)$)--($(pi5)!0.6!(r)$);
  \draw[arrw] ($(pi5)!0.6!(r)$)--(r);
  \draw[thick] (pi6)--($(pi6)!0.25!(r)$);
  \draw[thick,dotted] ($(pi6)!0.25!(r)$)--($(pi6)!0.6!(r)$);
  \draw[arrw] ($(pi6)!0.6!(r)$)--(r);
\end{tikzpicture}
  \caption{Computational flow of \dcsmc. Each node corresponds to a target distribution $\{\pi_t : t\in T\}$ and,
    thus, to a call to \dcsmc (Algorithm~\ref{alg:dcsir}).
    The arrows illustrate the computational flow of the algorithm via its recursive dependencies.}
  \label{fig:dcsmc-flow}
\end{figure}
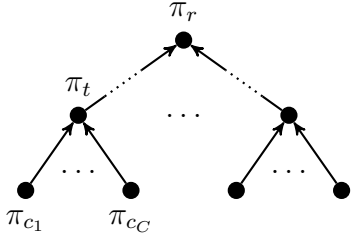

Figure~\ref{fig:dcsmc-flow} illustrates the execution flow of the \dcsmc algorithm
(which is detailed in the subsequent section), which performs inference for
the distributions $\{\pi_t : t \in T\}$ from leaves to root
in the tree. As pointed out above, the computational tree $T$ does not necessarily correspond
to a tree-structured \pgm. Nevertheless, when the \pgm of interest \emph{is} in fact a tree, the computational
flow of the algorithm can be easily related to the structure of the model
(just as the computational flow of standard \smc is easily understood when the
\pgm is a chain, although the \smc framework is in fact more general). Let us
therefore consider an example of how the target distributions $\{\pi_t : t \in T\}$
can be constructed in such a case, to provide some intuition for the proposed
inference strategy before getting into the details of the algorithm.

\begin{example*}[Hierarchical models]
Consider the simple tree-structured Bayesian network of
Figure~\ref{fig:method:hierarchical_model}~(rightmost panel), with three observations $y_{1:3}$, and five latent variables $\ix_{1:5}$.
The distribution of interest is the posterior $p(\ix_{1:5} \mid y_{1:3})$. To put this in the notation
introduced above, we define $\x_5 = \ix_{1:5}$ and $\pi(\x_5) = \pi_5(\x_5) = p(\ix_{1:5} \mid y_{1:3})$.
To obtain a tree decomposition of $\pi_5$ we can make use of the hierarchical structure of the \pgm.
By removing the root node $\ix_5$ we obtain two decoupled components (Figure~\ref{fig:method:hierarchical_model}, middle) for which we can define the auxiliary target distributions
$\pi_4(\x_4) = p(\x_4 \mid y_{1:2})$ and $\pi_3(\x_3) = p(\xi_3 \mid y_3)$, respectively, where $\x_4 = (\ix_1, \ix_2, \ix_4)$
and $\x_3 = \ix_3$. If the marginal priors for the root nodes in the decomposed models (here, $p(\ix_4)$ and $p(\ix_3)$)
are intractable to compute, we can instead define the auxiliary distribution $\pi_t(\x_t)$, $t=3,4$, using an arbitrary ``artificial prior''  $u_t(\ix_t)$
for its root (similar to the two-filter smoothing approach of \citet{BriersDM:2010}). This arbitrariness is ultimately corrected for by
importance weighting and does not impinge upon the validity of the proposed
inference algorithm (see Section~\ref{sec:theory}), although the choice of
$u_t$ can of course affect the computational efficiency of the algorithm.
Finally, by repeating this procedure, we can further decompose $\pi_4(\x_4)$ into two components, $\pi_1(\x_1)$ and $\pi_2(\x_2)$, as illustrated in Figure~\ref{fig:method:hierarchical_model}~(left). The target distributions
$\{\pi_t(\x_t) : t \in \crange{1}{5} \}$ can be organised on a tree (with the same graph topology as the \pgm under study,
excluding the observed variables) which satisfies the conditions for being a tree decomposition of
the sought posterior $p(\ix_{1:5} \mid y_{1:3})$.
\end{example*}

\begin{figure}[ptb]
  \centering
\tikzstyle{arrw} = [->,thick]
\tikzstyle{var} = [draw,circle,thick,inner sep=0,minimum width=0.55cm]
\tikzstyle{obs} = [draw,circle,thick,inner sep=0,minimum width=0.55cm,fill=black!15]
\def\voffset{0.8}
\def\voffsety{0.9}
\def\hoffset{1}
  \begin{tikzpicture}[>=stealth',node distance=0.6cm]
    \begin{scope}
      \node at (0,\voffsety) (a) [] {Level 2:};
      \node at (0,\voffsety+\voffset) (b) [] {Level 1:};
      \node at (0,\voffsety+2*\voffset) (c) [] {Level 0:};
    \end{scope}  
    \begin{scope}[shift={(2,0)}]
    \begin{scope}
      \node at (0,0) (y1) [obs] {$y_1$};
      \node at (\hoffset,0) (y2) [obs] {$y_2$};
      \node at (2*\hoffset,0) (y3) [obs] {$y_3$};
      \node at (0,\voffsety) (th21) [var] {$\widetilde\x_1$};
      \node at (\hoffset,\voffsety) (th22) [var] {$\widetilde\x_2$};
      \node at (2*\hoffset,\voffsety) (th23) [var] {$\widetilde\x_3$};
      \draw [arrw] (th21)--(y1);
      \draw [arrw] (th22)--(y2);
      \draw [arrw] (th23)--(y3);
    \end{scope}
    \begin{scope}[shift={(4,0)}]
      \node at (0,0) (y1) [obs] {$y_1$};
      \node at (\hoffset,0) (y2) [obs] {$y_2$};
      \node at (2*\hoffset,0) (y3) [obs] {$y_3$};
      \node at (0,\voffsety) (th21) [var] {$\widetilde\x_1$};
      \node at (\hoffset,\voffsety) (th22) [var] {$\widetilde\x_2$};
      \node at (2*\hoffset,\voffsety) (th23) [var] {$\widetilde\x_3$};
      \draw [arrw] (th21)--(y1);
      \draw [arrw] (th22)--(y2);
      \draw [arrw] (th23)--(y3);
      \node at (0.5*\hoffset,\voffsety+\voffset) (th11) [var] {$\widetilde\x_4$};
      \draw [arrw] (th11)--(th21);
      \draw [arrw] (th11)--(th22);
    \end{scope}
    \begin{scope}[shift={(8,0)}]
      \node at (0,0) (y1) [obs] {$y_1$};
      \node at (\hoffset,0) (y2) [obs] {$y_2$};
      \node at (2*\hoffset,0) (y3) [obs] {$y_3$};
      \node at (0,\voffsety) (th21) [var] {$\widetilde\x_1$};
      \node at (\hoffset,\voffsety) (th22) [var] {$\widetilde\x_2$};
      \node at (2*\hoffset,\voffsety) (th23) [var] {$\widetilde\x_3$};
      \draw [arrw] (th21)--(y1);
      \draw [arrw] (th22)--(y2);
      \draw [arrw] (th23)--(y3);
      \node at (0.5*\hoffset,\voffsety+\voffset) (th11) [var] {$\widetilde\x_4$};
      \draw [arrw] (th11)--(th21);
      \draw [arrw] (th11)--(th22);
      \node at (1.25*\hoffset,\voffsety+2*\voffset) (th01) [var] {$\widetilde\x_5$};
      \draw [arrw] (th01)--(th11);
      \draw [arrw] (th01)--(th23);
    \end{scope}
    \end{scope}
  \end{tikzpicture}
%
  \caption{{\footnotesize Decomposition of a hierarchical Bayesian model.}\vspace{-0.2in}}
  \label{fig:method:hierarchical_model}
\end{figure}
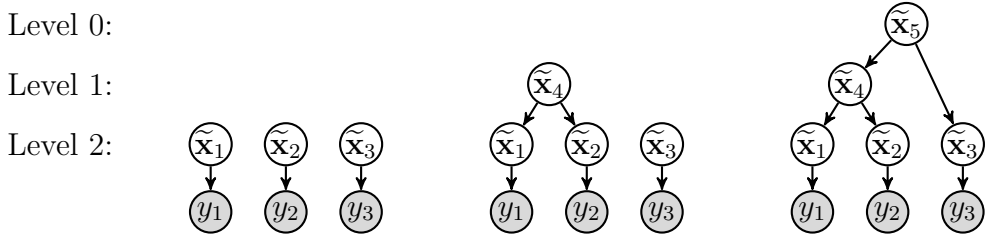

In Section~\ref{sec:decompositions} we formalise the decomposition strategy illustrated in the
example above, and also generalise it to a broader class of, so called, self-similar \pgm{s}.

\subsection{Divide-and-Conquer Sequential Importance Resampling}\label{sec:dc-sir}

We now turn to the description of the \dcsmc algorithm---a Monte Carlo procedure for approximating
the target distribution $\pi = \pi_r$ based on the auxiliary distributions ${\{\pi_t : t\in T\}}$.
For pedagogical purposes, we start by presenting the simplest possible implementation of the algorithm,
which can be thought of as the analogue to the \sir implementation of \smc.
Several possible extensions are discussed in Section~\ref{sec:extensions}.

As in standard SMC, \dcsmc approximates each $\pi_t$ by a collection of weighted samples, also referred to as a
particle population. 
Unlike a standard \smc sampler, however, the method maintains multiple \emph{independent} populations
of weighted particles, $(\{\x_t^i, \w_t^i\}_{i=1}^\Np : t\in\T_k)$, which are propagated  and merged as the algorithm progresses. Here $T_k \subset T$ is the set of indices of ``active'' target distributions at iteration $k$, $1 \leq k \leq \textrm{depth}(T)$.

The \dcsmc algorithm uses a bottom-up approach to simulate from the auxiliary target distributions
defined on the tree, by repeated resampling, proposal, and weighting steps,
which closely mirror standard SMC. We describe the algorithm by specifying the operations that are carried out at each node of the tree, leading to a recursive definition of the method.
For $t\in\T$, we define a procedure $\RMC(t)$ which returns, \emph{(1)} a weighed particle population $\{ \x_t^i, \w_t^i\}_{i=1}^\Np$
approximating $\pi_t$ as $\piN_t$ in Equation~\eqref{eq:particle-approx}, and \emph{(2)} an estimator $\widehat Z^\Np_t$ of the normalizing constant $Z_t$ (such that $\pi_t(\x_t) = \gamma_t(\x_t)/Z_t$). The procedure is given in Algorithm~\ref{alg:dcsir}.

\begin{algorithm}[ptb]\singlespace
  \caption{$\RMC(t)$}
  \label{alg:dcsir}
{\footnotesize 
  \begin{enumerate}
  \item For $c\in\C(t)$:
    \begin{enumerate}
    \item \label{step:recursive-call} $(\{\x_c^i, \w_c^i \}_{i=1}^\Np, \widehat Z_c^\Np ) \gets \RMC(c)$.
    \item \label{step:resample} Resample $\{\x_c^i, \w_c^i \}_{i=1}^\Np$
      to obtain the equally weighted particle system $\{ \rx_c^i, 1 \}_{i=1}^\Np$.
    \end{enumerate}
  \item \label{step:propagate} For particle $i = \range{1}{\Np}$:
    \begin{enumerate}
    \item If $\tilde \setX_t \neq \emptyset$, simulate $\widetilde\x_t^i \sim q_t(\cdot \mid \rx_{c_1}^i, \dots, \rx_{c_C}^i)$, where $(c_1, c_2, \dots, c_C) = \C(t)$; \\ else $\widetilde\x_t^i \gets \emptyset$.
    \item Set $\x_t^i = (\rx_{c_1}^i, \dots, \rx_{c_C}^i, \widetilde\x_t^i)$.
    \item \label{step:weighting} Compute
 $ {  \displaystyle    \w_t^i = 
        \frac{ \gamma_t(\x_t^i) }{ \prod_{c\in\C(t)} \gamma_{c} (\rx_c^i)}
        \frac{1}{q_t(\widetilde\x_t^i \mid \rx_{c_1}^i, \dots, \rx_{c_C}^i)}   }$.
    \end{enumerate}
    \item Compute
 $       \widehat Z_t^\Np = \left\{ \frac{1}{\Np} \sum_{i=1}^\Np \w_t^i  \right\} \prod_{c\in\C(t)} \widehat Z_c^\Np$.
  \item Return $( \{ \x_t^i, \w_t^i \}_{i=1}^\Np, \widehat Z_t^\Np)$.
  \end{enumerate}
}
\end{algorithm} 

The first step of the algorithm is to acquire, for each child node $c \in \C(t)$, a particle approximation
of $\pi_c$ 
by a recursive call (Line~\ref{step:recursive-call}).
Jointly, these particle populations provide an approximation of the product measure,
\begin{align}
  \label{eq:dcsir-product-measure}
  \kronecker_{c\in\C(t)} \pi_c(\rmd\x_c) \approx  \kronecker_{c\in\C(t)} \piN_c(\rmd\x_c).
\end{align}
Note that this point-mass approximation has support on $N^{C}$, $C = |\C(t)|$, points, although these
support points are implicitly given by the $NC$ unique particles (assuming no duplicates among the particles in the individual
child populations).

To obtain a computationally manageable approximation of the product measure, we generate $\Np$ samples
from the approximation in \eqref{eq:dcsir-product-measure}.
This is equivalent to performing standard multinomial resampling for each child particle population (Line~\ref{step:resample}), obtaining equally weighted samples $\{ \rx_c^i, 1 \}_{i=1}^\Np$ for each $c$, and for all
$i = \range{1}{\Np}$, combining all indices $i$ of the $c$ lists to create $N$ equally weighted tuples,
$\{(\rx_{c_1}^i, \dots, \rx_{c_{C}}^i), 1 )_{i=1}^\Np$.
This basic merging strategy can thus be implemented in $\Ordo(\Np)$ computational cost,
since there is no need to explicitly form the approximation of the product measure in \eqref{eq:dcsir-product-measure}.

The latter, resampling-based, description of how the child populations are merged
provide natural extensions to the methodology, \eg by using low-variance resampling schemes (\eg, \citet{CarpenterCF:1999})
and adaptive methods that monitor effective sample size to perform resampling only when particle degeneracy is severe \citep{KongLW:1994}.
\begin{remark}
 Note, however, that if we perform resampling amongst the child populations
 separately and then combine the resulting particles in this way, we require
$\Prb(\rx_c^i = \x_c^j) =  ( \sum_{l=1}^\Np \w_c^l)^{-1} \w_c^j$, $j=\range{1}{\Np}$,
for each $i = \range{1}{\Np}$, since the particles are combined based on their indices
(\ie, it is not enough that the marginal equality $\sum_{i=1}^\Np \Prb(\rx_c^i =
\x_c^j) = \Np (\sum_{l=1}^\Np \w_c^l)^{-1} \w_c^j$ holds).
Consequently, if the resampling mechanism that is employed results in an ordered list of resampled particles,
then a random permutation of the particles indices should be carried out before combining particles
from different child populations.
\end{remark}

Proposal sampling (Line~\ref{step:propagate}), similarly to standard SMC, is based on user-provided proposal densities $q_t$. However, the proposal has access to more information in \dcsmc, namely to the state of all the children $c_1, c_2, \dots, c_C$ of node $t$: $q_t(\cdot \mid \rx_{c_1}^i, \dots, \rx_{c_C}^i)$. For each particle tuple $(\rx_{c_1}^i, \dots, \rx_{c_{C}}^i)$ generated in the resampling stage, we sample a successor state $\widetilde\x_t^i \sim q_t(\cdot \mid \rx_{c_1}^i, \dots, \rx_{c_C}^i)$. Note that in some cases, parts of the tree structured decomposition do not require this proposal sampling step, namely when $\isetX_t = \emptyset$. We simply set $\widetilde\x_t^i$ to $\emptyset$ in these cases (the resampling and reweighting are still non-trivial).

Finally, we form the $i$-th sample at node $t$ of the tree by concatenating the tuple of resampled
child particles $(\rx_{c_1}^i, \dots, \rx_{c_{C}}^i)$ and the proposed state $\widetilde\x_t^i$ (if it is non-empty).
The importance weight is given by the ratio of the (unnormalised) target densities, divided by the proposal
density (Line~\ref{step:weighting}). We use the convention here that $\prod_{c\in\emptyset}(\cdot) = 1$, to take into account the base case of this recursion, at the leaves of the tree.

\begin{example*}[Hierarchical models, continued]
A simple choice for $q_t(\,\cdot \mid \rx_{c_1},\dots, \rx_{c_C}^i)$ in this example is to use $u_t$, the (artificial)
prior at the sub-tree rooted at node $\ix_t$. An alternative is to select $u_t$ as a conjugate prior to the distributions of the children, $p(\ix_c \mid \ix_t)$, $c\in\C(t)$, and to propose according to the posterior distribution of the conjugate pair. 
To illustrate the weight update, we show its simplified form in the simplest situation, where $q_t = u_t$:
\begin{align*}
  \w_t^i = 
  \frac{ \gamma_t(\x_t^i) }{ \prod_{c\in\C(t)} \gamma_{c} (\rx_c^i)}
  \frac{1}{q_t(\widetilde\x_t^i \mid \rx_{c_1}^i, \dots, \rx_{c_C}^i)}
  = \frac{ u_t(\x_t^i) \prod_{c\in\C(t)} p(\rx_c^i \mid \ix_{t}^i) }{\prod_{c\in\C(t)} u_{c}(\rx_{c}^i)}
  \frac{1}{u_t(\x_t^i)}
  = \prod_{c\in\C(t)} \frac{ p(\rx_c^i \mid \ix_{t}^i) }{ u_{c}(\rx_{c}^i)}.
\end{align*}
\end{example*}

If executed serially, the running time of \dcsmc is $O(N\cdot |T|)$. However, a running time of $O(N\cdot \textrm{depth}(T))$ can be achieved via parallelized or distributed computing (see Section~\ref{sec:distributed}).  In terms of memory requirements, they grow at the rate of $O(N\cdot \textrm{depth}(T)\cdot \max_t |\C(t)|)$ in the serial case,\footnote{The $\textrm{depth}(T)$ factor comes from the maximum size of the recursion stack, and $N \max_t |\C(t)|$ comes from the requirement for each level of the stack to store a particle population for each child.} and at the rate of $O(N\cdot |T|)$ in the parallel case.\footnote{In the extreme case where $|T|/2$ compute nodes are used, one for each leaf of a binary tree.}
Note that the \dcsmc algorithm generalizes the usual \smc framework; if $|\C(t)| = 1$ for all internal nodes, then the \dcsmc[\sir] procedure
described above reduces to a standard \sir method (Algorithm~\ref{alg:std-sir}).


\subsection{Theoretical Properties}\label{sec:theory}

As \dcsmc consists of standard \smc steps combined with merging of particle populations
via resampling, it is possible (with care) to extend many of the
results from the standard, and by now well-studied, \smc setting (see \eg, \cite{DelMoral:2004}
for a comprehensive collection of theoretical results). Here,
we present two results to justify Algorithm~\ref{alg:dcsir}.
The proofs of the two propositions stated below are given in Appendix~\ref{app:proofs}.

First of all, the unbiasedness of the normalizing constant estimate
of standard \smc algorithms \citep[Prop.~7.4.1]{DelMoral:2004} is inherited by \dcsmc. 

\begin{prop}\label{prop:bias}
Provided that 
$\gamma_t \ll \otimes_{c\in\mathcal{C}(t)} \gamma_c \otimes q_t$ for every $t \in T$
and an unbiased,
exchangeable resampling scheme is applied to every population at every
iteration, we have for any $N 
\geq 1$:
\[\mathbb{E}[\widehat{Z}_r] = Z_r = \int \gamma_r(d\x_r). \]
\end{prop}

An important consequence of Proposition~\ref{prop:bias} is that the \dcsmc algorithm
can be used to construct efficient block-sampling \mcmc kernels in the framework of
particle \mcmc \citet{AndrieuDH:2010}; see Section~\ref{sec:pmcmc}. 
Our second result shows that the particle system generated by the \dcsmc procedure is
consistent as the number of particles tends to infinity.

\begin{prop} \label{prop:consistent}
Under regularity conditions detailed in Appendix~\ref{app:proofs:consistency},
the weighted particle system $(\x_{r,\Np}^i, \w_{r,\Np}^i)_{i=1}^\Np$ generated by $\RMC(r)$
is consistent in that for all functions 
$f:\setZ\rightarrow \mathbb{R}$ satisfying the assumptions listed in Appendix~\ref{app:proofs:consistency}:
\begin{align*}
  \sum_{i=1}^{\Np} \frac{\w_{r,\Np}^i}{\sum_{j=1}^\Np \w_r^{\Np,j}}f(\x_{r,\Np}^i) &\pcv \int f(\x) \pi(\x)d\x,
&&\text{as $\Np\goesto\infty$}.
\end{align*}
\end{prop}

\newcommand{\model}{{\mathcal M}}

\subsection{Tree structured auxiliary distributions from graphical models}\label{sec:decompositions}

We now present one strategy for building tree structured auxiliary distributions from graphical models. There are other ways of constructing these auxiliary distributions,
but for concreteness we focus here on a method targeted at posterior inference
for \pgm{s}. On the other hand, the method we present here is more general than it may appear at first: in particular, although the flow of the algorithm follows a tree structure, we do not assume that the graphical models are acyclic.

To illustrate the concepts in this section, we will use two running examples: one coming from a directed \pgm, and one coming from an undirected one. We use the factor graph notation from Section~\ref{sec:factor-graphs} to introduce the features of these examples salient to the present discussion. We give a more detailed description of the two examples in Section~\ref{sec:experiments}.

\begin{example*}[Hierarchical models, continued]
Consider a situation where the data is collected according to a known hierarchical structure. For example, test results for an examination are collected by school, which belong to a known school district, which belong to a known county. 
This situation is similar to the example shown in Figure~\ref{fig:method:hierarchical_model}, but where we generalize the number of level to be an arbitrary integer, $\alpha$. This yields the factor graph shown in Figure~\ref{fig:factor-graphs}(a), where we assume for simplicity a binary structure (this is lifted in Section~\ref{sec:experiments}). The nodes in the set $V$ correspond to latent variables specific to each level of the hierarchy. For example, a variable, $x_v$, at a leaf encodes school-specific parameters from a set, $\isetX_v$, those at the second level, district-specific parameters, etc. The set of factors, $F$, contain one binary factor, $\phi(x_v, x_{v'})$, between each internal node, $v'$, and its parent, $v$. There is also one factor, $u_r$, at the root to encode a top level prior.
\end{example*}

\begin{figure}[tp]
\begin{center}
  \includegraphics[width=3.5in]{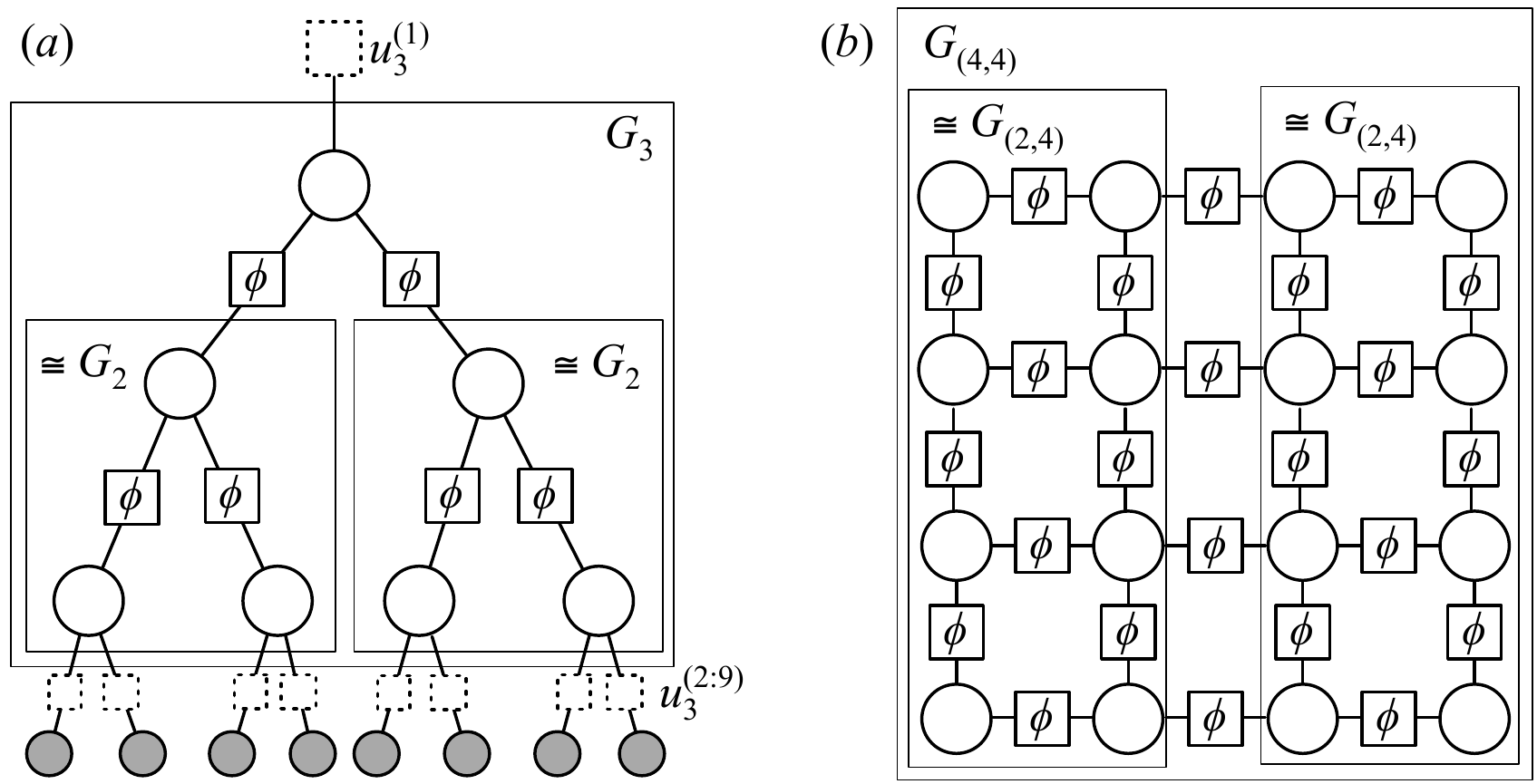}
  \caption{Examples of factor graph families, and self-similarities among them. (a) A hierarchical model for $\alpha = 3$. The unaries $u_\alpha$ consist in a product of 9 individual unary factors: one for the root, and 8 for the leaves (note that the binary factors connected to the 8 observed leaves can be considered as unary factors since one of their arguments is fixed and known for the purpose of posterior inference). Hence, $|V_3| = 7$, $|F_3| = 6$, $k_\alpha = 2$, $\alpha_1 = \alpha_2 = 2$. (b) A rectangular lattice model (\eg, an Ising model) for $\alpha = (4,4)$. Here, $k_\alpha = 2$, $\alpha_1 = \alpha_2 = (2,4)$. }
  \label{fig:factor-graphs}
\end{center}
\end{figure}

\begin{example*}[Lattice models]
Two-dimensional regular lattice models such as the Ising model are frequently used in spatial statistics and computer vision to encourage nearby locations in a spatial latent field to take on similar values; see Figure~\ref{fig:factor-graphs}(b). We denote the width of the grid by $\alpha^{(1)}$ and the height by $\alpha^{(2)}$. The cardinality of $V$ is thus $\alpha^{(1)} \alpha^{(2)}$. The bivariate factors connect variables with Manhattan distance of one to each other.
\end{example*}

Note that the previous two examples actually describe a collection of factor graphs indexed by $\alpha$: in the hierarchical model example, $\alpha$ in a positive integer encoding the number of hierarchical levels; in the lattice model example, $\alpha$ is a pair of positive integers, $\alpha = (\alpha^{(1)}, \alpha^{(2)})$ encoding the  width and height of the grid.
To formalize this idea, we define a \emph{model family} as a collection of factor graphs: $\model = \{G_\alpha = (V_\alpha, F_\alpha)\}$, where $V_\alpha \neq \emptyset$. Since we would like the concept of model family to encode the model structure rather than some observation-specific configurations, it will be useful in the following to assume that the sets $F_\alpha$ only contain factors linking at least two nodes. Given $G_\alpha$ and a dataset, it is trivial to add back the unary factors, denoted $u_\alpha$. We assume that for all $\alpha$, adding these unary factors to the product of the factors in $F_\alpha$ yields a model of finite normalization, $\int u_\alpha \prod_{\phi\in F_\alpha} \phi \ud \mu < \infty$.

To build a tree of auxiliary distributions, we rely on a notion of self-similarity. We start with an illustration of what we mean by self-similarity in the two running examples.
\begin{example*}[Hierarchical models, continued]
Consider the factor graph $G_3 = (V_3, F_3)$ corresponding to a three-level hierarchical model. If we exclude the unary factor at the root, we see that $G_3$ contains $G_2$ as a subgraph (see Figure~\ref{fig:factor-graphs}(a)). In fact, $G_3$ contains two distinct copies of the graph $G_2$.
\end{example*}
\begin{example*}[Lattice models, continued]
Consider the factor graph $G_{(4,4)}$ corresponding to a $4$-by-$4$ Ising model (Figure~\ref{fig:factor-graphs}(b)). The graph $G_{(4,4)}$ contains the graph $G_{(2,4)}$ as a subgraph. Again, $G_{(4,4)}$ contains in fact two distinct copies of the subgraph.
\end{example*}
Formally, we say that a model family is \emph{self-similar}, if given any $G_\alpha \in \model$ with $|V_\alpha| > 1$, we can find $\alpha_1, \alpha_2, \dots, \alpha_{k_\alpha}, k_\alpha > 1$ such that the disjoint union $\sqcup_i G_{\alpha_i}$ can be \emph{embedded} in $G_\alpha$. By embedding, we mean that there is a one-to-one graph homomorphism from $\sqcup_i G_{\alpha_i}$ into $G_\alpha$. This graph homomorphism should respect the labels of the nodes and edge (i.e. differentiates variable, factors, and the various types of factors).
\begin{example*}[Lattice models, continued]
Therefore, if $|V_\alpha| > 1$, then at least one of $\alpha^{(1)}$ or $\alpha^{(2)}$ is greater than one, let us say the first one without loss of generality. As shown in Figure~\ref{fig:factor-graphs}(b), we can therefore pick $k_\alpha = 2$ and $\alpha_1 = (\lfloor\alpha^{(1)}/2\rfloor, \alpha^{(2)})$, $\alpha_2 = (\lceil\alpha^{(1)}/2\rceil, \alpha^{(2)})$. 
\end{example*}
Given a member $\alpha_0$ of a self-similar model family, there is a natural way to construct a tree decomposition
of auxiliary distributions. First, we recursively construct $T$ from the self-similar model indices: we set $r = \alpha_0$, and given any $t = \alpha \in T$, we set  
$\C(t)\subset\T$ to $\alpha_1, \alpha_2, \dots, \alpha_{k_\alpha}$.\footnote{This recursive process will yield a finite set $T$: since $V_\alpha$ is assumed to be non-empty, it suffices to show that $|V_{\alpha_i}| < |V_{\alpha}|$ for all $i \in \{1, \dots, k_\alpha\}$ whenever $|V_\alpha| > 1$. But since $k_\alpha > 1$, and that the disjoint union $\sqcup_i G_{\alpha_i}$ can be embedded in $G_\alpha$, it follows that $|V_{\alpha}| \geq |V_{\alpha_i}| + \sum_{j \neq i} |V_{\alpha_j}|$. Since $|V_{\alpha_j}| > 0$, the conclusion follows.} Second, given an index $t = \alpha \in T$, we set $\pi_t$ to $u_\alpha \prod_{\phi\in F_\alpha} \phi$. Note that by the embedding property, this choice is guaranteed to satisfy Equation~(\ref{eq:space-decomposition}), where $\setX_{c_i}$ corresponds to the range of the random vector defined from the indices in $V_{\alpha_i}$


\section{Extensions}\label{sec:extensions}

Algorithm~\ref{alg:dcsir} is essentially an SIR algorithm and variables are
not rejuvenated after their first sampling. Inevitably, as in particle
filtering, this will lead to degeneracy as repeated resampling steps reduce
the number of unique particles. 
Techniques employed to ameliorate this problem in the
particle filtering literature could be used---fixed lag techniques \citep{Kitagawa:1996} might
make sense in some settings, as could incorporating MCMC moves
\citep{GilksB:2001}.
In this section we present several extensions to address the degeneracy problem
more directly, and we also discuss adaptive schemes for improving the computational efficiency
of the proposed method.
The extensions presented here comprise fundamental components of the general strategy
introduced in this paper, and may be required to obtain good performance in
challenging settings. 

\subsection{Merging subpopulations via mixture sampling}\label{sec:extensions:mixturesampling}

The resampling in Step~\ref{step:resample} of the $\RMC$ procedure, which combines subpopulations to target a new distribution on a larger space,
is critical. The independent multinomial resampling of child populations in the basic
\dcsmc[\sir] procedure corresponds to sampling $\Np$ times with replacement
from the product measure \eqref{eq:dcsir-product-measure}.
The low computational cost of this approach is appealing, but unfortunately it can lead to high variance when the product
 $\prod_{c\in\C(t)} \pi_c(\x_c)$ differs substantially from the corresponding
marginal of $\pi_t$.

An alternative approach, akin to the mixture proposal approach \citep{CarpenterCF:1999} or
the auxiliary particle filter \citep{PittS:1999}, is described below.
The idea is to exploit the fact that the product
measure has mass upon $\Np^{|\C(t)|}$ points, in order to capture the dependencies among
the variables in the target distribution $\pi_t(\x_t)$. 
Let $\check\pi_{t}(\x_{c_1},\ldots,\x_{c_C})$ be some distribution which
incorporates this dependency (in the
simplest case we might take  $\check\pi_{t}(\x_{c_1},\ldots,\x_{c_C}) \approx
\int\pi_{t}(\x_{c_1},\ldots,\x_{c_C},\ix_t) \rmd\ix_t$ or, when 
$\widetilde{\setX}_t = \emptyset$, $\check\pi_{t} \equiv \pi_{t}$; see below for an
alternative). We can then replace  
Step~\ref{step:resample} of Algorithm~\ref{alg:dcsir} with simulating $\{(\rx_{c_1}^i, \dots,
\rx_{c_C}^i)\}_{i=1}^{\Np}$ from
\begin{align}
  \label{eq:extensions:resampling}
  Q_{t}(\rmd\x_{c_1}, \dots, \rmd\x_{c_C} ) &:=
  \sum_{i_1=1}^{\Np} \ldots \sum_{i_C=1}^{\Np}
  \frac{v_{t}(i_1,\ldots,i_C)  \delta_{(\x_{c_1}^{i_1},\ldots,\x_{c_C}^{i_C})} ( \rmd\x_{c_1}, \dots, \rmd\x_{c_C} )}{
    \sum_{j_1=1}^{\Np} \ldots \sum_{j_C=1}^{\Np}  v_{t}(j_1,\ldots,j_C)
  },\\
  v_{t}(i_1,\ldots,i_C) &:=  \left( \prod_{c\in\C(t)} \w_c^{i_c}
  \right)\check\pi_t(\x_{c_1}^{i_1},\ldots,\x_{c_C}^{i_C}) \bigg/ \prod_{c\in\C(t)}
    \pi_{c} (\x_c^{i_c}), \notag
\end{align}
with the weights of Step~\ref{step:weighting} computed using
$ \w_t^i \propto \pi_t(\x_t^i) / \left[ \check\pi_{t}(\rx_{c_1}^{i_1},\ldots,\rx_{c_C}^{i_C}) q_t(\ix_t^i \mid \rx_{c_1}^{i_1}, \dots, \rx_{c_C}^{i_C}) \right]$.
Naturally, if we take $\check\pi_{t}(\x_{c_1},\ldots,\x_{c_C}) = \prod_{c\in\C(t)} \pi_c(\x_c)$ we recover the
basic approach discussed in Section~\ref{sec:dc-sir}.

Clearly, the computational cost of simulating from $Q_t$ will be
$O(\Np^{|\mathcal{C}(t)|})$. However, we envisage that both $|\mathcal{C}(t)|$
and the number of coalescence events (\ie combinations of subpopulations via
this step) are sufficiently small that this is not a problem in many cases.
Should it be problematic, computationally efficient use of products of mixture
distributions is possible employing the strategy of \citet{BriersDS:2005}.
At the cost of introducing a small and controllable bias, techniques borrowed from $N$-body problems could also be
used when dealing with simple local interactions
\citep{GrayM:2001}.
Furthermore, if the mixture sampling
approach is employed it significantly mitigates the
impact of resampling, and it is possible to reduce the branching factor
by introducing additional (dummy) internal nodes in $T$.
For example, by introducing additional nodes in order to obtain a binary tree
(see Section~\ref{sec:expts:mrf}), the merging of the child populations
will be done by coalescing pairs,
then pairs of pairs, \etc., gradually taking the dependencies between the variables into account.

\subsection{SMC samplers and tempering within \dcsmc} \label{sec:tempering}

As discussed in Section~\ref{sec:back-smc}, a common strategy when simulating from some
complicated distribution using \smc is to construct a synthetic sequence of distributions \eqref{eq:auxiliary-pi} which
moves from something tractable to the target distribution of interest \citep{DelMoralDJ:2006}.
The \smc proposals can then, for instance, be chosen as MCMC transition kernels---this is the approach that
we detail below for clarity.

Step~\ref{step:propagate} of Algorithm~\ref{alg:dcsir} corresponds essentially to a (sequential)
importance sampling step.
Using the notation introduced in the previous section, we obtain after the resampling/mixture sampling
step an unweighted sample $\{(\rx_{c_1}^i, \dots, \rx_{c_C}^i)\}_{i=1}^{\Np}$
targeting $\check\pi_t$, which is extended by sampling
from $q_t(\ix_t \mid \x_{c_1}, \dots, \x_{c_C})$, and then re-weighted
to target $\pi_t(\x_t)$. We can straightforwardly replace this
with several \smc sampler iterations, targeting
distributions which bridge from
$\pi_{t,0}(\x_t) = \check\pi_t(\x_{c_1}, \dots, \x_{c_C}) q_t(\widetilde \x_t \mid \x_{c_1}, \dots, \x_{c_C})$
to $\pi_{t,\nn_t}(\x_t) = \pi_t(\x_t)$,
typically by following a geometric path $\pi_{t,\jj}
\propto \pi_{t,0}^{1-\alpha_\jj} \pi_{t,\nn_t}^{\alpha_\jj}$ with $0 < \alpha_1 < \ldots < \alpha_{\nn_t} = 1$.
Step~\ref{step:propagate} of Algorithm~\ref{alg:dcsir} is then replaced by:
\begin{enumerate}[$2^\prime$.]
\item\begin{enumerate}
  \item For $i=1$ to $\Np$, simulate $\ix_t^i \sim q_t(\cdot \mid \rx_{c_1}^i, \dots, \rx_{c_C}^i)$.
  \item For $i=1$ to $\Np$, set $\x_{t,0}^i = (\rx_{c_1}^i, \dots, \rx_{c_C}^i, \ix_t^i)$ and $\w_{t,0}^i = 1$.
  \item For SMC sampler iteration $\jj=1$ to $\nn_t$:
    \begin{enumerate}
    \item For $i=1$ to $\Np$, compute
      $ \w_{t,\jj}^i = \w_{t,\jj-1}^i \gamma_{t,\jj}(\x_{t,\jj-1}^i) / \gamma_{t,\jj-1}(\x_{t,\jj-1}^i) $.
    \item Optionally, resample $\{ \x_{t,\jj-1}^i, \w_{t,\jj}^i \}_{i=1}^\Np$ and override the notation $\{ \x_{t,\jj-1}^i, \w_{t,\jj}^i \}_{i=1}^\Np$
      to refer to the resampled particle system.
    \item For $i=1$ to $\Np$, draw
      \( \x_{t,\jj}^i \sim K_{t,\jj}(\x_{t,\jj-1}^i ,\cdot) \) using a $\pi_{t,\jj}$-reversible Markov kernel $K_{t,\jj}$.
    \end{enumerate}
  \item Set $\x_t^i = \x_{t,\nn_t}^i$ and $\w_t^i = \w_{t,\nn_t}^i$.
  \end{enumerate}
\end{enumerate}
The computation of normalizing constant estimates has been omitted
for brevity, but follows by standard arguments (the complete algorithm is provided in Appendix~\ref{app:algo}).

We believe that the mixture sampling approach \eqref{eq:extensions:resampling}
can be particularly useful when combined with \smc tempering as described above.
Indeed, mixture sampling can be used to enable efficient initialization of each (node-specific) \smc sampler
by choosing,
for $\alpha_\star \in [0,1]$,
\begin{align}
  \label{eq:extensions:alphastar}
  \check\pi_t(\x_{x_1},\ldots,\x_{c_C}) \propto
  \left[ \prod_{c\in\C(t)} \pi_{c} (\x_c)  \right]^{1-\alpha\star} \left[\int
    \pi_{t,n_t} (\x_{c_1},\ldots,\x_{c_C},\tilde{\x}_t) d\tilde{\x}_t\right]^{\alpha_{\star}}.
\end{align}
That is, we exploit the fact that the distribution \eqref{eq:extensions:resampling} has support
on $N^{|\C(t)|}$ points to warm-start the annealing procedure at a non-zero value of the annealing parameter $\alpha_\star$.
In practice, this has the effect that we can typically use fewer temperatures $n_t$.
In particular, if simulating from the MCMC kernels $K_{t,j}$ is computationally costly,
requiring fewer samples from these kernels can
compensate for the $\Ordo(\Np^{|\C(t)|})$ computational cost associated with \eqref{eq:extensions:resampling}.


\subsection{Particle MCMC} \label{sec:pmcmc}

The seminal paper by \cite{AndrieuDH:2010} demonstrated that \smc algorithms
 can be used to produce approximations of idealized 
block-sampling proposals within MCMC algorithms.
By interpreting these particle MCMC algorithms as standard
algorithms on an extended space, incorporating all of the variables sampled
during the running of these algorithms, they can be shown to be exact, in the sense that the
apparent approximation does not change the
invariant distribution of the resulting MCMC kernel. Proposition~\ref{prop:bias}, and in particular the construction used
in its proof, demonstrates how the class of \dcsmc algorithms can be incorporated within
the particle MCMC framework.
Such techniques are now essentially standard, and we do not dwell on this approach here.

\subsection{Adaptation} \label{sec:adaptation}

Adaptive SMC algorithms have been the focus of much attention in recent years.
\citet{DelMoralDJ:2012} provides the first formal validation of algorithms in
which resampling is conducted only sometimes according to the value of some
random quantity obtained from the algorithm itself.
We advocate the use of low variance resampling algorithms \citep[e.g.]{DoucCM:2005} to be
applied adaptively.
Other adaptation is possible within SMC algorithms. Two approaches
are analyzed formally by \citet{BeskosJKT:2014}: the adaptation of the
parameters of the MCMC kernels employed (step $2^\prime$(c)iii.) and of the
number and location of distributions employed within tempering algorithms,
\ie, $\nn_t$ and $\alpha_1,\ldots,\alpha_{\nn_t}$; see \eg, \cite{ZhouJA:2015} for one
approach to this.

Adaptation is especially appealing within \dcsmc: beyond the usual advantages
it allows for the concentration of computational effort on the more
challenging parts of the sampling process. Using adaptation will lead to more
intermediate distributions for the subproblems (\ie, the steps of the \dcsmc algorithm)
for which the starting and ending distributions are more different. 
Furthermore, it is also possible to adapt the parameter $\alpha_\star$~in~\eqref{eq:extensions:alphastar}---that is,
the starting value for the annealing process---based, \eg, on the effective sample size of the $\Np^{|\C(t)|}$
particles comprising \eqref{eq:extensions:resampling}.
In simulations (see Section~\ref{sec:expts:mrf}) we have found that the effect of such adaptation can result in $\alpha_\star = 1$
for many of the ``simple'' subproblems, effectively removing the use of tempering when this is not needed
and significantly reducing the total number of MCMC simulations.

As a final remark, we have assumed throughout that all particle populations are
of size~$\Np$, but this is not necessary. Intuitively, fewer particles are required
to represent simpler low-dimensional distributions than to represent more
complex distributions. Manually or adaptively adjusting the number of
particles used within different steps of the algorithm remains a direction for
future investigation.


\section{Experiments}\label{sec:experiments}

\subsection{Markov Random Field}\label{sec:expts:mrf}
One model class for which the \dcsmc algorithm can potentially be useful are Markov random fields (\mrf{s}).
To illustrate this, we consider the well-known square-lattice Ising model: each lattice site is associated with
a binary spin variable $x_k \in \{-1,1\}$ and the configuration probability is given by
$p(\z) \propto e^{-\beta E(\z)}$, where $\beta \geq 0$ is the inverse temperature and
$E(\z) = -\sum_{(k,\ell) \in \mathcal{E}} x_kx_{\ell}$ is the energy of the system.
Here, $\mathcal{E}$ denotes the edge set for the graphical model which we
assume correspond to nearest-neighbour interactions with periodic boundary
conditions, see Figure~\ref{fig:ising} (rightmost figure).

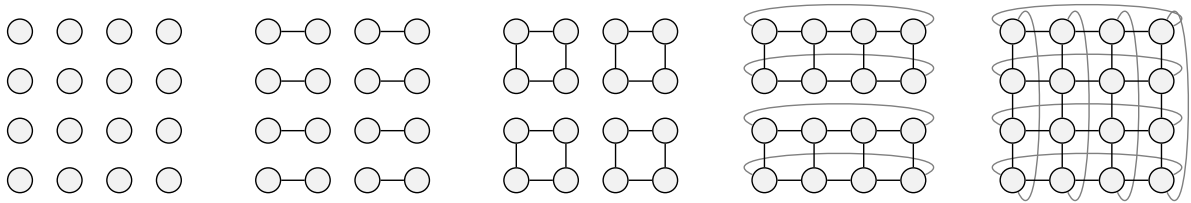
\begin{figure}[ptb]
  \centering
\tikzstyle{edge} = [-,thick]
\tikzstyle{thinedge} = [-,thick,color=black!50]
\tikzstyle{var} = [draw,circle,thick,inner sep=0,minimum width=0.5cm,fill=black!5]
\tikzstyle{angles} = [out=25,in=155] 
\tikzstyle{angles2} = [out=65,in=295] 
\resizebox{\linewidth}{!}{ 
  \begin{tikzpicture}[>=stealth',node distance=0.6cm]
    \begin{scope}
      \foreach \x in {0,1,2,3} {
        \foreach \y in {0,1,2,3} {
          \node at (\x,\y) (x\x\y) [var] {};
        }
      }
    \end{scope}
    \begin{scope}
      \foreach \x in {0,1,2,3} {
        \foreach \y in {0,1,2,3} {
          \node at (5+\x,\y) (x\x\y) [var] {};
        }
      }
      \foreach \x in {0,2} {
        \pgfmathtruncatemacro\xend{\x+1}
        \foreach \y in {0,1,2,3} {
          \draw[edge] (x\x\y) -- (x\xend\y) {};
        }
      }
    \end{scope}
    \begin{scope}
      \foreach \x in {0,1,2,3} {
        \foreach \y in {0,1,2,3} {
          \node at (10+\x,\y) (x\x\y) [var] {};
        }
      }
      \foreach \x in {0,2} {
        \pgfmathtruncatemacro\xend{\x+1}
        \foreach \y in {0,1,2,3} {
          \draw[edge] (x\x\y) -- (x\xend\y) {};
        }
      }
      \foreach \x in {0,1,2,3} {
        \foreach \y in {0,2} {
          \pgfmathtruncatemacro\yend{\y+1}
          \draw[edge] (x\x\y) -- (x\x\yend) {};
        }
      }
    \end{scope}
    \begin{scope}
      \foreach \x in {0,1,2,3} {
        \foreach \y in {0,1,2,3} {
          \node at (15+\x,\y) (x\x\y) [var] {};
        }
      }
      \foreach \y in {0,1,2,3} {
        \draw[thinedge] (x3\y) to [angles] (x0\y) {};
      }
      \foreach \x in {0,1,2} {
        \pgfmathtruncatemacro\xend{\x+1}
        \foreach \y in {0,1,2,3} {
          \draw[edge] (x\x\y) -- (x\xend\y) {};
        }
      }
      \foreach \x in {0,1,2,3} {
        \foreach \y in {0,2} {
          \pgfmathtruncatemacro\yend{\y+1}
          \draw[edge] (x\x\y) -- (x\x\yend) {};
        }
      }
    \end{scope}
    \begin{scope}
      \foreach \x in {0,1,2,3} {
        \foreach \y in {0,1,2,3} {
          \node at (20+\x,\y) (x\x\y) [var] {};
        }
      }
      \foreach \y in {0,1,2,3} {
        \draw[thinedge] (x3\y) to [angles] (x0\y) {};
      }
      \foreach \x in {0,1,2,3} {
        \draw[thinedge] (x\x3) to [angles2] (x\x0) {};
      }
      \foreach \x in {0,1,2} {
        \pgfmathtruncatemacro\xend{\x+1}
        \foreach \y in {0,1,2,3} {
          \draw[edge] (x\x\y) -- (x\xend\y) {};
        }
      }
      \foreach \x in {0,1,2,3} {
        \foreach \y in {0,1,2} {
          \pgfmathtruncatemacro\yend{\y+1}
          \draw[edge] (x\x\y) -- (x\x\yend) {};
        }
      }
    \end{scope}
  \end{tikzpicture}
}%
%
  \vspace{-1.2\baselineskip} 
  \caption{The disconnected components correspond to the groups of variables that are targeted by the
    different populations of the \dcsmc algorithm. At the final iteration, corresponding to the rightmost figure, we recover the original, connected model.}
  \label{fig:ising}
\end{figure}

Let the lattice size be $M\times M$, with $M$ being a power of~2 for simplicity.
To construct the computational tree $T$ we make use of the strategy of Section~\ref{sec:decompositions}.
That is, we start by dividing the lattice into two halves, removing all the edges between them.
We then continue recursively, splitting each sub-model in two, until we obtain
a collection of $M^2$ disconnected nodes; see Figure~\ref{fig:ising}.
This decomposition of the model defines a binary tree $T$, on which the \dcsmc algorithm operates.
At the leaves we initialize $M^2$ independent particle populations by sampling uniformly on $\{-1, 1\}$. These
populations are then resampled, merged, and reweighted as we proceed up the tree, successively reintroducing the ``missing'' edges of the model
(note that $\isetX_t = \emptyset$ for all non-leaf nodes $t$ in this example).
This defines the basic \dcsmc[\sir] procedure for the \mrf. We also consider three extensions of this procedure:
\begin{description}
\item[\dcsmc (mix)] uses the mixture sampling strategy described in Section~\ref{sec:extensions:mixturesampling}: the edges
  connecting any two sub-graphs are introduced before the corresponding sub-populations are merged.
\item[\dcsmc (ann)] uses the tempering method discussed in Section~\ref{sec:tempering}: when the edges connecting two sub-graphs are
reintroduced this is done gradually according to an annealing schedule to avoid severe particle depletion at the later stages of the algorithm.
\item[\dcsmc (mix+ann)] uses both mixture sampling and tempering.
\end{description}
For the annealed methods, we use single-flip Metropolis-Hastings kernels. The annealing schedules are set adaptively based on the conditional ESS (CESS)
criterion of \citet{ZhouJA:2015}, with CESS threshold 0.995. For \dcsmc (mix+ann) we warm-start each annealing process by selecting $\alpha_\star$ in
\eqref{eq:extensions:alphastar} based on the CESS (threshold 0.95) in each of the marginals of $\check \pi_t(\x_{c_1}, \x_{c_2})$.
We also compare these methods with, \emph{(i)} a standard SMC sampler with adaptive annealing \cite{DelMoralDJ:2006}, and \emph{(ii)} a single-flip Metropolis-Hastings sampler. All methods were implemented in Matlab 8.0.

We consider a grid of size $64 \times 64$ with $\beta=0.4407$ (the critical temperature).
We ran the methods listed above with $\Np = 2^6$ to $2^{11}$ particles,
with the exception of \dcsmc[\sir] which got $2^{10}$ to $2^{15}$ particles
to more closely match its computational cost with the others'. 
The single-flip MH sampler was run for $2^{14}$ MCMC iterations (each iteration being one complete sweep), with the first $2^{10}$ iterations
discarded as burn-in.
We ran each method 50 times and considered the estimates of
\emph{(i)} the normalising constant $Z$ and \emph{(ii)} the expected value of the energy $\E[E(\x)]$.
The results are given in Figure~\ref{fig:expts:ising-E}.


\begin{figure}[tb]
  \centering
%
  \includegraphics[width=0.8\columnwidth]{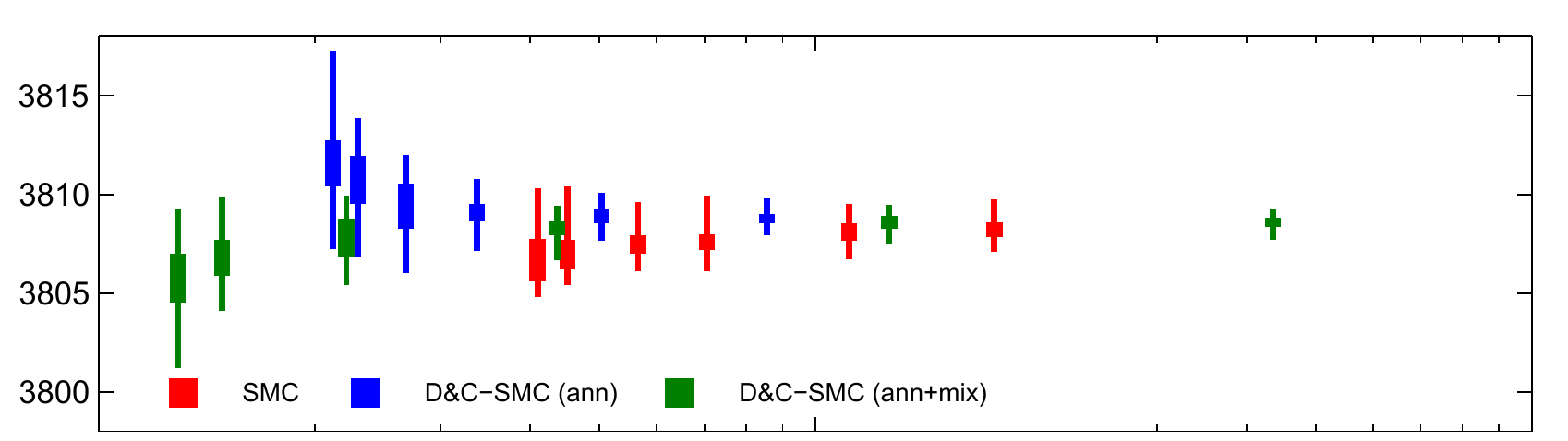}\\
  \includegraphics[width=0.8\columnwidth]{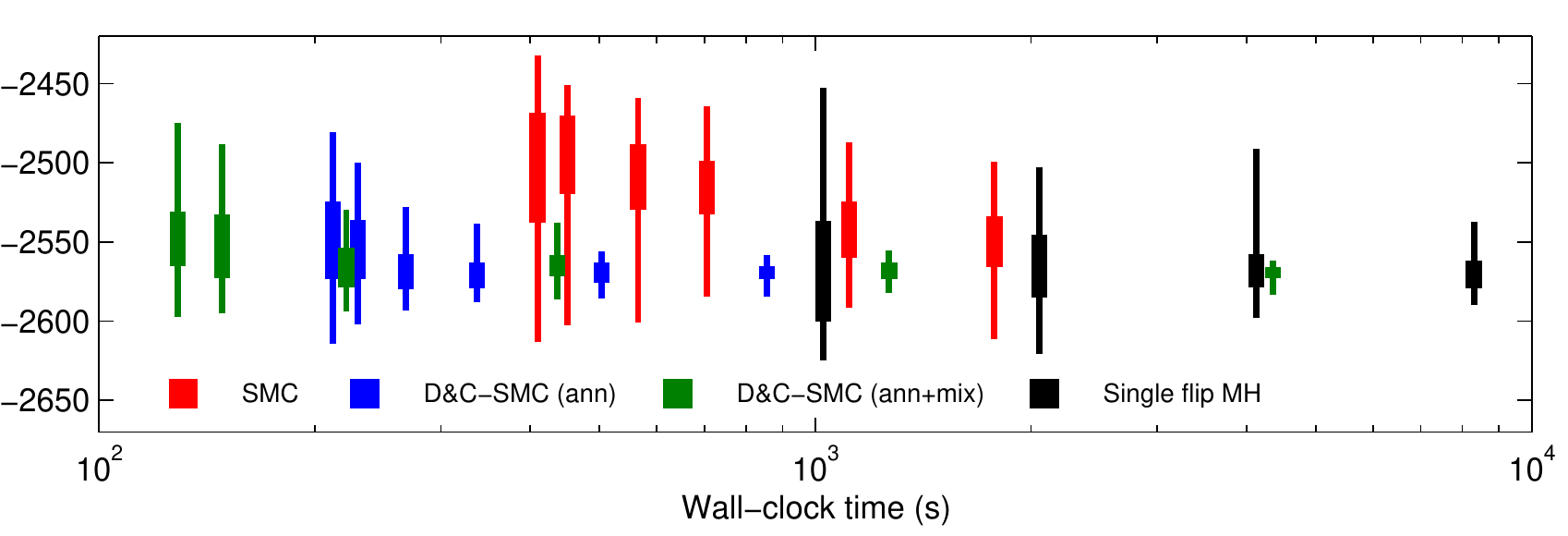}%
  \caption{Box-plots (min, max, and inter-quartile) of estimates of $\log Z$ (top) and $\E[E(\x)]$ (bottom) over 50 runs of each sampler
    (excluding single flip MH in the top panel since it does not readily  provide an estimate of $\log Z$).
    The boxes, as plotted from left to right, correspond to increasing number of particles $\Np$ (or number of MCMC iterations for single flip MH).}
  \label{fig:expts:ising-E}
\end{figure}

\dcsmc[\sir] and \dcsmc~(mix) gave inferior results to \dcsmc~(ann) and \dcsmc~(ann+mix) for this model and
have therefore been excluded (the results for all methods are given in Appendix~\ref{app:mrf}).
Among the remaining methods, \dcsmc (ann) and \dcsmc (mix+ann) give the overall best performance, significantly outperforming both standard \smc
and single flip MH sampling for the same computational time. 

For the two \dcsmc samplers---\dcsmc (ann) and \dcsmc (mix+ann)---the performance is comparable.
%
Whether or not it is worthwhile to make use of the mixture sampling strategy (Section~\ref{sec:extensions:mixturesampling})
is likely to be highly problem dependent. The benefit of using mixture sampling is that it can result in that fewer
annealing steps need to be taken. Indeed, for the simulations presented above, the \smc sampler
used on average (over all different settings and runs) 685 MCMC updates for each site. The corresponding
numbers for \dcsmc (ann) and \dcsmc (mix+ann) were 334 and 176, respectively. That is, for this example mixture sampling
essentially halves the number MCMC iterations that are taken compared to \dcsmc (ann) (which in turn uses only half the number of MCMC iterations
compared to standard \smc). Hence, for models where simulation from the MCMC kernel is computationally expensive it can be worthwhile to
use mixture sampling, even though its intrinsic computational cost is $\Ordo(\Np^2)$.

The fact that mixture sampling automatically results in more computational effort being spent on the most difficult subproblems
can be further illustrated by considering the distribution of values of the parameter $\alpha_\star$ in \eqref{eq:extensions:alphastar}.
Recall that $\alpha_\star \in [0,1]$ is the value of the annealing parameter at which the annealing process is warm-started.
In Figure~\ref{fig:expts:ising-beta-star} we show the distribution of $\alpha_\star$ for \dcsmc (mix+ann) with $\Np = 2\thinspace048$ particles
(similar distributions were obtained for the other settings as well), at different levels at the computational tree~$T$.
Due to the way $T$ is constructed, the number of edges that are ``added to the model'' at the merge steps increases
as we move upward in $T$. Indeed, for this model the depth of $T$ is $2\log_2(64) + 1 = 13$ and the number of edges
that are added during the merge-steps of the 12 non-leaf levels are: 1, 2, 2, 4, 4, 8, 8, 16, 16, 32, 64, 128. For the first five levels of $T$, 
we obtained $\alpha_\star = 1$ for all nodes, meaning that no annealing was performed during these steps of the algorithm (these levels are excluded from the figure).
For the subsequent levels we obtained values of $\alpha_\star$ less than 1, as can be seen in Figure~\ref{fig:expts:ising-beta-star},
but we are still able to warm-start the annealing at a non-zero value, effectively reducing the number of annealing steps that needs to be taken.

\begin{figure}
  \centering
  \includegraphics[width=0.5\columnwidth]{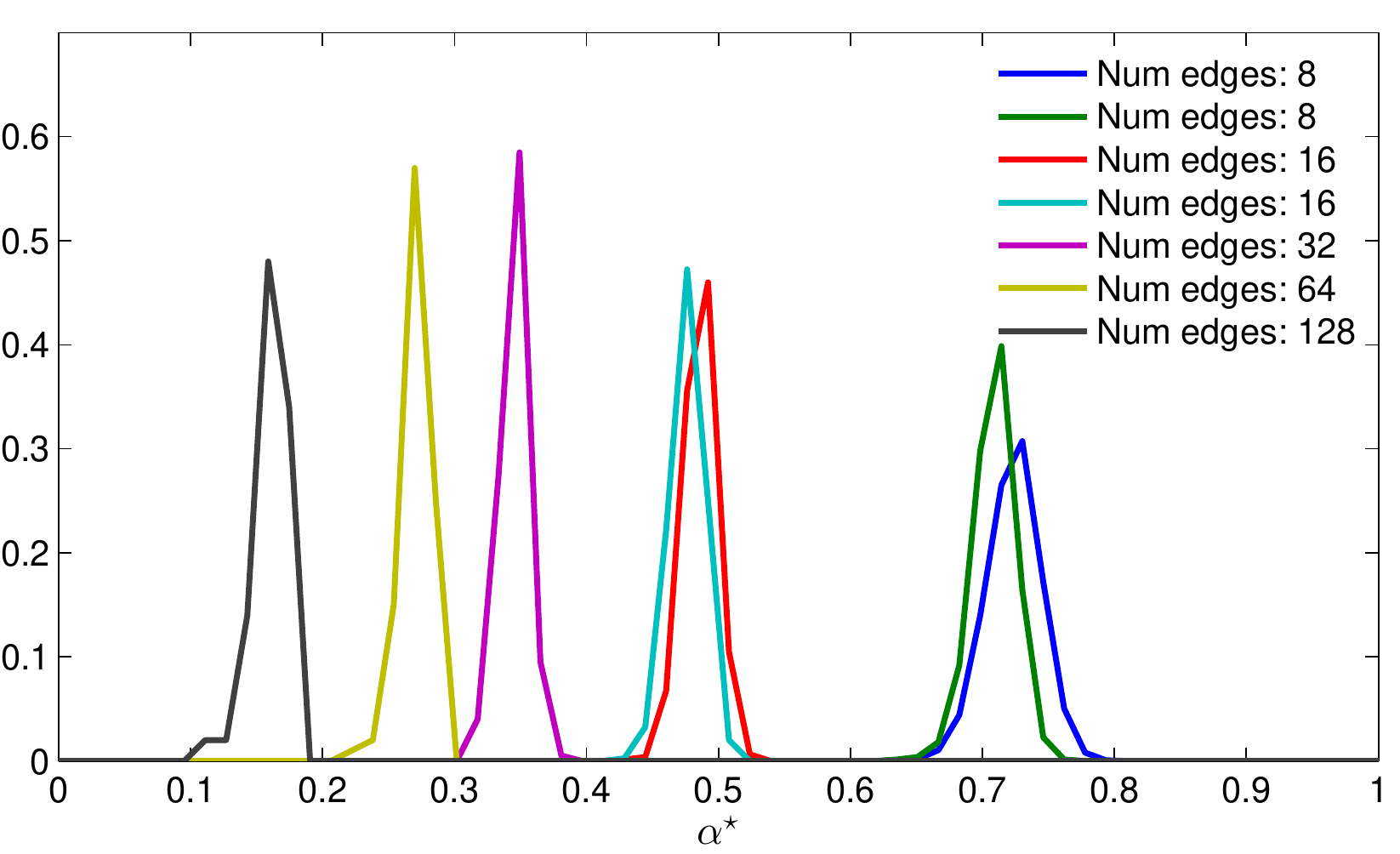}%
  \caption{Distributions of $\alpha_\star$ for the merge-steps at the 7 top-most levels of $T$ (computed for all nodes and over all 50 runs) for
    \dcsmc (mix+ann) with 2\thinspace048 particles.}
  \label{fig:expts:ising-beta-star}
\end{figure}

In Appendix~\ref{app:mrf} we report additional numerical results for the Ising model (different temperatures) as well as
for another square-lattice \mrf model with continuous latent variables and a multimodal posterior. These
additional results are in general agreement with the ones presented here.


\subsection{Hierarchical Bayesian Model -- New York State Mathematics Test}\label{sec:hierarchical}

In this section, we demonstrate the scalability of our method by analysing a
dataset containing \emph{New York State Mathematics Test} results for
elementary and middle schools. 

\subsubsection*{Data and model}

After preprocessing (data acquisition and preprocessing are described in detail in Appendix~\ref{app:expts:dataAnalysis}), we organize the data into a tree $T$.
A path from the root to a leaf has the following form: NYC (the root, denoted by $r\in T$), borough, school district, school, year. Each leaf $t\in T$ comes with an observation of $m_t$ exam successes out of $M_t$ trials. There were a total of 278\thinspace399 test instances in the dataset, split across five borough (Manhattan, The Bronx, Brooklyn, Queens, Staten Island), 32 distinct districts, and 710 distinct schools.

We use the following model, based on standard techniques from multi-level data analysis \citep{Gelman2006Multilevel}.
The number of successes $m_t$ at a leaf $t$ is assumed to be binomially distributed, with success probability parameter $p_t = \textrm{logistic}(\theta_t)$, where $\theta_t$ is a latent parameter. Moreover, we attach latent variables $\theta_t$ to internal nodes of the tree as well, and model the difference in values along an edge $e = (t \to t')$ of the tree with the following expression: $\theta_{t'} = \theta_t + \Delta_e$, where, $\Delta_e \sim \textrm{N}(0, \sigma_e^2)$.  We put an improper prior (uniform on $(-\infty, \infty)$) on $\theta_r$.\footnote{When $m_t \notin \{0, M_t\}$ for at least one leaf, this can be easily shown to yield a proper posterior.} We also make the variance random, but shared across siblings, $\sigma^2_e = \sigma_t^2 \sim \textrm{Exp}(1)$.

\subsubsection*{D\&C SMC implementation}

We apply the basic \dcsmc[\sir] to this problem, using the natural hierarchical structure provided by the model (see Section~\ref{sec:decompositions}).
Note that conditionally on values for $\sigma^2_t$ and for the $\theta_t$ at the leaves, the other random variables are multivariate normal. 
Therefore, we instantiate values for $\theta_t$ only at the leaves, and when proposing at an internal node $t'$, we only need to propose a value for $\sigma^2_{t'}$ as the internal parameters $\theta_{t'}$ can be analytically marginalized conditionally on $\sigma^2_{t'}$ and $\theta_{t'}$
using a simple message passing algorithm.

Each step of \dcsmc therefore falls in exactly one of two cases: \emph{(i)} At the leaves we propose a value for $p_t$
from a Beta distribution with parameters $1 + m_t$ and $1 + M_t - m_t$,
which we map deterministically to $\theta_t = \textrm{logit}(p_t)$. The corresponding weight update is a constant.
\emph{(ii)} At the internal nodes we propose $\sigma^2_t \sim \textrm{Exp}(1)$ from its prior. The weight update ratio involves the densities of marginalized multivariate normal distributions 
which can be computed efficiently using message passing.
Our Java implementation is open source and can be adapted to other multilevel Bayesian analysis scenarios. The code and scripts used to perform our experiments are available at {\texttt https://github.com/alexandrebouchard/multilevelSMC}.

The qualitative results obtained from DC with 10\,000 particles (Figure~\ref{fig:means}) are in broad agreement with other socio-economic indicators. For example, among the five counties corresponding to each of the five boroughs, Bronx County has the highest fraction (41\%) of children (under 18) living below poverty level.\footnote{Statistics from the New York State Poverty Report 2013, {\tiny \url{http://ams.nyscommunityaction.org/Resources/Documents/News/NYSCAAs\_2013\_Poverty\_Report.pdf}}} Queens has the second lowest (19.7\%), after Richmond (Staten Island, 16.7\%). However the fact that Staten Island contains a single school district means that our posterior distribution is much flatter for this borough.

\begin{figure}[tp]
\begin{center}
  \includegraphics[width=4.5in]{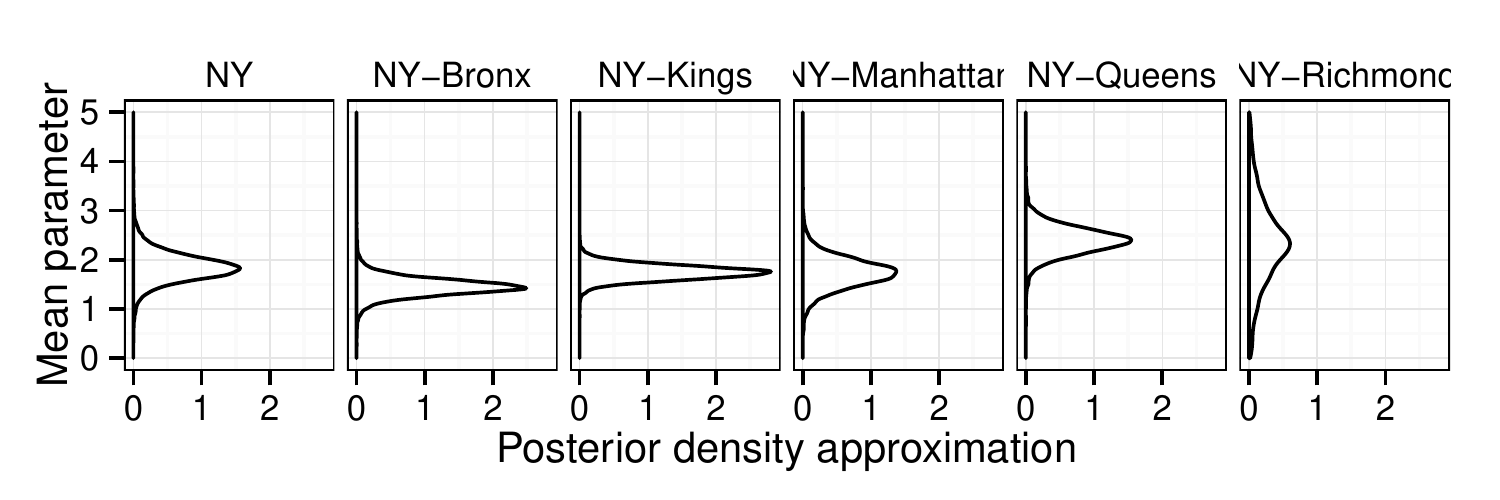}
  \caption[labelInTOC]{Posterior densities for the parameters $\theta_t$ with \dcsmc[\sir] ran with 10\thinspace000 particles.
    The values of the internal $\theta_t$ are marginalized during inference, but can be easily reinstantiated from the samples as a post-processing step.}
  \label{fig:means}
\end{center}
\end{figure}

\subsubsection*{Comparison of posterior inference methods}

For the purpose of comparison, we also applied three additional methods
\begin{description}
  \item[Gibbs:] A Metropolis-within-Gibbs algorithm, proposing a change on a single variable using a normal proposal of unit variance. As with \dcsmc[\sir], we marginalize the internal $\theta_t$ parameters. (Java implementation.)
  \item[STD:] A standard (single population) bootstrap filter with the intermediate distributions being sub-forests incrementally built in post-order. 
    The internal $\theta_t$-parameters are marginalized. (Java implementation.)
  \item[Stan:] An open-source implementation of the Hamiltonian Monte Carlo algorithm (see, \eg, \citet{Neal:2011}).
    We did not implement marginalization of the internal $\theta_t$-parameters. Stan includes a Kalman inference engine, however it is limited to chain-shaped \pgm{s} as of version 2.6.0. (C++ implementation.)
\end{description}
Further details on the baselines and the experimental setup can be found in Appendix~\ref{app:expts:dataAnalysis}.

%
%

%
%

We measure efficiency using effective sample size (ESS) per minute, as well as convergence of the posterior distributions on the parameters.  
For the MCMC methods (Gibbs and Stan), the ESS is estimated using the standard auto-regressive method, as implemented by \citet{plummer2012package}. 
For the SMC methods (\dcsmc and STD), the non-sequential nature of the samples dictates a different estimator, hence, again following standard practices, we use the estimator described by \citet{KongLW:1994}. For both MCMC and SMC methods, wall-clock time is measured on Intel Xeon E5430 quad-core processors, running at 2.66 GHz.
We replicated all running-time experiments ten times.

We begin with the ESS per minute results for the SMC methods ran with 10\,000 particles. For \dcsmc, we obtained a mean ESS/min of 636.8 (19.3), and for STD, of 537.8 (53.2). The diagnostics suggest that both methods perform reasonably well, with a slight advantage to \dcsmc. 
In contrast, the MCMC diagnostics raised inefficiency concerns. 
For Gibbs (300\,000 iterations), we obtained a mean ESS/min of 0.215 (standard deviation of 0.010). 
The performance of Stan (20\,000 iterations) was inferior, and more volatile, with a mean of 0.000848 (0.0016).
We attribute the poor performance of the Stan baseline to the fact that it does not marginalize the parameters $\theta$ (the reason for this is explained in the previous section).

Since the different types of samples impose the use of two different ESS estimators, direct comparisons of ESS/min between an SMC and an MCMC method should be taken with a pinch of salt. However, these results show that the sampling problem we are investigating in this section is indeed a  challenging one. This is not a surprise, given the high-dimensionality of the latent variables (3\thinspace555 remaining parameters after marginalization of the multivariate normal). Moreover, our results on the convergence of the posterior distributions on the parameters recapitulate that \emph{(i)} the SMC methods strongly outperform the MCMC baselines in this problem, and \emph{(ii)} \dcsmc and STD perform similarly, with a slight advantage for \dcsmc. Additional results supporting this claim
can be found in Appendix~\ref{app:expts:dataAnalysis}.

Next, to better differentiate the two SMC methods, we investigate estimation of the log-normalizing constant $\log(Z)$. The results are shown in Figure~\ref{fig:logZ}. Since there is little change between the \dcsmc estimate with 100k and 1M particles ($-3\,811.60 (0.80)$ and $-3\,811.15 (0.33)$ respectively), it is reasonable to assume that the true negative log-normalization is in the range 3\,811--3\,812.  
Under this assumption, \dcsmc outperforms STD on all computational budgets. Note that the abscissa is in logscale, suggesting that in the large computational budgets, \dcsmc requires roughly one order of magnitude less particles than STD to reach a similar accuracy. 
These results have practical implications to situations where particle MCMC is required. Indeed, in the light of \citet{DoucetPDK:2014}, where the authors recommend a standard deviation of the log-likelihood estimator in the range of 1--1.7, around 10\,000 particles would be sufficient in the case of \dcsmc (standard deviation for 10\,000 particles is 1.7), whereas closer to 100\,000 particles would be needed in the case of STD (standard deviation for 10\,000 particles is 2.5).

\begin{figure}[tp]
\begin{center}
  \includegraphics[height=2in]{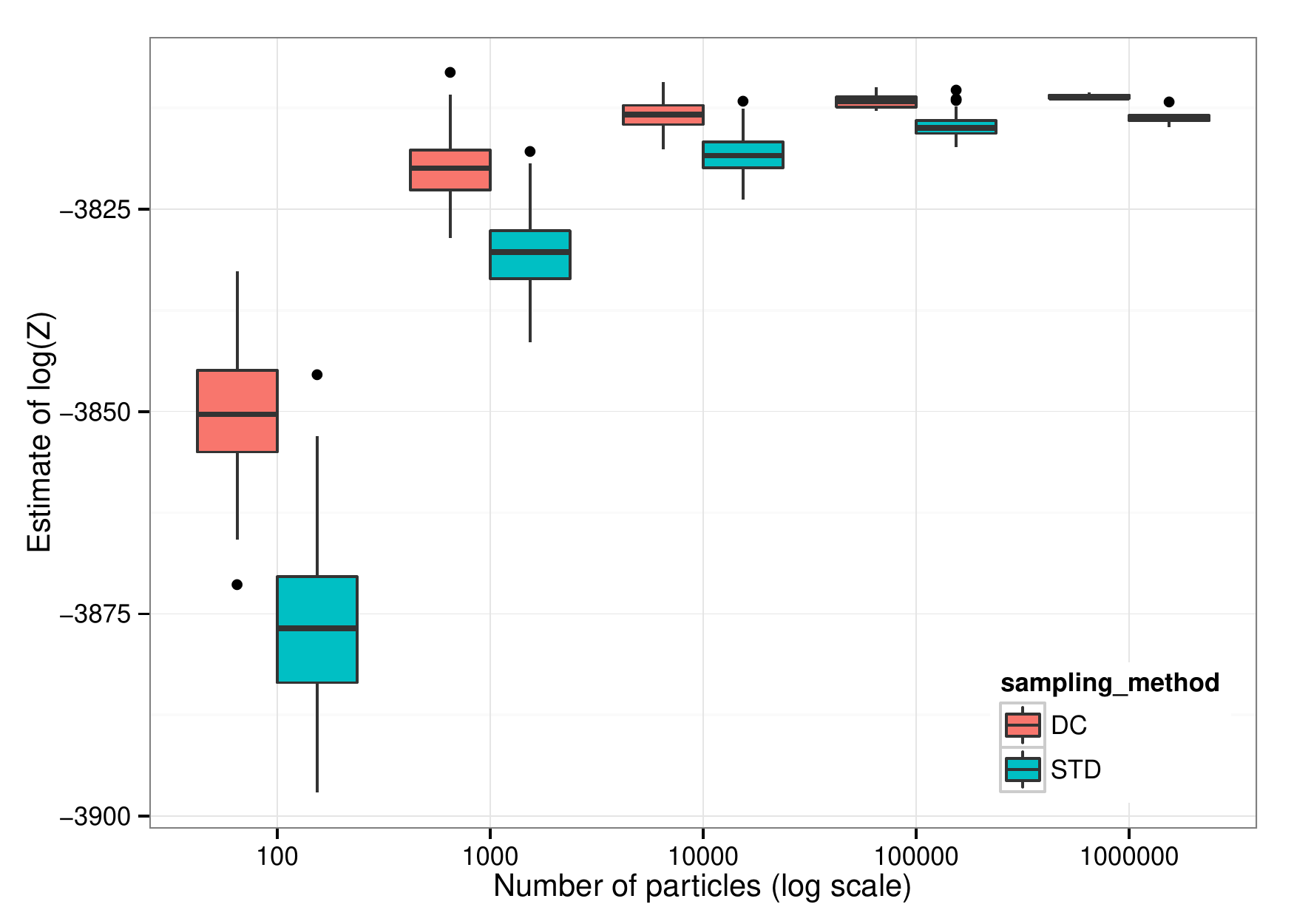}
  \includegraphics[height=2in]{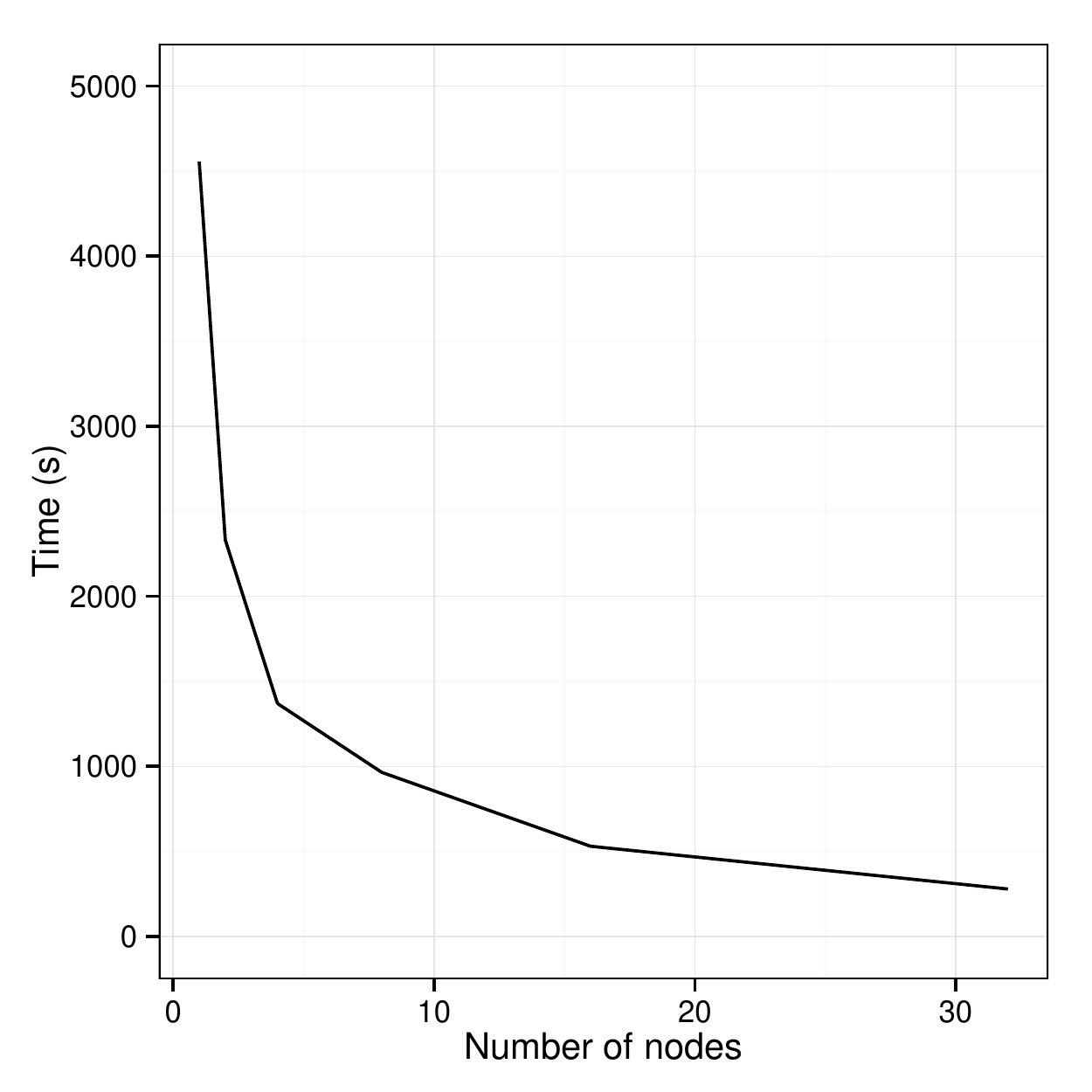}
  \caption{Left: Estimates of $\log(Z)$ obtained using \dcsmc and STD with different numbers of particles. Each experiment was replicated 110 times (varying the Monte Carlo seed), except for the experiments with 1M particles, which were replicated only 10 times. Right: Wall-clock times for the distributed \dcsmc algorithm. See Appendix~\ref{app:expts:dataAnalysis} for speed-up results.}
  \label{fig:logZ}
\end{center}
\end{figure}

\subsection{Distributed Divide-and-Conquer}\label{sec:distributed}

To demonstrate the suitability of \dcsmc to distributed computing environments, we have implemented a proof-of-concept distributed \dcsmc algorithm. The main idea in this implementation is to split the work at the granularity of populations, instead of the more standard particle granularity. The description and benchmarking of this implementation can be found in Appendix~\ref{app:expts:dataAnalysis}.
Using this distributed implementation, we see for example (Figure~\ref{fig:logZ}, right) that the running time for $100\thinspace000$ particles can be reduced from $4\,557$ seconds for one machine (about $1\frac{1}{4}$ hours), to $279$ seconds (less than five minutes) using 32 compute nodes (each using a single thread).


\section{Discussion}
We have shown that trees of auxiliary distribution can be leveraged by \dcsmc samplers to provide computationally efficient approximations of the posterior distribution of high-dimensional and possibly loopy probabilistic graphical models. Our method, which generalizes the \smc framework, is additionally easy to distribute across several compute nodes.

As with standard \smc (and other advanced computational inference methods) \dcsmc allows for a large degree of flexibility, and the method 
should be viewed as a toolbox rather than as a single algorithm. 
Indeed, we have discussed several possible extensions of the basic method, and their utility is problem-specific. 
Furthermore, the interplay between these extensions needs to be taken into account.
In particular, based on the numerical results in Section~\ref{sec:expts:mrf} we argue that mixture
sampling can be useful when used in conjunction with MCMC-based tempering, especially when simulating from
the MCMC kernel is computationally costly. In such scenarios, the warm-starting of the tempering process
enabled by mixture sampling can compensate for the polynomial (in the branching factor of the tree) computational cost of mixture
sampling. 

We have assumed in this work that the topology of the tree of auxiliary distributions is known and fixed. In practice, several different decompositions are possible. We have presented one systematic way of obtaining a tree decomposition for self-similar graphical models. However, a natural question to ask is how to choose an optimal decomposition. We are exploring several approaches to address this question, including strategies that mix several decompositions. How the components of these mixtures should interact is a question we leave for future work.



{\footnotesize
\bibliographystyle{apalike}
\bibliography{references}
}

\appendix
\section{Theoretical Properties} \label{app:proofs}

\subsection{Proof of Proposition~\ref{prop:bias}}


We consider all of the random variables simulated in the running of the
algorithm, following the approach of \cite{AndrieuDH:2010} for standard SIR
algorithms. A direct argument provides for the unbiasedness of the normalizing
constant estimate by demonstrating that the ratio $\widehat{Z}_r / Z_r$ is equal
to the Radon-Nikod{\'y}m derivative of a particular extended target
distribution to the joint sampling distribution of all of the random variables
generated during the running of the algorithm (and is hence of unit expectation).

We provide explicit details for the case in which $T$ is a
balanced binary tree (i.e. $|\mathcal{C}(t)| = 2$ for every non-leaf
node). The more general case follows by the same argument \emph{mutatis mutandis}.
We note in particular that the extension to balanced trees of degree greater
than two is trivial and that unbalanced trees may be addressed by the
introduction of trivial dummy nodes (or directly, at the expense of further
complicating the notation).

We assume that subpopulation $h$ at depth $d$ is
obtained from subpopulations $2h-1$ and $2h$ at depth $d+1$. 
Let,
\[\x_{\mis{D}}^{1:N} : = \{ \x_{\mi{d}{h}}^i: \mi{d}{h} \in \mis{D}, i \in \{1,\ldots,N\}\}\]
be the collection containing $N$ particles within each sub-population and let
\[\aa_{\miss{D}}^{1:N} := \{ \aa_{\mi{d}{h}}^i : \mi{d}{h} \in \miss{D}, i \in \{1,\ldots,N\}\}\]
be the ancestor indices associated with the resampling step; $\aa_{\mi{d}{h}}^i$ is
the ancestor of the $i^\textrm{th}$ particle obtained in the resampling step
associated with subpopulation $h$ at level $d$ of the tree.

First we specify sets in which multi-indices of particle populations live:
\begin{align*}
\mis{D} =& \bigcup_{d=0}^D \{d\} \otimes \{1,\ldots,2^d\}, &
\miss{D} =& \mis{D} \setminus \{0\} \otimes \{1\}.
\end{align*}
Thus, population $h$ at depth $d$, where the root of the tree is at a depth of
$0$, may be identified as $(d,h) \in \mis{D}$ for any $d \in \{0,\ldots, D\}$.

The joint distribution from which variables are simulated during the running
of the algorithm may be written as: 
\begin{align*}
\widetilde{q}&\left(\rmd\x_{\mis{D}}^{1:N}, \rmd\aa_{\miss{D}}^{1:N}\right) =
\prod_{h=1}^{2^D} \prod_{i=1}^N q_{\mi{D}{h}}(\rmd\x_{\mi{D}{h}}^i) \times
\prod_{\mi{d}{h} \in \miss{D}} \prod_{i=1}^N \nw_{(d,h)}^{\aa_{\mi{d}{h}}^i} \rmd\aa_{\mi{d}{h}}^i\\ & \times
\prod_{\mi{d}{h} \in \mis{D-1}} \prod_{i=1}^N
\delta_{\left(\x_{\mi{d+1}{2h-1}}^{\aa_{\mi{d+1}{2h-1}}^i},
    \x_{\mi{d+1}{2h}}^{\aa_{\mi{d+1}{2h}}^i}\right)} (\rmd(\x_{\mi{d}{h}}^i \setminus
\widetilde\x_{\mi{d}{h}}^i)) q_{\mi{d}{h}}(\rmd\widetilde\x_{\mi{d}{h}}^i|\x_{\mi{d}{h}}^i \setminus
\widetilde\x_{\mi{d}{h}}^i)
\end{align*}
where $\nw_{(d,h)}^i$ denotes the normalized (to sum to one within the realized
sample) importance weight of particle $i$ in subpopulation $h$ of depth
$d$. Note that this distribution is over \(
\left( \otimes_{(d,h)\in\mis{D}} \setX_{(d,h)} \right)^{\Np} \otimes  \{1,\ldots,\Np\}^{|\miss{D}|\Np}
\) and the $\rmd\aa$ corresponds to a counting measure over the index set. The
inclusion of the singular transition emphasizes the connection between this
algorithm and the standard SIR algorithm.

It is convenient to define ancestral \emph{trees} for our particles using the
following recursive definition:
\begin{align*}
b_{\mi{0}{1}}^i &= i & &\text{and } &
b_{\mi{d}{h}}^i&= \left\{
\begin{array}{ll}
\aa_{\mi{d}{h}}^{b_{\mi{d-1}{(h+1)/2}}^i }    & d \text{ odd,}\\
\aa_{\mi{d}{h}}^{b_{\mi{d-1}{h/2}}^i     }    & d \text{ even,}
\end{array}
\right.
\end{align*}
the intention being that $\{ b^i_{(d,h)} : (d,h) \in \miss{D}\}$ contains the
multi-indices of all particles which are ancestral to the $i^\text{th}$
particle at the root.

It is also useful to define an auxiliary distribution over all of the sampled
variables and an additional variable $k$ which indicates a selected ancestral
tree from the collection of $N$, just as in the particle MCMC context. Here we
recall that $\pi_r = \gamma_r / Z_r$:

\begin{align*}
\widetilde{\pi}_r\left(\rmd\x_{\mis{D}}^{1:N}, \rmd\aa_{\miss{D}}^{1:N}, \rmd k\right) =
\frac{\pi_r(\rmd\x_{\mi{0}{1}}^k)  \prod\limits_{\mi{d}{h} \in \mis{D-1}}
  \delta_{\x_{\mi{d}{h}}^{b_{\mi{d}{h}}^k}\setminus \widetilde \x_{\mi{d}{h}}^{b_{\mi{d}{h}}^k}}\left(\rmd\left(
    \x_{\mi{d+1}{2h-1}}^{b_{\mi{d+1}{2h-1}}^k} ,
    \x_{\mi{d+1}{2h-1}}^{b_{\mi{d+1}{2h}}^k}\right)\right) }{N^{|\mis{D}|}
 } & \\
\times 
\frac{
  \widetilde{q}\left(\rmd\x_{\mis{D}}^{1:N}, \rmd\aa_{\miss{D}}^{1:N}\right)
}{
  \prod\limits_{h=1}^{2^D} q_{\mi{D}{h}}(\rmd\x_{\mi{D}{h}}^{b_{\mi{D}{h}}^k})
   \left( \prod\limits_{(d,h) \in \miss{D}} \nw_{\mi{d}{h}}^{b_{\mi{d}{h}}^k} \right)
\prod\limits_{(d,h)\in\mis{D-1}}q_{\mi{d}{h}}(\rmd\widetilde\x_{\mi{d}{h}}^{b^k_{d,h}} |
 \x_{\mi{d}{h}}^{b^k_{d,h}}\setminus \widetilde\x_{\mi{d}{h}}^{b^k_{d,h}})
}\quad & \\
\times \left(
\prod\limits_{(d,h)\in\mis{D-1}}
\delta_{\left(
         \x_{\mi{d+1}{2h-1}}^{b_{\mi{d+1}{2h-1}}^k} ,
         \x_{\mi{d+1}{2h-1}}^{b_{\mi{d+1}{2h}}^k}\right)} \left(\rmd\left(\x_{\mi{d}{h}}^{b_{\mi{d}{h}}^k}\setminus
   \widetilde\x_{\mi{d}{h}}^{b_{\mi{d}{h}}^k}\right) \right)
\right)^{-1} \rmd k
\qquad&
\end{align*} 
which can be straightforwardly established to be a properly-normalized
probability. Note that the division by a product of singular measures should be
interpreted simply as removing the corresponding singular component of the
numerator.

Augmenting the proposal distribution with $k \in \{1,\ldots,\Np\}$ obtained in the same way is
convenient (and does not change the result which follows as $\widehat{Z}$ does
not depend upon $k$):
\begin{align*}
\bar{q}\left(\rmd\x_{\mis{D}}^{1:N}, \rmd\aa_{\miss{D}}^{1:N}, \rmd k\right) = \widetilde{q}\left(\rmd\x_{\mis{D}}^{1:N}, \rmd\aa_{\miss{D}}^{1:N}\right) \nw_{\mi{0}{1}}^k \rmd k.
\end{align*}

It is straightforward to establish that $\widetilde{\pi}_r$ is absolutely
continuous with respect to $\bar{q}$. Taking the Radon-Nikod{\'y}m derivative yields the following useful
result (the identification between $\x_{\mi{0}{1}}^k$ and its ancestors is
taken to be implicit for notational simplicity; as this equality
holds with probability one under the proposal distribution, this is
sufficient to define the quantity of interest almost everywhere). We allow $\pi_r$
and $q_{d,h}$ to denote the \emph{densities} of the measures which they correspond
to (with respect to an appropriate version of Lebesgue or counting measure)
\begin{align*}
\frac{d\widetilde{\pi}_r}{d\bar{q}} &\left(\x_{\mis{D}}^{1:N}, \aa_{\miss{D}}^{1:N},
    k\right)\\
&=\frac{\pi_r\left(\x_{\mi{0}{1}}^k
\right)}{N^{|\mis{D}|}
} 
\frac{
 1
}{
  \prod\limits_{h=1}^{2^D} q_{\mi{D}{h}}(\x_{\mi{D}{h}}^{b_{\mi{D}{h}}^k})
\prod\limits_{(d,h)\in\mis{D-1}}q_{\mi{d}{h}}(\widetilde\x_{\mi{d}{h}}^{b^k_{d,h}} |
 \x_{\mi{d}{h}}^{b^k_{d,h}}\setminus \widetilde\x_{\mi{d}{h}}^{b^k_{d,h}})
  \prod\limits_{(d,h) \in \mis{D}} \nw_{\mi{d}{h}}^{b_{\mi{d}{h}}^k}
}\\
&=
\frac{\pi_r\left(\x_{\mi{0}{1}}^k
\right)}{
 \prod\limits_{h=1}^{2^D} q_{\mi{D}{h}}(\x_{\mi{D}{h}}^{b_{\mi{D}{h}}^k})\prod\limits_{(d,h)\in\mis{D-1}}q_{\mi{d}{h}}(\widetilde\x_{\mi{d}{h}}^{b^k_{d,h}} |
 \x_{\mi{d}{h}}^{b^k_{d,h}}\setminus \widetilde\x_{\mi{d}{h}}^{b^k_{d,h}})
  \prod\limits_{(d,h) \in \mis{D}} \frac{\w_{\mi{d}{h}}^{b_{\mi{d}{h}}^k}}{\frac{1}{N} \sum_{i=1}^{N} \w_{\mi{d}{h}}^{i}}
}
\end{align*}
where $\w_{\mi{d}{h}}^i$ denotes the \emph{unnormalized} importance weight of
  particle $i$ in subpopulation $h$ at depth~$d$.

We can then identify the product $ \prod_{(d,h) \in \mis{D}}
\w_{\mi{d}{h}}^{b_{\mi{d}{h}}^k}$ with the ratio of density $\pi_r$ to
the assorted proposal densities evaluated at the appropriate ancestors of the
surviving particle 
multiplied by the normalizing constant $Z_r$ (by construction; these are
exactly the unnormalized weights). Furthermore, we have that
\[
\prod_{\mi{d}{h} \in \mis{D}} \frac{1}{N} \sum_{i=1}^{N} \w_{\mi{d}{h}}^{i} = \widehat{Z}_r
\]
which implies that $\widetilde{\pi}_r = (\widehat{Z}_r / Z_r) \bar{q}$.
Consequently, we obtain the result as:
\begin{align*}
\mathbb{E}_{\bar{q}}[\widehat{Z}_r] =
\mathbb{E}_{\bar{q}}[Z_r \widetilde{\pi}_r / \bar{q}] =  Z_r
\mathbb{E}_{\widetilde{\pi}_r}[1] = Z_r.
\end{align*}
\hfill\qedsymbol


\subsection{Consistency -- Proof of Proposition~\ref{prop:consistent}}\label{app:proofs:consistency}

We now turn to consistency. We make a few amendments to the notation used in the main text. First, as we consider limits in the number of particles $N$, we add an index $N$ the particles and their weights. Second, as in the preceding proof we use $\nw$ to denote the normalized weights. For simplicity we also assume a perfect
binary tree, $\C(t) = (\tleft, \tright)$, where $\tleft$ and $\tright$ denote the left
and right children of node $t$. The proof in this case captures the essential
features of the general case. 
We base our argument on \cite{DoucM:2008} (henceforth, DM08). We will use the following definitions and assumptions.

\begin{enumerate}
  \item \label{cond:L} 
We require that the Radon-Nikod{\'y}m derivative
\begin{align}\label{eq:support}
\frac{\text{d} \gamma_t}{\text{d } \gamma_{\tleft} \otimes \gamma_{\tright}
  \otimes q_t} (\x_\tleft,\x_\tright,\widetilde{\x}_t) < \infty,
\end{align}
is well-defined and finite almost everywhere.
\item \label{cond:test-fcts} We define $\setC_t$ to denote the collection of bounded integrable test function as in DM08.
\item \label{cond:bounds} We also make the following assumption (which could
  be relaxed and which is too strong to cover some of the algorithms and
  applications discussed in the main text, but which simplifies exposition): there exists
  a constant $C$ such that for all $t, \x_\tleft, \x_\tright, \widetilde{\x}_t$:
\begin{align*}
\frac{\text{d} \gamma_t}{\text{d } \gamma_{\tleft} \otimes \gamma_{\tright}
  \otimes q_t} (\x_\tleft,\x_\tright,\widetilde{\x}_t) < C,
\end{align*}
\item \label{cond:resampling} Assume multinomial or residual resampling is performed at every step.
\end{enumerate}

\begin{prop}\label{prop:consistent-appendix}
Under the above assumptions, the normalized weighted particles $(\x_N^i, \nw_N^i : i \in\{1, \dots, N\})$, $\nw_N^i = \w_N^i / \sum_j \w_N^j$, obtained from \RMC($r$) are consistent with respect to $(\pi, \setC_r)$, i.e.: for all test function $f\in \setC_r$, as the number of particles $N \to \infty$:
$$
\sum_{i=1}^N \nw_N^i f(\x_N^i) \pcv \int f(\x) \pi(\ud \x).
$$
\end{prop}

\begin{lemma}\label{lemma:finite-approx}
Let  $\sigmaalg$ denote a $\sigma$-algebra generated by a semi-ring $\sigmagen$, $\sigmaalg = \sigma(\sigmagen)$. Let $\pi$ be a finite measure constructed using a Carath\'eodory extension based on $\sigmagen$. Then, given $\epsilon > 0$ and $E \in \sigmaalg$, there is a finite collection of disjoint simple sets $R_1, \dots, R_n$ that cover $E$ and provide an $\epsilon$-approximation of its measure:
\begin{align*}
\mu(A) &\le \mu(E) + \epsilon, \\
A &= \bigcup_{j=1}^n R_j \supset E.
\end{align*}
\end{lemma}

\begin{proof}
From the definition of the Carath\'eodory outer measure, and the fact that it coincide with $\pi$ on measure sets such as $E$, we have a countable cover $R_1, R_2, \dots$ with:
\begin{align*}
\mu(\widetilde A) &\le \mu(E) + \frac{\epsilon}{2}, \\
\widetilde A &= \bigcup_{j=1}^\infty R_j \supset E.
\end{align*}Moreover, since $\mu(E) < \infty$, the sum can be truncated to provide a finite $\epsilon$-approximation. Finally, since $\sigmagen$ is a semi-ring, we can rewrite the finite cover as a disjoint finite cover.
\end{proof}

\begin{lemma}\label{lemma:uniform-approx}
Let $\pi_1$, $\pi_2$ be finite measures, and $f$ a measurable function on the product $\sigma$-algebra $\sigmaalg_1 \otimes \sigmaalg_2$. Then for all $\epsilon > 0$, there is a measurable set $B$ such that $\pi_1 \otimes \pi_2(B) < \epsilon$, and outside of which $f$ is uniformly approximated by simple functions on rectangles:
\begin{align*}
\lim_{M\to\infty} \sup_{x \notin B} &\left| f(x) - \sum_{m=1}^M c_{m}^M \1_{R_m^M}(x) \right| = 0,
\end{align*}
for some $R_m^M = F_m^M \times G_m^M$, $F_m^M \in \sigmaalg_1, G_m^M \in \sigmaalg_2$.
\end{lemma}

\begin{proof}
Assume $f \ge 0$ (if not, apply this argument to the positive and negative
parts of $f$)
then there exists simple functions $f_M$ that converge almost everywhere to $f$:
\begin{align*}
\lim_{M\to\infty} f_M &= f \;\;\;\textrm{(a.s.)} \\
f_M &= \sum_{m=1}^M c_m^M \1_{E_m^M}, \\
E_m^M &\in \sigmaalg_1 \otimes \sigmaalg_2
\end{align*}
Next, we apply Lemma~\ref{lemma:finite-approx} to each $E_m^M$, with $\mu = \pi_1 \otimes \pi_2$. We set an exponentially decreasing tolerance $\epsilon_m^M = \epsilon M^{-1} 2^{-M-1}$ so that:
\begin{align*}
\mu(A_m^M) &\le \mu(E_m^M) + \epsilon M^{-1} 2^{-M-1}, \\
A_m^M &= \bigcup_{j=1}^{n_m^M} R_{m,j}^M \supset E_m^M.
\end{align*}
Note that outside of some bad set $B_1$, we now have pointwise convergence of simple functions on rectangles (since each $A_m^M$ is a disjoint union of rectangles):
\begin{align} \label{eq:pointwise-conv}
\lim_{M\to\infty} &\sum_{m=1}^M c_m^M \1_{A_m^M} \Big|_{\compl{B_1}} = f\Big|_{\compl{B_1}} \;\;\;\textrm{(a.s.)} \\
B_1 &= \bigcup_{M=1}^\infty \bigcup_{m=1}^M (A_m^M \backslash E_m^M). \nonumber
\end{align}
Note that:
\begin{align*}
\pi_1 \otimes \pi_2(B_1) &\le  \sum_{M=1}^\infty \sum_{m=1}^M \epsilon M^{-1} 2^{-M-1} \\
&= \epsilon \sum_{M=1}^\infty 2^{-M-1} = \frac{\epsilon}{2}.
\end{align*}
Finally, pointwise convergence almost everywhere in Equation~(\ref{eq:pointwise-conv}) implies by Egorov's theorem the existence of a set $B_2$ with $\pi_1 \otimes \pi_2(B_2) < \epsilon/2$ and such that convergence is uniform outside of $B = B_1 \cup B_2$.
\end{proof}

\begin{lemma}\label{lemma:bounds}
Suppose $x < y + \beta_2$.
Then $\P(|X - x| > \beta_1) < \alpha$ implies that ${\P(X > \beta_1 + \beta_2 + y) < \alpha}$.
\end{lemma}

\begin{proof}
We prove that $\P(X > \beta_1 + \beta_2 + y) < \P(|X - x| > \beta_1)$ via the contrapositive on the events:
$|X - x| \le \beta_1 \Longrightarrow X - x \le \beta_1 \Longrightarrow X \le \beta_1 + \beta_2 + y.$
\end{proof}

\begin{proof}[Proof of Proposition~\ref{prop:consistent-appendix}]
Let us label the nodes of the tree using a \emph{height function} $h$ defined to be equal to zero at the leaves, and to $h(t) = 1 +
\max\{h(t_\tleft), h(t_\tright)\}$ for the internal nodes.

We proceed by induction on $h = 0, 1, 2, \dots$, showing that for all $t$ such
that $h(t) \le h$, the normalized weighted particles $(\x_{t,N}^i, \nw_{t,N}^i)$ obtained from \RMC($t$) are consistent with respect to $(\pi,
\setC_t)$.
The base case it trivially true. Suppose $t$ is one of the subtrees such that $h(t) = h$. Note that its two children $t_\tleft$ and $t_\tright$ are such that $h(t_c) < h(t)$, so the induction hypothesis implies these two children populations of weighted particles $(\x_{\tleft,N}^i, \nw_{\tleft,N}^i), (\x_{\tright,N}^i, \nw_{\tright,N}^i)$ are consistent. We now show that we adapt Theorem~1 of DM08 to establish that the weighted particles $(\x_{t,N}^i, \nw_{t,N}^i)$ are consistent as well.

Note that for each simple $f^M$ defined as in the proof of Lemma~\ref{lemma:uniform-approx}, we have:
\begin{align}\label{eq:conv-of-rectangles}
\sum_{i=1}^N \sum_{j=1}^N & \nw^i_{\tleft,N} \nw^j_{\tright,N} f^M(\x^i_{\tleft,N}, \x^i_{\tright,N}) \nonumber \\
&= \sum_{m=1}^{M} \left( \sum_{i=1}^N \w^i_{\tleft,N} \1_{A_{m}^M}(\x^i_{\tleft,N}) \right) \left( \sum_{j=1}^N \w^j_{\tright,N} \1_{B_{m}^M}(\x^j_{\tright,N}) \right) \nonumber \\
&\pcv \sum_{m=1}^{M} \left( \int \1_{A_{m}^M}(\x_\tleft) \pi_\tleft(\ud \x_\tleft) \right) \left( \int \1_{B_{m}^M}(\x_\tright)  \pi_\tright(\ud \x_\tright) \right) \nonumber \\
&= \iint f^M(\x_\tleft, \x_\tright) (\pi_\tleft \otimes \pi_\tright)(\ud \x_\tleft \times \ud \x_\tright).
\end{align}

Next, we show that this convergence in probability can be lifted from simple
$f^M$ to general bounded $\mathcal{F}_\tleft\otimes\mathcal{F}_\tright$-measurable functions. To shorten the notation, let:
\begin{align*}
\mu_A(f) &= \pi_\tleft \otimes \pi_\tright(\1_A f) \\
\mu^N_A(f) &= \sum_{i=1}^N \sum_{j=1}^N \nw^i_{\tleft,N} \nw^j_{\tright,N} \1_A f(\x^i_{\tleft,N}, \x^i_{\tright,N}),
\end{align*}(and $\mu(f), \mu^N(f)$ are defined similarly but without the indicator restriction).

Let $\epsilon, \delta > 0$ be given. Using the result of Equation~(\ref{eq:conv-of-rectangles}), first pick $\Nbound > 0$ such that for all $N \ge \Nbound$, $\P(|\mu^N_{\compl{B}}(f^M) - \mu_{\compl{B}}(f^M)| > \epsilon) < \delta/2$.
Second, using Lemma~\ref{lemma:uniform-approx} pick $B \in \sigmaalg_\tleft \otimes \sigmaalg_\tright$ and $M > 0$  such that $\sup_{\x\notin B}|f^M(\x) - f(\x)| < \epsilon/C$ and $\mu(B) < \epsilon/C$. This implies that both $|\mu_{\compl{B}}(f^M) - \mu_{\compl{B}}(f)|$ and $|\mu^N_{\compl{B}}(f^M) - \mu^N_{\compl{B}}(f)|$ are bounded above by $\epsilon$.
Third, using Lemma~\ref{lemma:finite-approx}, pick a cover $A$ of $B$, composed of a union of rectangles and such that $\mu(A) < \mu(B) + \epsilon/C$. Using Equation~(\ref{eq:conv-of-rectangles}) again, pick $\Nbound' > 0$ such that for all $N \ge \Nbound'$, $\P(|\mu^N(A) - \mu(A)| > \epsilon/C) < \delta/2$. Applying Lemma~\ref{lemma:bounds} with $X = \mu^N(A), \beta_1 = \beta_2 = \epsilon/C$, $x = \mu(A)$, $\alpha = \delta/2$ and $y = \mu(B)$, we get $\P(\mu^N(A) > 3\epsilon/C) < \delta/2$.

From these choices, we obtain that for all $N > \max\{\Nbound, \Nbound'\}$:
\begin{align*}
\P(|\mu^N(f) - \mu(f)| > 4\epsilon) &\le \P\left(\mu_B^N(f) > 4 \epsilon\right) + \P\left( D_1 + D_2 + D_3 + \mu_B(f) > 4 \epsilon \right),
\end{align*}where:
\begin{align*}
D_1 &= |\mu^N_{\compl{B}}(f) - \mu^N_{\compl{B}}(f^M)| \\
D_2 &= |\mu^N_{\compl{B}}(f^M) - \mu_{\compl{B}}(f^M) | \\
D_3 &= |\mu_{\compl{B}}(f^M) - \mu_{\compl{B}}(f)|.
\end{align*}Therefore:
\begin{align*}
\P(|\mu^N(f) - \mu(f)| > 4\epsilon)  &\le \P(\mu^N(B) > 4\epsilon/C) + \P(D_2 > \epsilon) \\
&\le  \P(\mu^N(A) > 3\epsilon/C) + \delta/2 \\
&\le \delta.
\end{align*}

Next, we use the fact that resampling is performed at every iteration, Condition~\ref{cond:resampling}, and Theorem~3 from DM08, to view resampling as reducing the $N^2$ particles into $N$ unweighted particles. We plug in the following quantities in their notation:
\begin{align*}
M_N &= N^2\\
\widetilde M_N &= N\\
\xi_{N,(i,j)} &= (\x^i_{\tleft,N}, \x^j_{\tright,N})\\
\omega_{N,(i,j)} &= \nw^i_{\tleft,N} \nw^j_{\tright,N}.
\end{align*}
This yields that the $N$ particles obtained from resampling $N$ time from the particle approximation
of $\pi_\tleft \times \pi_\tright$ creates a consistent approximation.
Finally, Theorem~1 from DM08 closes the induction argument.
\end{proof}


\section{\dcsmc algorithm}\label{app:algo}%
In Algorithm~\ref{alg:dcsmc-general} we provide pseudo-code for a more general \dcsmc algorithm than what is given in
Algorithm~\ref{alg:dcsir}, specifically, incorporating mixture
sampling (see Section~\ref{sec:extensions:mixturesampling}) and tempering (see Section~\ref{sec:tempering}).
\begin{algorithm}[ptb]\singlespace
  \caption{$\RMC(t)$ -- using mixture sampling and tempering}
  \label{alg:dcsmc-general}
  \begin{enumerate}
  \item
    \begin{enumerate}
    \item  For $c\in\C(t)$, $(\{\x_c^i, \w_c^i \}_{i=1}^\Np, \widehat Z_c^\Np ) \gets \RMC(c)$.
    \item For $i=1$ to $\Np$, draw $(\rx_{c_1}^i, \dots, \rx_{c_C}^i)$ from
      \begin{align*}
        Q_{t}(\rmd\x_{c_1}, \dots, \rmd\x_{c_C} ) =
        \sum_{i_1=1}^{\Np} \ldots \sum_{i_C=1}^{\Np}
        \frac{v_{t}(i_1,\ldots,i_C)  \delta_{(\x_{c_1}^{i_1},\ldots,\x_{c_C}^{i_C})} ( \rmd\x_{c_1}, \dots, \rmd\x_{c_C} )}{
          \sum_{j_1=1}^{\Np} \ldots \sum_{j_C=1}^{\Np}  v_{t}(j_1,\ldots,j_C)
        },
      \end{align*}
      where
      \begin{align*}
        v_{t}(i_1,\ldots,i_C) =  \left( \prod_{c\in\C(t)} \w_c^{i_c}
        \right)\check\pi_t(\x_{c_1}^{i_1},\ldots,\x_{c_C}^{i_C}) \bigg/ \prod_{c\in\C(t)}
        \pi_{c} (\x_c^{i_c}),
      \end{align*}
      and where $\check\pi_t$ is chosen by the user (\eg, according to \eqref{eq:extensions:alphastar}).
    \item Compute
      \begin{align*}
        \widehat Z_t^\Np = \left( \prod_{c\in\C(t)} \widehat Z_c^\Np  \left\{\frac{1}{\Np} \sum_{i=1}^\Np \w_{c}^i\right\}^{-1} \right)
\left\{ \frac{1}{\Np^C} \sum_{i_1=1}^{\Np} \ldots \sum_{i_C=1}^{\Np} v_{t}(i_1,\ldots,i_C) \right\}.
    \end{align*}
  \end{enumerate}
\item\begin{enumerate}
    \item For $i=1$ to $\Np$, if $\tilde \setX_t \neq \emptyset$, simulate $\widetilde\x_t^i \sim q_t(\cdot \mid \rx_{c_1}^i, \dots, \rx_{c_C}^i)$ where $(c_1, c_2, \dots, c_C) = \C(t)$; else $\widetilde\x_t^i \gets \emptyset$.
    \item For $i=1$ to $\Np$, set $\x_{t,0}^i = (\rx_{c_1}^i, \dots, \rx_{c_C}^i, \ix_t^i)$ and $\w_{t,0}^i = 1$.
    \item For SMC sampler iteration $\jj=1$ to $\nn_t$: (N.B. $\gamma_{t,0} = \check \pi_t q_t$.)
      \begin{enumerate}
      \item For $i=1$ to $\Np$, compute
        $ \w_{t,\jj}^i = \w_{t,\jj-1}^i \gamma_{t,\jj}(\x_{t,\jj-1}^i) / \gamma_{t,\jj-1}(\x_{t,\jj-1}^i) $.
      \item Optionally, resample $\{ \x_{t,\jj-1}^i, \w_{t,\jj}^i \}_{i=1}^\Np$:
        \begin{enumerate}
        \item[-] Update $\widehat Z_{t}^\Np \gets \widehat Z_t^\Np \times \{\frac{1}{\Np} \sum_{i=1}^\Np \w_{t,j}^i\}$.
        \item[-] Override the notation $\{ \x_{t,\jj-1}^i, \w_{t,\jj}^i \}_{i=1}^\Np$
        to refer to the resampled particle system.
        \end{enumerate}
      \item For $i=1$ to $\Np$, draw
        \( \x_{t,\jj}^i \sim K_{t,\jj}(\x_{t,\jj-1}^i ,\cdot) \) using a $\pi_{t,\jj}$-reversible Markov kernel $K_{t,\jj}$.
      \end{enumerate}
    \item For $i=1$ to $\Np$, set $\x_t^i = \x_{t,\nn_t}^i$ and $\w_t^i = \w_{t,\nn_t}^i$.
    \end{enumerate}
  \item Update $\widehat Z_t^\Np \gets \widehat Z_t^\Np \times \{\frac{1}{\Np} \sum_{i=1}^\Np \w_{t}^i\}$.
  \item Return $( \{ \x_t^i, \w_t^i \}_{i=1}^\Np, \widehat Z_t^\Np)$.
  \end{enumerate}
\end{algorithm}

\section{Supplement on the square-lattice \mrf models} \label{app:mrf}
\subsection{Additional numerical results for the Ising model}
In this section we provide some additional numerical results for the Ising model presented in Section~\ref{sec:expts:mrf} of
the main text. First, in Figures~\ref{fig:app:ising-extra2-logNC}--\ref{fig:app:ising-extra2-E} we repeat the results shown in Figure~\ref{fig:expts:ising-E},
but complemented with the estimates for \dcsmc[\sir] and \dcsmc (mix).

\begin{figure}[ptb]
  \centering
  \includegraphics[width=0.8\columnwidth]{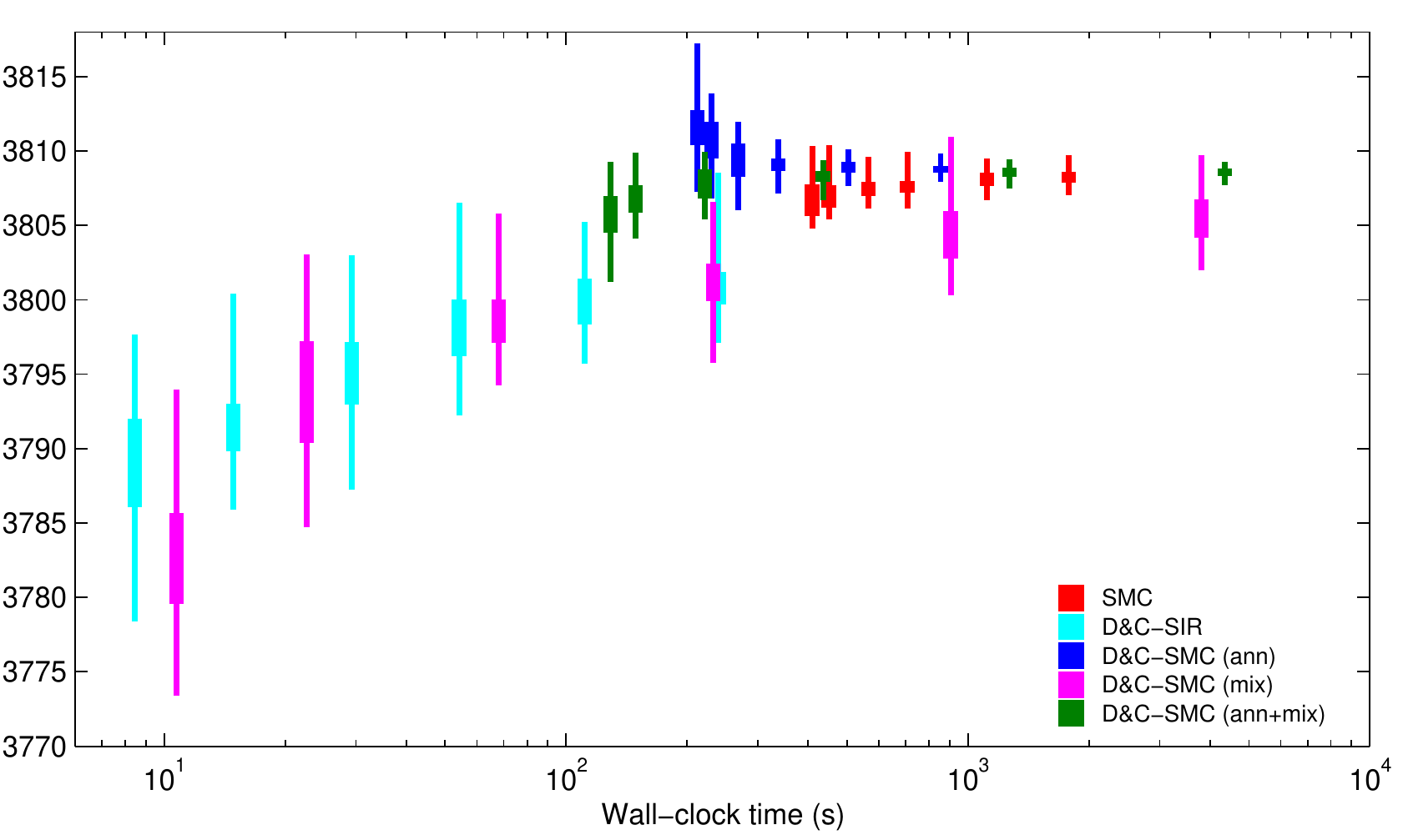}%
  \caption{Box-plots (min, max, and inter-quartile) of estimates of $\log Z$ for the Ising model with $\beta = 0.44$ over 50 runs of each sampler
    (excluding single flip MH which does not readily
    provide an estimate of $\log Z$). The boxes, as plotted from left to right, correspond to increasing number of particles $\Np$.}
  \label{fig:app:ising-extra2-logNC}
  \vspace{\baselineskip}
  \includegraphics[width=0.8\columnwidth]{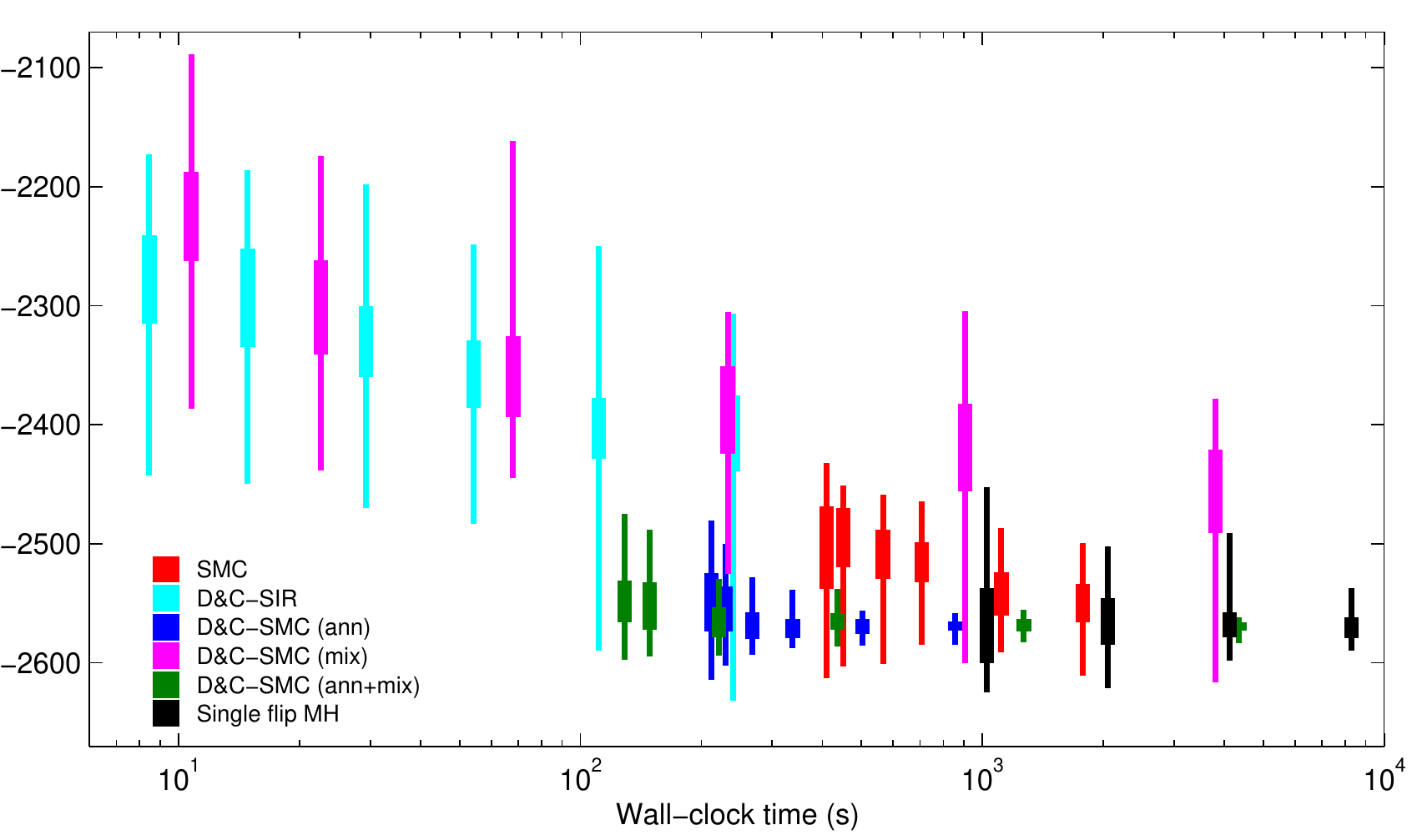}%
  \caption{Box-plots (min, max, and inter-quartile) of estimates of $\E[E(\x)]$ for the Ising model with $\beta = 0.44$ over 50 runs of each sampler.
    The boxes, as plotted from left to right, correspond to increasing number of particles $\Np$ (or number of MCMC iterations for single flip MH).}
  \label{fig:app:ising-extra2-E}
\end{figure}

Second, we have repeated the numerical evaluation described in Section~\ref{sec:expts:mrf} of
the main text for different inverse temperatures of the Ising model: $\beta = 0.40$ (90\% of critical inverse temperature) and $\beta = 0.48$
(110\% of critical inverse temperature). The results are given in Figures~\ref{fig:app:ising-extra1-logNC}--\ref{fig:app:ising-extra1-E}
and Figures~\ref{fig:app:ising-extra3-logNC}--\ref{fig:app:ising-extra3-E}, respectively.

\begin{figure}[ptb]
  \centering
  \includegraphics[width=0.8\columnwidth]{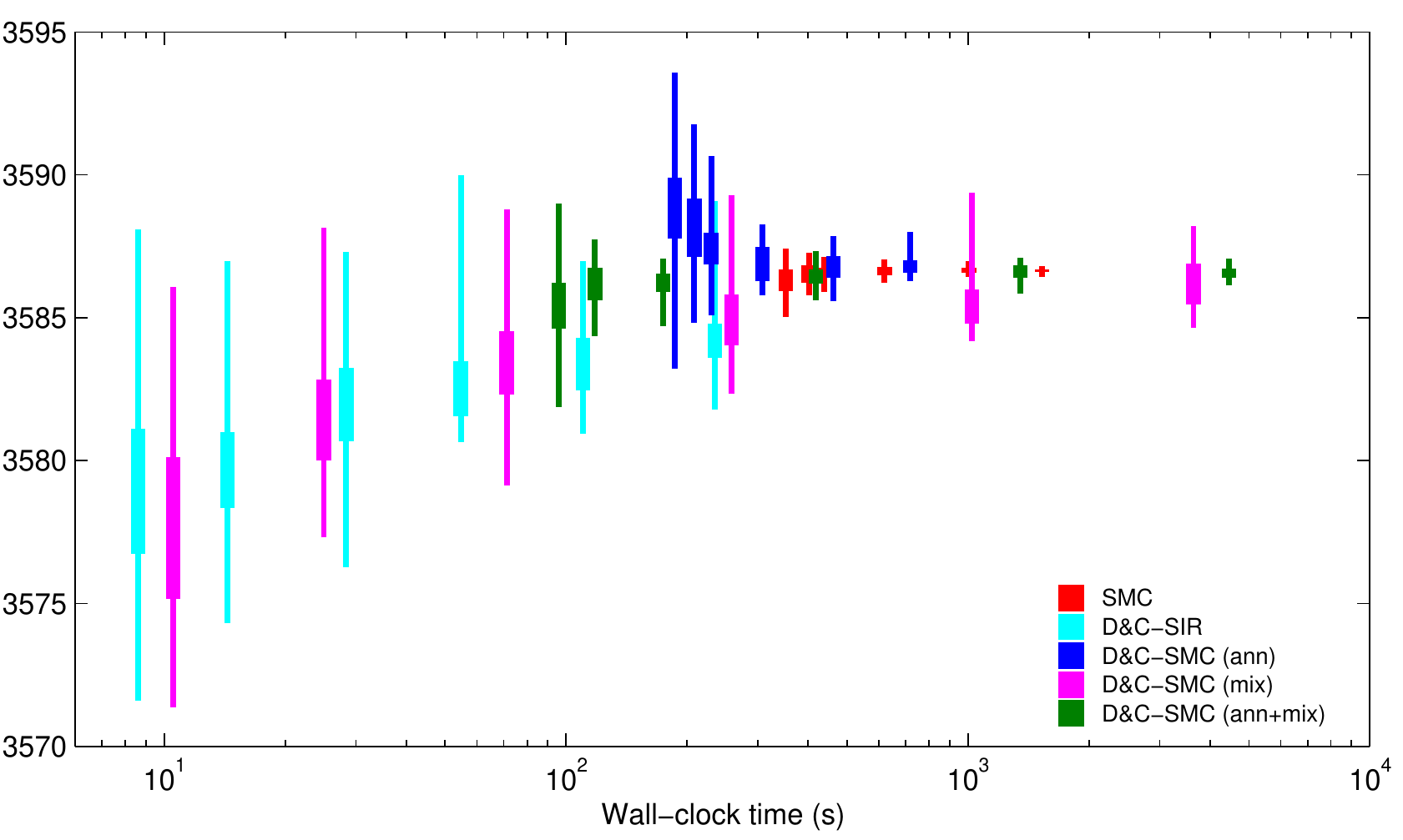}%
  \caption{Box-plots (min, max, and inter-quartile) of estimates of $\log Z$ for the Ising model with $\beta = 0.40$ over 50 runs of each sampler
    (excluding single flip MH which does not readily
    provide an estimate of $\log Z$). The boxes, as plotted from left to right, correspond to increasing number of particles $\Np$.
    }
  \label{fig:app:ising-extra1-logNC}
  \vspace{\baselineskip}
  \includegraphics[width=0.8\columnwidth]{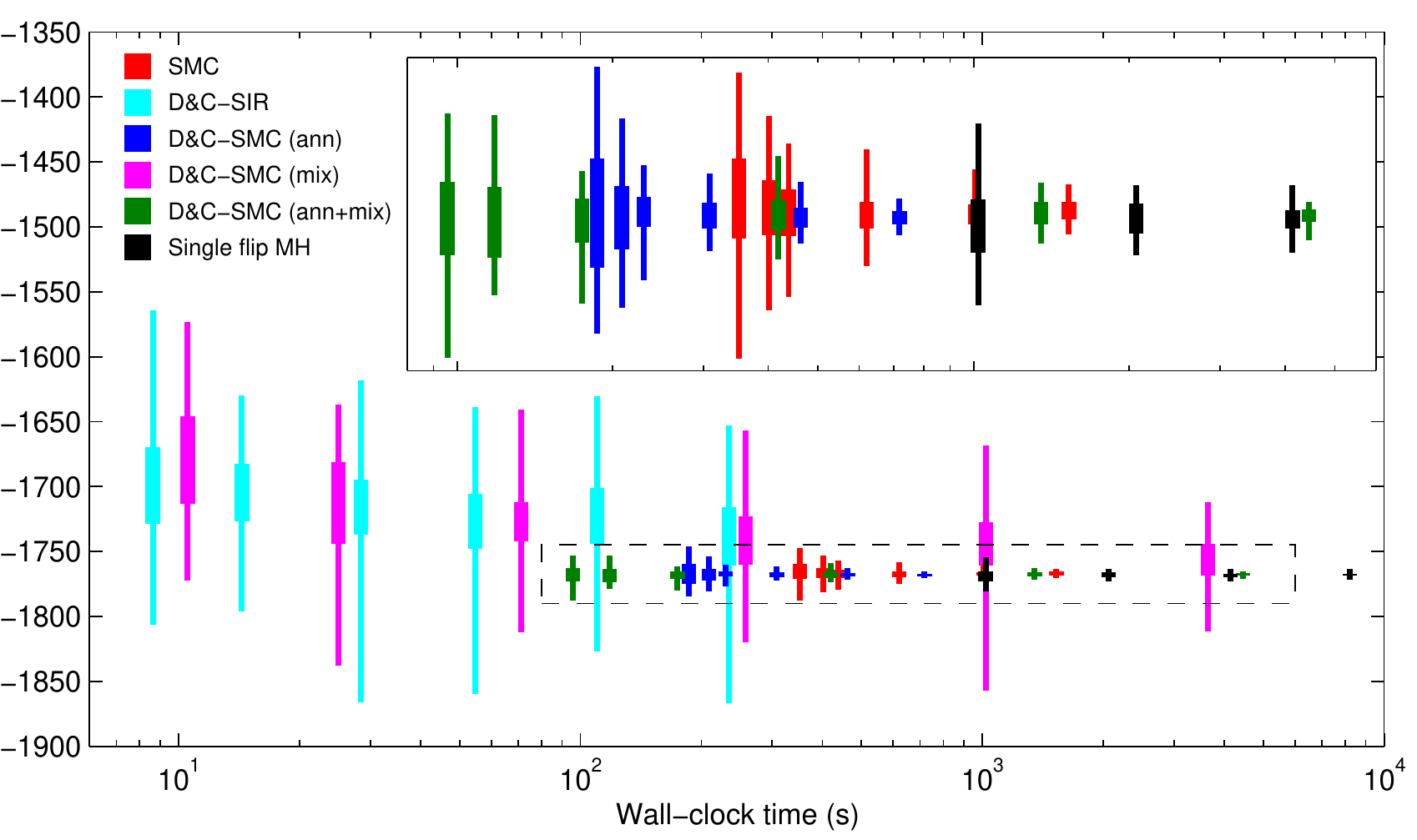}%
  \caption{Box-plots (min, max, and inter-quartile) of estimates of $\E[E(\x)]$ for the Ising model with $\beta = 0.40$ over 50 runs of each sampler.
    The boxes, as plotted from left to right, correspond to increasing number of particles $\Np$ (or number of MCMC iterations for single flip MH).
    The overlaid axes is a zoom-in on the dashed region.}
  \label{fig:app:ising-extra1-E}
\end{figure}

\begin{figure}[ptb]
  \centering
  \includegraphics[width=0.8\columnwidth]{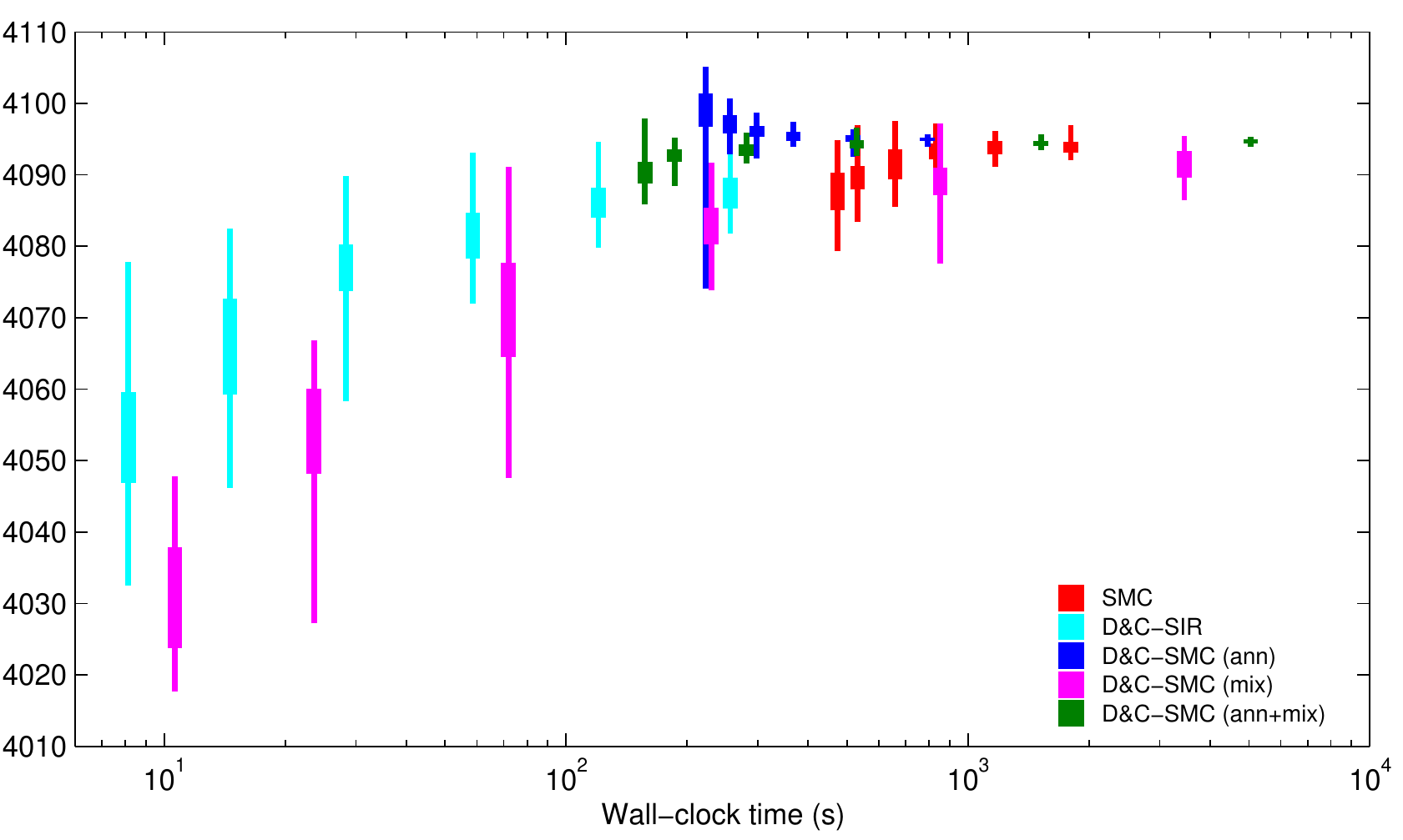}%
  \caption{Box-plots (min, max, and inter-quartile) of estimates of $\log Z$ for the Ising model with $\beta = 0.48$ over 50 runs of each sampler
    (excluding single flip MH which does not readily
    provide an estimate of $\log Z$). The boxes, as plotted from left to right, correspond to increasing number of particles $\Np$.
    }
  \label{fig:app:ising-extra3-logNC}
  \vspace{\baselineskip}
  \includegraphics[width=0.8\columnwidth]{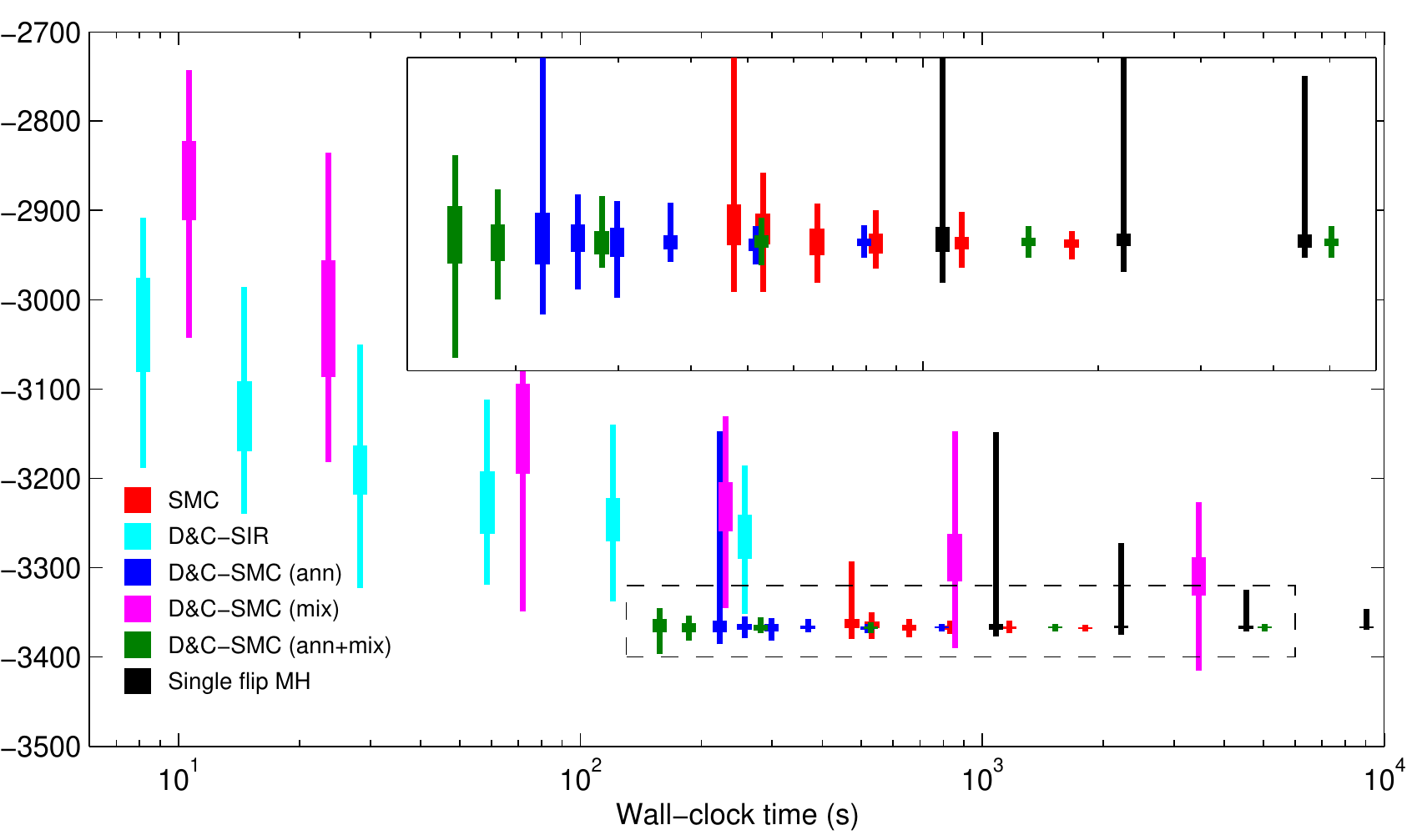}%
  \caption{Box-plots (min, max, and inter-quartile) of estimates of $\E[E(\x)]$ for the Ising model with $\beta = 0.48$ over 50 runs of each sampler.
    The boxes, as plotted from left to right, correspond to increasing number of particles $\Np$ (or number of MCMC iterations for single flip MH).
    The overlaid axes is a zoom-in on the dashed region.}
  \label{fig:app:ising-extra3-E}
\end{figure}


\subsection{Square-lattice \mrf with continuous latent variables}
In this section we evaluate the proposed \dcsmc method on a second square-lattice \mrf example.
We consider a toy-model consisting of a latent Gaussian field $\z \in \reals^{M\times M}$ ($M=32$),
with nearest-neighbour interactions and periodic boundary conditions:
\begin{align}
  \label{app:mrf2:prior}
 p(\z) \propto \exp\left( -\frac{1}{2}\left \{ \sum_{(k,\ell) \in \mathcal{E}} \lambda_{1} (x_k - x_{\ell})^2 + \lambda_2 \sum_{k \in \mathcal{V}} x_k^2 \right\} \right),
\end{align}
with $\mathcal{E}$ and $\mathcal{V}$ being the edge set and vertex set of the \mrf, respectively. We set $\lambda_1 = 10$
and $\lambda_2 = 0.01$ for the interaction strength and node potential, respectively.
To each node of the \mrf we also attach an observed variable $y_k$, conditionally distributed according to
$p(y_{k} \mid x_{k}) = \N(y_{k} \mid x_{k}^2, 0.05^2)$, $k\in\mathcal{V}$. The target distribution
for the samplers is then given by the Bayesian posterior distribution $p(\z \mid \mathbf{y})$, where $\mathbf{y} = \{y_k\}_{k\in \mathcal{V}}$.

We simulate a batch of data from the model and apply the same inference methods as used in Section~\ref{sec:expts:mrf} to sample
from the posterior distribution.
We initialise all methods by sampling independently for each $k\in\mathcal{V}$ from a distribution given by,
\begin{align}
  \label{app:mrf2:initialisation}
  \mu_k(x_k) \propto  N(y_{k} \mid x_{k}^2, 0.05^2) \exp\left(-\frac{\lambda_2}{2} x_k^2 \right),
\end{align}
which we compute to high precision by (one-dimensional) adaptive quadrature. This allows us to focus the evaluation on the
difficulties arising from the interactions between neighbouring sites, since \eqref{app:mrf2:initialisation}
is exactly the posterior distribution of $x_k$ if we neglect all interactions.
Note that we observe the square of the latent variables $x_k$ (plus noise).
Consequently, the distributions \eqref{app:mrf2:initialisation}, as well as the marginal posteriors $p(x_k \mid \mathbf{y})$,
tend to be bimodal whenever $|x_{k}|$ is large (relative to the observation variance). This poses significant difficulties for
the single site MH sampler and the standard \smc sampler, as we shall see below.

All algorithmic settings are the same as in Section~\ref{sec:expts:mrf}, with the exception that the MCMC kernel used for the annealed methods (and for the single site MH sampler)
uses a Gaussian random walk proposal with standard deviation 0.132 (chosen by manual tuning; we obtained an average acceptance rate of 0.6--0.7 for all methods).
It is interesting to note that we only needed to make small modification to the code used for the Ising model, changing only the initialisation procedure, the energy function, and the MCMC kernel.

Figures~\ref{fig:expts:mrf2-logNC} and \ref{fig:expts:mrf2-E} show results on the estimates of the log-normalising-constant
and the expected energy for the posterior distribution, obtained from 50 runs of each algorithm on a single data set.
The results for the single site random-walk MH sampler are not shown since the method fails completely to converge to samples
from the correct posterior distribution (inter-quartile range in the estimates of the expected energy over the 50 independent runs of the MCMC sampler
is more than \thsnd{7}; \cf, Figure~\ref{fig:expts:mrf2-E}).

We also note a superior performance of all \dcsmc samplers compared to \smc,
which is attributed to the fact that the \dcsmc samplers are better able to handle the multimodality of the
initial distributions than the ``standard'' \smc approach considered:
in standard \smc we initialise a single particle population by simulating $\Np$ times from the product measure
$\kronecker_{k\in\mathcal{V}} \mu_k(\rmd x_k)$, where $\mu_k$ is given by \eqref{app:mrf2:initialisation}.
Since many of the $\mu_k$'s are expected to be multimodal, this will likely result in
particles that have very low posterior probability under the ``correct model'' (\ie, with the interactions taken into account).
Indeed, any neighbouring pair $(x_k, x_\ell)$ where $x_k$ is drawn from the ``positive mode'' and $x_\ell$ is drawn from the
``negative mode'' will have low probability under the prior \eqref{app:mrf2:prior}.
Furthermore, since we use local MCMC moves to update the particle positions, the method is not able
to correct for this in a sufficiently efficient manner during the annealing process. The \dcsmc samplers
are much better suited to handling this difficulty, since they make use of \emph{multiple} particle populations at initialisation,
one for each site. The interactions are thereafter gradually taken into
account in the merge steps of the algorithm. This improvement comes without
any application-specific implementation effort.

\begin{figure}[ptb]
  \centering
  \includegraphics[width=0.85\columnwidth]{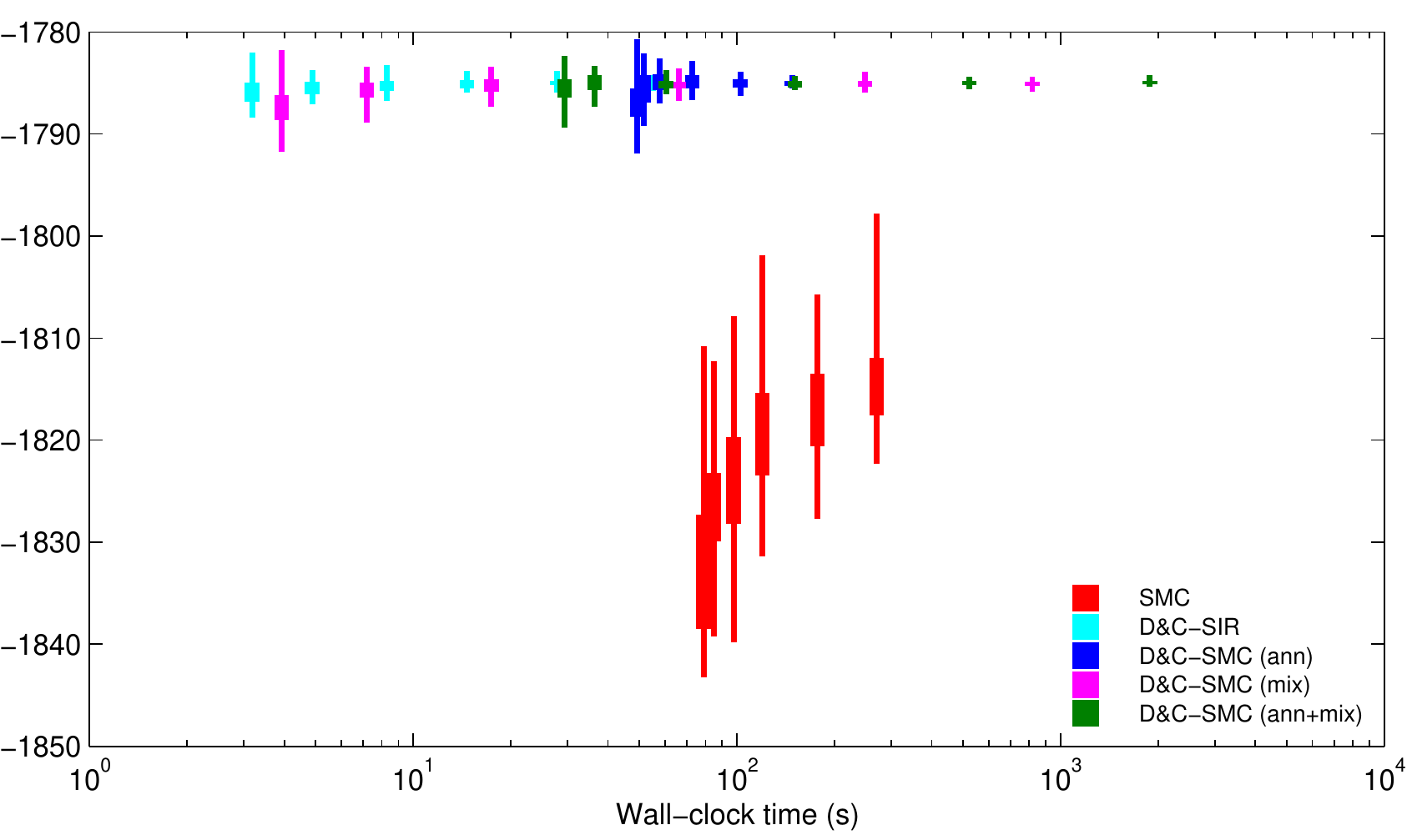}%
  \caption{{\footnotesize Box-plots of estimates of $\log Z$ over 50 runs of each sampler.
      The boxes, as plotted from left to right, correspond to increasing number of particles $\Np = 2^6$ to $2^{11}$ ($\Np = 2^{10}$ to $2^{15}$ for \dcsmc[\sir]).
      The overlaid axes is a zoom-in on the dashed region.}}
  \label{fig:expts:mrf2-logNC}
  \vspace{\baselineskip}
  \includegraphics[width=0.85\columnwidth]{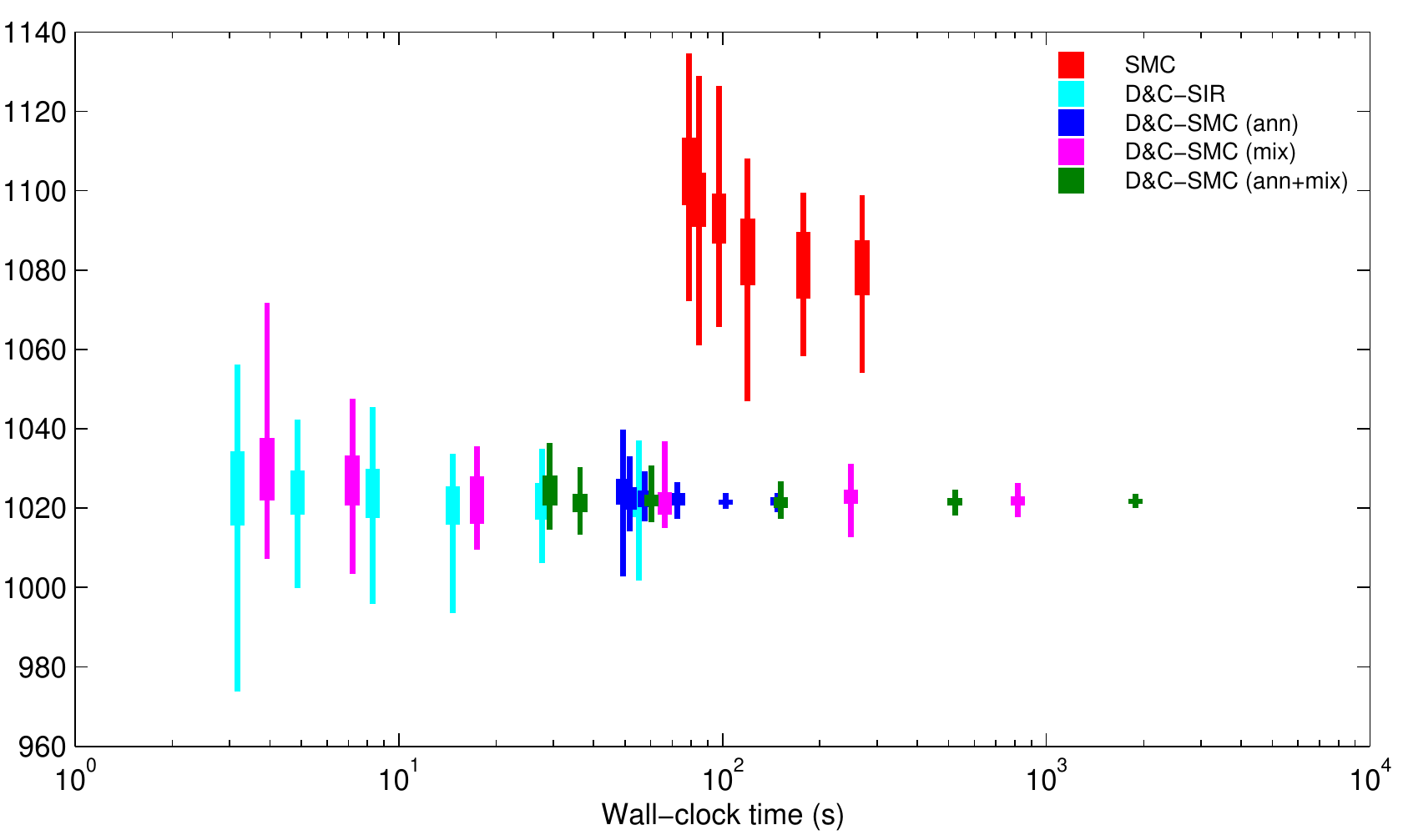}%
  \caption{{\footnotesize Box-plots of estimates of $\E[E(\x)]$ over 50 runs of each sampler.
      The boxes, as plotted from left to right, correspond to increasing number of particles $\Np = 2^6$ to $2^{11}$ ($\Np = 2^{10}$ to $2^{15}$ for \dcsmc[\sir]).
      The overlaid axes is a zoom-in on the dashed region.}}
  \label{fig:expts:mrf2-E}
\end{figure}

To further illustrate the ability of \dcsmc to capture the multimodality of the marginal posteriors,
we show in Figure~\ref{fig:app:mrf2-cdf} the empirical cumulative distribution functions for the samplers
for four sites with increasing $|x_{k}|$.
The results are for $\Np = 2\thinspace048$ ($\Np = 32\thinspace768$ for \dcsmc[\sir]) and each line correspond
to one of the 50 independent runs. It is clear that the \dcsmc samplers maintains the multimodality
of the posterior much better than standard \smc.

\begin{figure}[ptb]
  \centering
  \includegraphics[width=0.85\columnwidth]{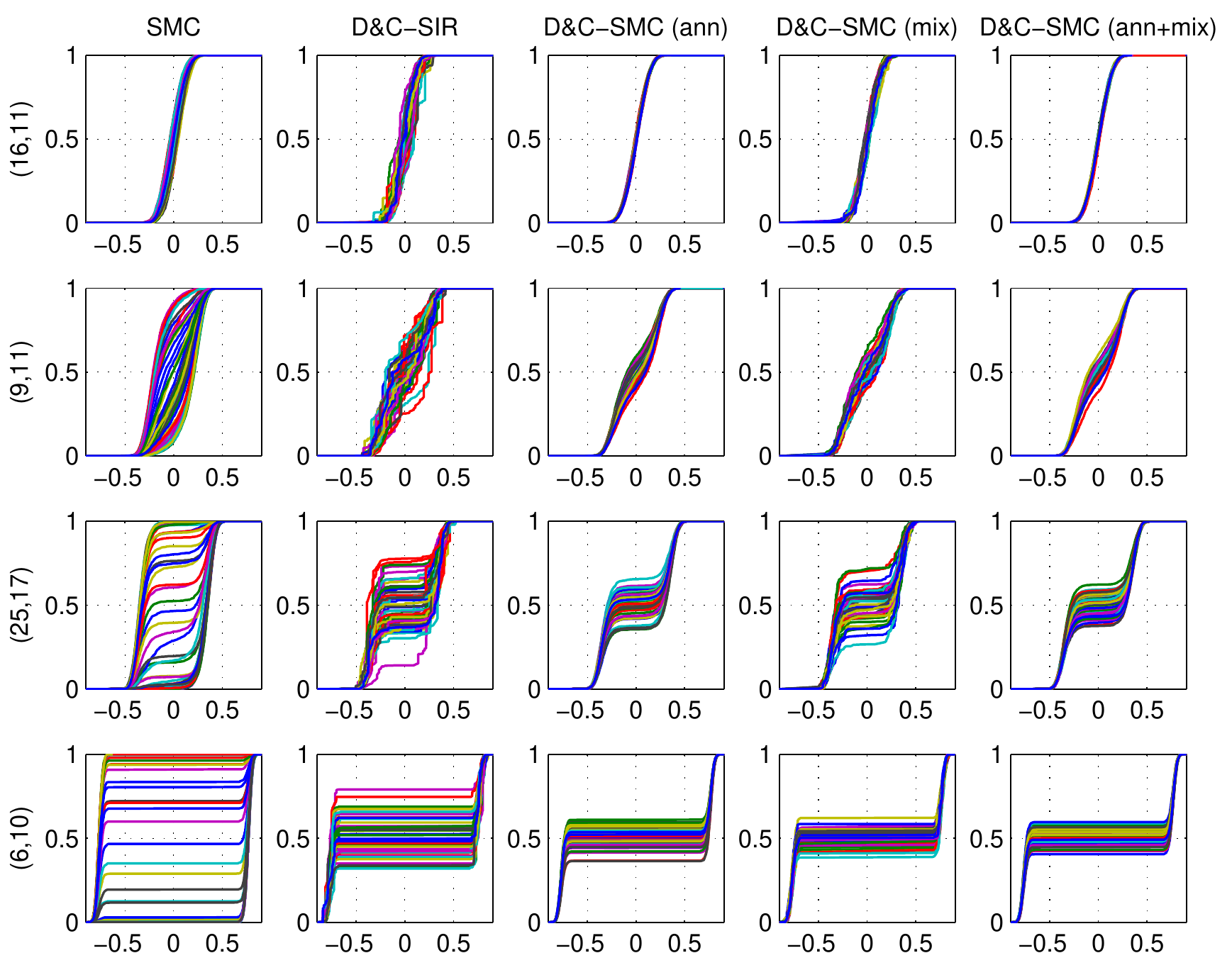}%
  \caption{Empirical cumulative distribution functions for the marginal posteriors $p(x_k \mid \mathbf{y})$ for four sites with increasing $|x_k|$ (from top to bottom). Each line correspond to one of the 50 independent runs of each algorithm.}
  \label{fig:app:mrf2-cdf}
\end{figure}

\section{Supplement on the hierarchical Bayesian model data analysis} \label{app:expts:dataAnalysis}

\subsection{Data pre-processing}


The data was downloaded from {\url{https://data.cityofnewyork.us/Education/NYS-Math-Test-Results-By-Grade-2006-2011-School-Le/jufi-gzgp}} on May 26, 2014. It contains New York City Results on the \emph{New York State Mathematics Tests}, Grades 3-8. We used data from grade 3 only.

Data is available for the years 2006--2011. We removed years 2010 and 2011, since the documentation attached to the above URL indicates that: ``Starting in 2010, NYSED changed the scale score required to meet each of the proficiency levels, increasing the number of questions students needed to answer correctly to meet proficiency.'' 

Each row in the dataset contains a school code, a year, a grade, the number of students tested, and summary statistics on their grades.
We use the last column of these summary statistics, which corresponds to the number of students that obtained a score higher than a fixed threshold. 

Moreover, for each school code, we were able to extract its school district. 
We removed the data from the schools in School District 75. This is motivated by the specialized character of School District 75: ``District 75 provides citywide educational, vocational, and behavior support programs for  students who are on the autism spectrum, have significant cognitive delays, are severely emotionally challenged, sensory impaired and/or multiply disabled.'' (\url{http://schools.nyc.gov/Academics/SpecialEducation/D75/AboutD75/default.htm})

For each school district we can also extract its county, one of Manhattan, Bronx, Kings, Queens, Richmond (note that some of these correspond to NYC boroughs with the same name, while Kings corresponds to Brooklyn; Richmond, to Staten Island; and Bronx, to The Bronx). The pre-processing steps can be reproduced using the script \texttt{scripts/prepare-data.sh} in the repository.

\subsection{Additional information on the experimental setup}

We checked correctness of the implementations by verifying on small examples that the posterior distributions obtained from all four methods (the three baselines and ours) become arbitrarily close when increasing the number of particles or MCMC iterations. We also subjected our Gibbs implementation to the Prior/Prior-Posterior correctness test of \citet{Geweke2004}. 

The precise command line arguments used in the experiments as well as the scripts used for creating the plots can be found at {\url{https://github.com/alexandrebouchard/multilevelSMC-experiments}}. To describe in more detail the methods we compared to, we introduce the following notation: given a permutation $t_1, t_2, \dots$ of the index set $T$, we set $\Xi_j$ to $\theta_{t_j}$ if $t_j$ is a leaf, or to $\sigma^2_{t_j}$ otherwise.
The three baseline methods are:

\begin{description}
  \item[Gibbs:] A Metropolis-within-Gibbs algorithm, proposing a change on a single variable, using a normal proposal of unit variance. As with \dcsmc[\sir], we marginalize the internal $\theta_t$ parameters. More precisely, we first sample a permutation $t_1, t_2, \dots$ uniformly at random, then using the Metropolis-within-Gibbs kernel, we sample $\Xi_1$, followed by $\Xi_2$, and so, until we sample $\Xi_{|T|}$, at which point we sample an indendent permutation uniformly at random and restart the process.  The Java implementation is available at {\url{https://github.com/alexandrebouchard/multilevelSMC}}, in the package ``multilevel.mcmc''.  We use a burn-in of 10\% 
and collect sample each time we resample a permutation of the nodes.
  \item[Stan:] An open-source implementation of  the Hamiltonian Monte Carlo algorithm \citep{Neal:2011}, more precisely of the No U-Turn Sampler \citep{Hoffman2012NUTS}, which adaptively selects the number of ``leap frog'' steps. Stan generates efficient C++ code. In contrast to the other methods, we did not implement marginalization of the internal $\theta_t$ parameters. Stan includes a Kalman inference engine, however it is limited to chain-shaped graphical models as of version 2.6.0. We use the default setting for the burn-in and adaptation period (50\%)
and a thinning of 10.\footnote{In other words, we store one sample every 10 accept-reject rounds. Since a large number of Hamilton Monte Carlo iterations was needed, this was useful to decrease storage. We verified using an ACF plot that this did not significantly penalize this method.} 
  \item[STD:] A standard (single population) bootstrap filter with the intermediate distributions being sub-forests incrementally built in post-order. More precisely, let $t_1, t_2, \dots$ denote a fixed post-order traversal of the nodes in the set $T$. Note that a prefix of this traversal, $t_1, t_2, \dots, t_k$, corresponds to a forest, with a corresponding parameter vector $\Xi_{1:k} = (\Xi_1, \Xi_2, \dots, \Xi_k)$. We use a standard SIR algorithm to sample $\Xi_{1:1}, \Xi_{1:2}, \Xi_{1:3}, \dots$, using the canonical sequence of intermediate distributions. Again, as with \dcsmc[\sir], we marginalize the internal $\theta_t$ parameters.  To propose an additional parameter $\Xi_i$ given its children, we use the same proposal distributions as those we used for \dcsmc[\sir] (described in Section~\ref{sec:hierarchical} of the main text). The Java implementation is available at {\url{https://github.com/alexandrebouchard/multilevelSMC}}, in the package ``multilevel.smc''.
\end{description}

\subsection{Additional results (serial)}

\begin{figure}[p]
\begin{center}
  \includegraphics[width=6.5in]{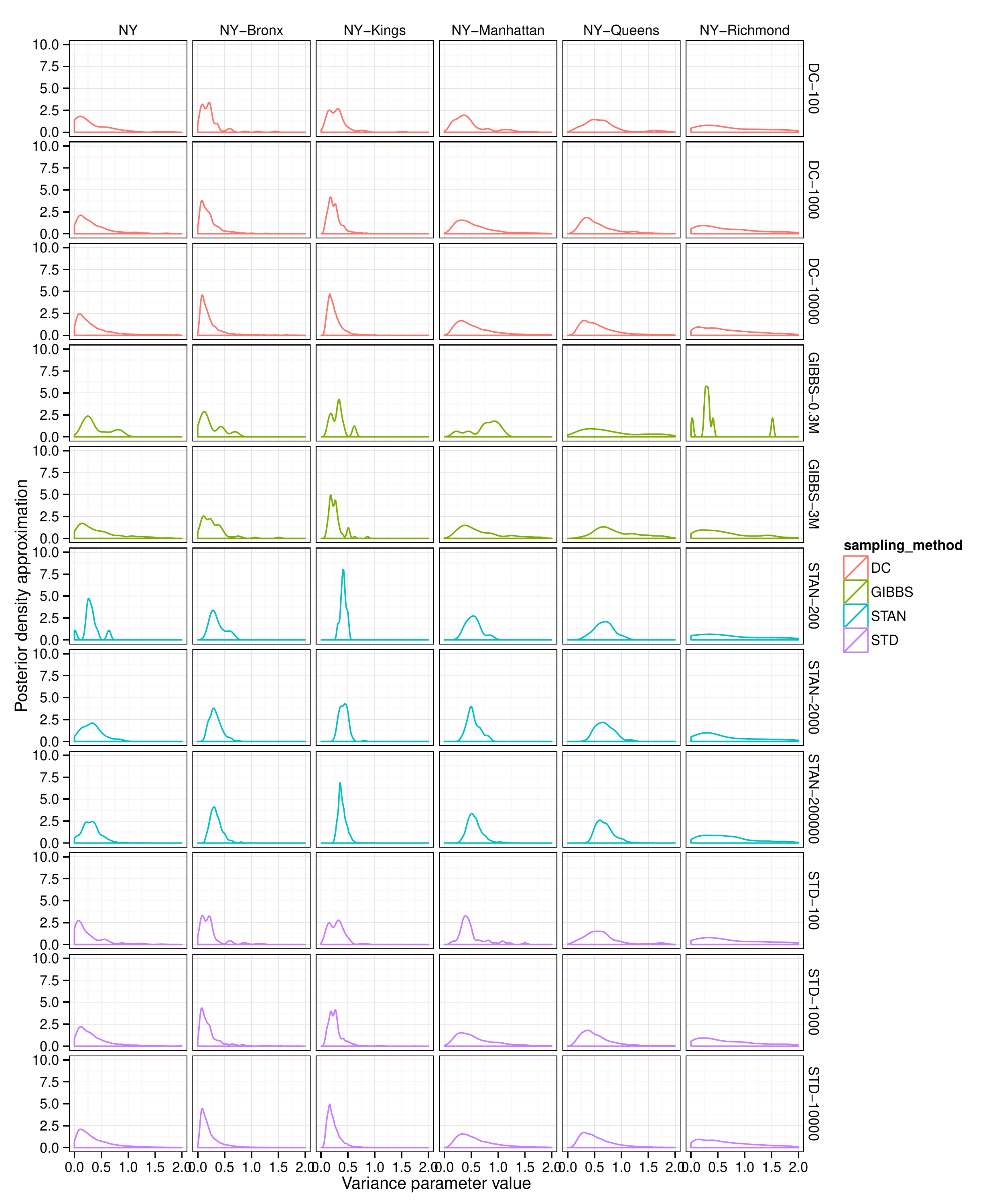}
  \caption{Posterior densities for the parameters $\sigma_t^2$. The rows index different methods and numbers of particles/iterations. The columns index different parameters $t$ (only the first two levels shown).}
  \label{fig:densities-var}
\end{center}
\end{figure}

We show in Figure~\ref{fig:densities-var} the posterior densities of the parameters $\sigma_t^2$, for $t$ ranging in the nodes in the two top levels of the tree, and for the four methods with different numbers of particles or MCMC iterations. The plots support that 1\,000 particles are sufficient for \dcsmc to output a reasonable approximation (see rows DC-$N$), while a larger number of MCMC iterations seem required for Gibbs (GIBBS-$N$) or Stan (STAN-$N$). 
Note however that DC-$N$ and STD-$N$ performed similarly in this first set of results.

Next, we show in Figure~\ref{fig:timing} a comparison of the wall clock times for the different algorithm configurations.
These experiments were performed on Intel Xeon E5430 quad-core processors, running at 2.66 GHz.
While the detailed rankings of the methods are implementation and architecture dependent, the results show that the DC-1\,000 was computed in approximately one minute, while all the reasonably mixed MCMC methods required at least one order of magnitude more time.
For example, running 2\,000 Stan iterations (1\,000 burn-in + 1\,000 samples) took around 6 hours. All Stan running times exclude the one-time model compilation. 

The results for both SMC and MCMC are all performed on a single thread. 
We attribute the high computational efficiency of the SMC methods to the favorable memory locality of the particle arrays. 
This is supported by the non-linearity of the running time going from 1000 to 10\thinspace000 particles (the theoretical running time is linear in the number of particles).

Note that for any given particle budget, our implementation of \dcsmc was slightly faster than STD. 
This could be caused by implementation details related to the fact that the particle datastructure for \dcsmc are simpler than those required by STD (the particles being forests in STD, instead of trees for \dcsmc).

\begin{figure}[tp]
\begin{center}
  \includegraphics[width=6.5in]{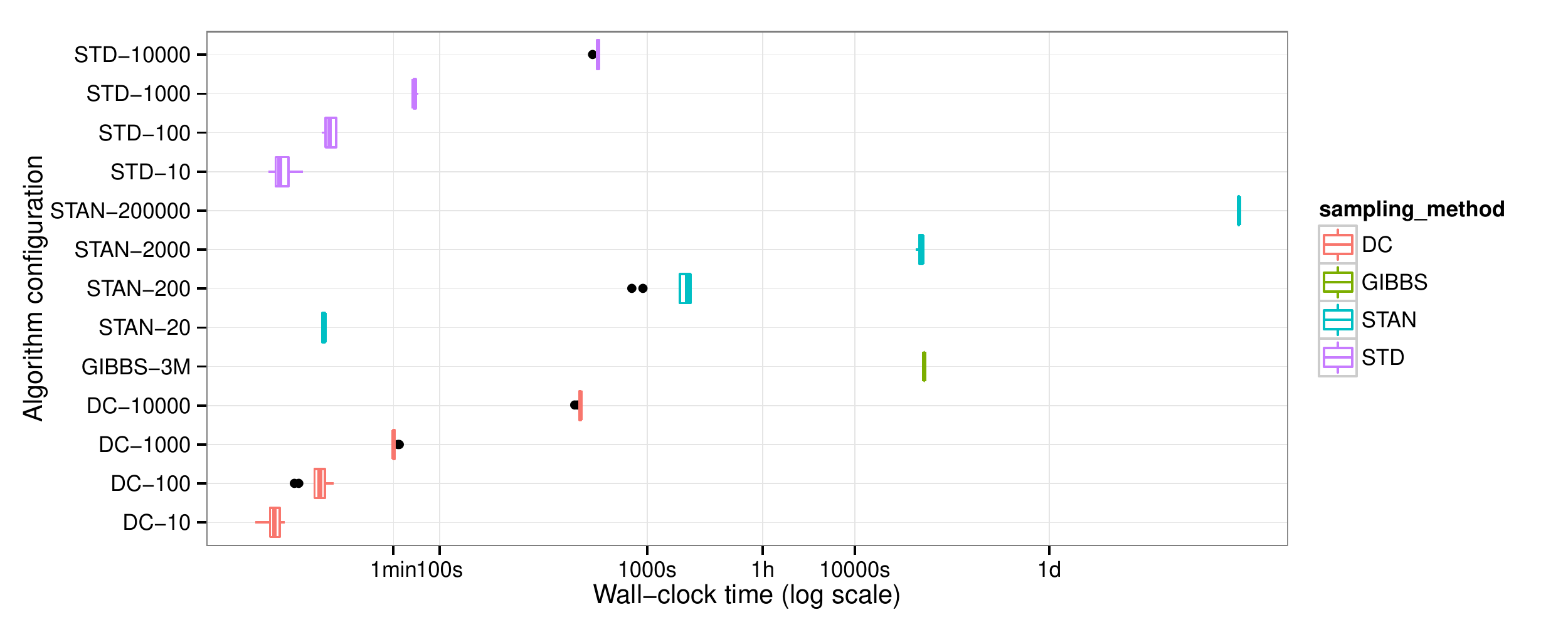}
  \caption{Comparison of the wall clock times for the different algorithm configurations. Note that the time axis is in log scale, and that iterations/particles were incremented exponentially. All experiments were replicated 10 times (varying only the Monte Carlo random seed), with the exception of the largest experiments taking more than 6 hours, which we ran only once (GIBBS-3M and STAN-200000).}
  \label{fig:timing}
\end{center}
\end{figure} 

These results further demonstrate that the model and data used in this section yield a challenging posterior inference problem.   

\subsection{Additional results (distributed)}

In this section, we provide a proof-of-concept demonstrating the suitability of \dcsmc to distributed computing environments (supplementing the discussion in Section~\ref{sec:distributed} of the main text).

In this discussion, it is useful to emphasize the usual distinction between \emph{parallel} architectures, where a single computer is composed of several cores sharing fast access to memory, and \emph{distributed} architectures, where several computers (nodes) are connected by a relatively slower network. 

Parallelizing SMC is typically done in the proposal step, at the granularity of particles: when sampling from $q$ or reweighting is expensive, or both, it becomes advantageous to parallelize these operations. However, the resampling step requires frequent communication, making this strategy less attractive in a distributed setting (but see for example  \citet{BolicDH05,Jun2012Entangled,VergeDDM:2014,Lee:2014,WhiteleyAD:2015} for alternatives).    

In contrast, distributed computation in \dcsmc can be scheduled to incur low communication costs. The main idea is to split the work at the granularity of populations, instead of the more standard particle granularity. 
In the following, we describe a simple strategy to decrease communication costs when the \dcsmc recursion is computed in parallel at several locations of the tree $t \in T$. To simplify the exposition, we assume that the tree is binary. 

To describe our distribution strategy, we introduce the notion of the \emph{depth} of a vertex in the tree $t\in T$. We set the depth to be zero at the root $r$, $\textrm{depth}(r) = 0$, and recursively, for each vertex $v$ with parent $v'$, $\textrm{depth}(v) = \textrm{depth}(v') + 1$. We consider the largest depth $d$ such that the number of vertices at that depth is greater or equal to the number $c$ of compute nodes available. For each vertex $v$ such that $\textrm{depth}(v) = d$, we assign the computation of all the \dcsmc recursions for the vertices under $v$ to the same computer.  

With this architecture, communication is only required at edges of the tree $T$ connecting nodes above depth $d$. 
This implies that the number of particles transmitted over the network will be $O(c N)$.  In contrast, naively distributing using a  particle granularity would require $O(c |T| N)$ particle transmissions. Moreover, the population granularity strategy can be  done in a completely decentralized fashion.

Based on this method, we have implemented an open source library for performing distributed, decentralized \dcsmc computation. Thanks to its decentralized nature, the package is simple to use: nodes discover each other automatically and only require a rough estimate on the number of nodes available, to specify the value $d$ discussed above. However, the algorithm is fault tolerant, in the sense that if some of the compute nodes fail, the computation will still be completed.

We applied this algorithm to the New York Mathematics Test dataset. We varied the number of active nodes in the cluster in $\{1, 2, 4, 8, 16, 32\}$, using $100\,000$ particles in all cases. The output of the distributed algorithm is identical in all cases, so it suffices to compare wall-clock times. The compute nodes consist in Xeon X5650 2.66GHz processors connected by a non-blocking Infiniband 4X QDR network. We show the results in Figure~\ref{fig:distributed}.

Each node used a single thread for simplicity, but note that combining parallelism and distribution can be done by parallelizing either at the level of particles, or populations with depths greater than $d$.

\begin{figure}[t]
\begin{center}
  \includegraphics[width=2in]{times.pdf} 
  \includegraphics[width=2in]{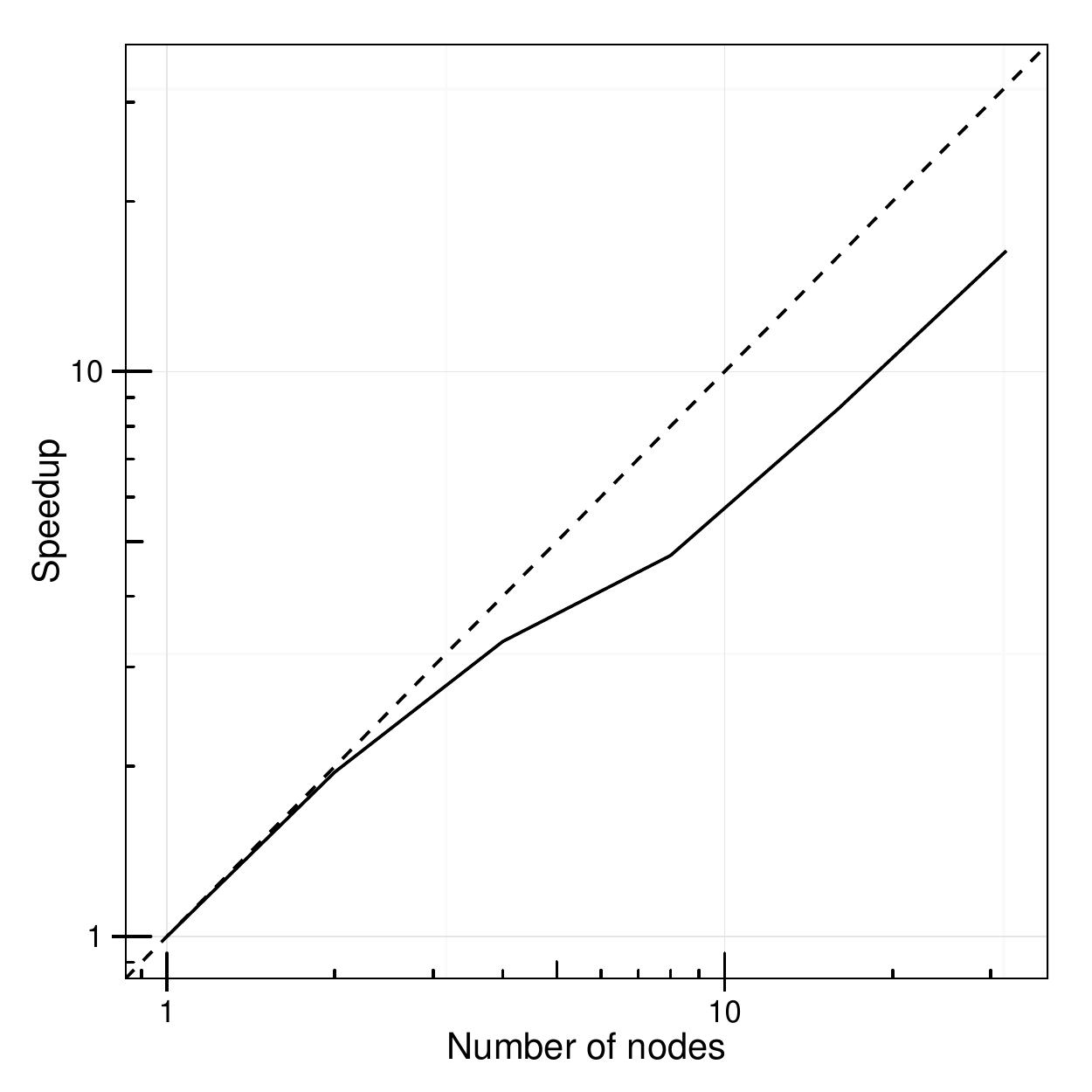} 
  \caption{(a) Wall-clock times for the distributed \dcsmc algorithm. (b) Speedup for the distributed \dcsmc algorithm (log scale).}
  \label{fig:distributed}
\end{center}
\end{figure}



\end{document}